\documentclass[letterpaper,11pt]{article}
\usepackage[margin=1in]{geometry}
\usepackage[OT1]{fontenc}

\usepackage{amsmath,amssymb,amsthm,mathtools,bm}
\usepackage{booktabs,array,tabularx,enumitem}
\usepackage[square,sort,comma,numbers]{natbib}
\usepackage{xcolor,tikz,float}
\usepackage{pgfplots}
\pgfplotsset{compat=1.18}
\usetikzlibrary{arrows.meta,positioning,calc,fit,backgrounds}
\definecolor{citationgreen}{rgb}{.42,.56,.14}
\definecolor{inkblue}{RGB}{35,65,94}
\definecolor{softblue}{RGB}{240,244,248}
\definecolor{softgreen}{RGB}{241,246,243}
\usepackage{amsmath,amssymb,amsthm,mathtools,bm}
\usepackage{booktabs,array,tabularx,enumitem}
\usepackage[square,sort,comma,numbers]{natbib}
\usepackage{xcolor}
\definecolor{citationgreen}{rgb}{.42,.56,.14}
\usepackage[pdfusetitle,pagebackref]{hyperref}
\hypersetup{
  colorlinks=true,linkcolor=blue,citecolor=citationgreen,
  urlcolor=blue,filecolor=blue,linktoc=section,
  bookmarksnumbered=true,bookmarksopen=true,bookmarksopenlevel=1,
  pdfsubject={Annotated framework for a merged paper on unitary synthesis},
  pdfkeywords={unitary synthesis, coordinate encoding, circuit depth,
    clean ancillas, adaptive quantum circuits}}
\setlist[itemize]{itemsep=0pt,label=\(\bullet\)}
\setlist[enumerate]{itemsep=0pt}
\allowdisplaybreaks
\numberwithin{equation}{section}
\newtheorem{theorem}{Theorem}[section]
\newtheorem{lemma}[theorem]{Lemma}
\newtheorem{proposition}[theorem]{Proposition}
\newtheorem{corollary}[theorem]{Corollary}
\theoremstyle{definition}
\newtheorem{definition}[theorem]{Definition}
\newtheorem{remark}[theorem]{Remark}
\newcommand{\iu}{\mathrm{i}}
\newcommand{\vct}[1]{\bm{#1}}
\newcommand{\ket}[1]{\lvert#1\rangle}
\newcommand{\bra}[1]{\langle#1\rvert}
\newcommand{\norm}[1]{\left\lVert#1\right\rVert}
\newcommand{\Id}{I}
\newcommand{\C}{\mathbb C}
\newcommand{\R}{\mathbb R}
\newcommand{\Z}{\mathbb Z}
\newcommand{\op}{\mathrm{op}}
\newcommand{\F}{\mathrm F}
\newcommand{\diag}{\operatorname{diag}}

\newcommand{\purpose}[1]{\par\noindent\textbf{Purpose.} #1\par\medskip}
\newcommand{\deliverable}[1]{\par\medskip\noindent\textbf{Required output.} #1\par}

\newcommand{\appendixspace}{\par
\ifdim\dimexpr\pagegoal-\pagetotal\relax<12\baselineskip
\newpage\fi}

\definecolor{zwcolor}{RGB}{180,100,0}

\newcommand{\email}[1]{%
  \thanks{Email: \href{mailto:#1}{\nolinkurl{#1}}}%
}

 \title{Exact Unitary Synthesis: Simultaneously Optimal in Size and Depth-Ancilla Tradeoff}

\author{%
  Pei Yuan\textsuperscript{1,*}%
  \thanks{These authors contributed equally to this work.}%
  \email{peiyuan@tencent.com}%
  \quad
  Wei Zi\textsuperscript{2,*}%
  \footnotemark[1]%
  \email{ziwei.quantum@outlook.com}%
  \quad
  Shengyu Zhang\textsuperscript{1}%
  \email{shengyzhang@tencent.com}%
  \\[0.8em]
  {
    \textsuperscript{1}Tencent Quantum Laboratory,
    Shenzhen, Guangdong, China
  }%
  \\[0.2em]
  {
    \textsuperscript{2}T Lab, Shenzhen, Guangdong, China
  }%
}
\date{}
\begin{document}
\maketitle

\begin{abstract}
Unitary synthesis underlies universal quantum computation by translating unitary transformations into circuits of elementary gates.
We show that for any budget $m\geq 0$ of ancilla, every $n$-qubit unitary can be exactly implemented using a quantum circuit consisting of $O(4^n)$ CNOT and single-qubit gates in $O(n+4^n/(n+m))$ layers
attaining the asymptotically optimal depth--ancilla tradeoff while simultaneously achieving optimal circuit size.
{The} key idea is to encode computational basis states as wavefunctions sharing a Gaussian envelope, converting unitary action on amplitudes into coordinate transformations that admit shallow circuit implementations.
The continuous encoding is then discretized, with approximation errors removed by exact residual correction. The depth-ancilla tradeoff follows by recursively applying the quantum Shannon decomposition, synthesizing its smaller unitary factors, and reusing ancillary qubits.
\end{abstract}

\pdfbookmark[0]{Contents}{contents}
\begingroup
\setlength{\parskip}{0pt}
\tableofcontents
\endgroup

\section{Introduction}
\label{sec:intro}

Unitary synthesis connects the mathematical description of a quantum computation to its realization by elementary gates. Given an $n$-qubit unitary \(U\in\mathrm U(2^n)\), the unitary synthesis (US) problem asks
for a quantum circuit implementing the following transformation
\begin{equation}
    \ket{\vct\psi}\mapsto U\ket{\vct\psi}, \quad \text{for any $n$-qubit input state }\ket{\vct\psi}.
\end{equation}
Arbitrary single- and two-qubit gates suffice for exact synthesis \cite{Barenco1995}, and general decomposition methods provide basic tools for quantum circuit compilation \cite{ShendeBullockMarkov2006}. Unitary extensions also give a route to implementing state preparation, isometries, and quantum channels, with suitable ancillary initialization and discarding
\cite{ItenEtAl2016}. Understanding synthesis complexity therefore clarifies the fundamental cost of quantum operations and supplies building blocks for more structured computations.

The complexity of a quantum circuit is usually measured by its size, depth, and number of qubits. The size is the number of elementary gates, and the depth is the number of layers, where gates acting on disjoint qubits can be applied in parallel. Circuit depth and the number of qubits correspond to time and space, respectively. Smaller depth can shorten the computation time, which is important for quantum devices with limited coherence time, while fewer gates can reduce the accumulation of gate errors. However, reducing the depth by parallelization may increase the circuit size \cite{MooreNilsson2001}. Experimental benchmarks also consider both depth and the number of qubits, with performance affected by gate errors, connectivity, and compilation \cite{CrossEtAl2019}. It is therefore important to reduce the circuit depth while keeping the circuit size and the number of ancillary qubits small.

Ancillary qubits make this joint optimization especially important. They provide workspace for parallel execution, but additional quantum space is itself a limited resource. The relevant question is therefore how much sequential computation can be saved for each available ancilla budget, not only how small depth becomes with unrestricted workspace. This viewpoint underlies space--depth tradeoffs for CNOT circuits, state preparation, and general unitary synthesis
\cite{JiangEtAl2020,SunEtAl2023,YuanZhang2023,ZHANGETAL2022}. Although worst-case exact synthesis of \(n\)-qubit unitaries requires \(\Theta(4^n)\) gates \cite{ShendeMarkovBullock2004,ShendeBullockMarkov2006}, this does not determine its optimal depth. A tight depth--ancilla tradeoff quantifies the benefit of workspace and identifies when further ancilla cease to reduce the asymptotic depth. Our objective is to attain this tradeoff while preserving optimal gate complexity.

\subsection{Main results}
In this paper, we answer these questions by optimal and explicit circuit designs. First, we construct a family of circuits that achieves the optimal depth–ancilla tradeoff for any number of ancilla qubits, closing an exponential gap and thereby resolving a prominent open question raised in \cite{SunEtAl2023,YuanZhang2023}. The construction automatically achieves the optimal size.
\begin{theorem}[Optimal size and depth--ancilla tradeoff for unitary synthesis]
\label{thm:all-budget}
For every integer \(n\ge1\) and  \(m\ge0\), any \(U\in\mathrm U(2^n)\) can be implemented by a quantum circuit of $O(4^n)$ one- and two-qubit gates with depth
\begin{equation}
 O\!\left(n+\frac{4^n}{n+m}\right),
 \label{eq:all-budget-main}
\end{equation}
using at most $m$ ancillary qubits. Both of the depth and size are optimal.
\end{theorem}

This result matches the known lower bounds for both size $\Omega(4^n)$ and depth $\Omega(n+4^n/(n+m))$ simultaneously \cite{SunEtAl2023}. 

The key step in our proof is the following construction, which achieves linear depth using an asymptotically optimal number of ancillary qubits.
\begin{theorem}[Linear-depth with minimal ancilla]
\label{thm:main}
Any \(U\in\mathrm U(2^n)\) can be implemented by a standard circuit of size $O(4^n)$ and depth $O(n)$, using $O(4^n/n)$ ancillary qubits.
\end{theorem}

Theorem~\ref{thm:main} reduces the depth of exact unitary synthesis from the previously known exponential bounds \cite{YuanZhang2023,Rosenthal2026} to the asymptotically optimal \(O(n)\). This exponential improvement is achieved while retaining optimal gate complexity and using asymptotically minimal ancillary space among linear-depth implementations.

\subsection{Prior results and the bottlenecks}
Ancillary qubits have long been used to reduce circuit depth by parallel computation \cite{MooreNilsson2001}. For CNOT circuits, the optimal tradeoff between depth and the number of ancillary qubits is known \cite{JiangEtAl2020}. For general unitary synthesis, Ref.~\cite{SunEtAl2023} constructed quantum circuits with optimal depth $O(4^n/(n+m))$ and size \(O(4^n)\) when the number of ancillary qubits satisfies \(m=O(2^n/n)\). For quantum state preparation, Yuan and Zhang \cite{YuanZhang2023} obtained the optimal depth \(O(n+2^n/(n+m))\) with \(m\) ancillary qubits, together with a corresponding result for controlled quantum state preparation. They also constructed circuits of depth \(O(n2^{n/2})\) for general unitary synthesis when sufficiently many ancillary qubits are available. Rosenthal obtained exact unitary synthesis of depth \(\widetilde O(2^{n/2})\) using \(\widetilde O(4^n)\) ancillary qubits \cite{Rosenthal2026}. However, these depth and ancillary-qubit bounds alone do not imply a circuit size of \(\widetilde O(4^n)\). The remaining challenge in the regime with many ancillary qubits is to achieve linear depth while maintaining optimal circuit size and using an asymptotically minimal number of ancillary qubits.

Previous constructions use controlled quantum state preparation (CQSP) as a subroutine for unitary synthesis. With sufficiently many ancillary qubits, the CQSP \(\lvert j\rangle\lvert0\rangle\mapsto\lvert j\rangle U\lvert j\rangle\) can be implemented in depth \(O(n)\). However, implementing \(U\) on an arbitrary input also requires uncomputing the register containing \(j\) while preserving \(U\lvert j\rangle\). Rosenthal \cite{Rosenthal2026} achieves this by exact Grover search, using \(\Theta(2^{n/2})\) sequential iterations involving the state-preparation circuit and its inverse. Yuan and Zhang \cite{YuanZhang2023} improve the controlled state-preparation circuit and combine this reduction with a truncated cosine-sine decomposition to obtain better depth bounds for different numbers of ancillary qubits. The decomposition reduces the number of target qubits in each subproblem, but introduces more sub-circuits that are applied sequentially. For sufficiently many ancillary qubits, their construction achieves depth \(O(n2^{n/2})\), with the exponential factor coming from the Grover iterations. Thus, reducing the depth of state preparation alone does not remove the sequential cost in these constructions. The remaining challenge is to implement general unitaries without this long sequence of oracle calls, while controlling the circuit size and the number of ancillary qubits.

\subsection{Technical overview}
\label{subsec:technical-overview}

We first construct a quantum circuit of depth $O(n)$ and size
$O(4^n)$, using $O(4^n/n)$ ancillary qubits. We then extend this
construction to any given number of ancillary qubits. The main idea
is to encode the input so that a transformation of coordinates
implements the desired transformation of amplitudes. We first obtain
an approximate circuit, reduce its size and ancillary requirements,
and then remove its error exactly.

\paragraph{From amplitudes to coordinate transformations.}
We associate each computational-basis state with one real coordinate.
The basis state $\ket j$ is encoded as a Gaussian wavefunction
multiplied by the corresponding coordinate variable. Thus, for any
input state $\ket\psi=\sum_j\psi_j\ket j$, the encoded wavefunction is
\[
    f_\psi(\vct x)
    =\sqrt{2}\left(\sum_j\psi_jx_j\right)G(\vct x),
\]
where $G$ is the normalized Gaussian proportional to
$e^{-\norm{\vct x}^2/2}$. The input amplitudes are now the coefficients
of a linear function of the coordinates.

An orthogonal transformation of the coordinates induces the same
transformation on these coefficients. The Gaussian remains unchanged
because it depends only on the distance from the origin. For a
general complex unitary, we use a fixed ancillary state to embed
its action into a real orthogonal transformation on a larger space.
Writing $\mathcal C$ for the resulting encoding and $\mathcal P_U$
for the corresponding coordinate operation, we obtain
\[
    \mathcal P_U\mathcal C=\mathcal C U.
\]
This identity gives a way to implement $U$: encode the input state,
transform the coordinates, and decode. 

\paragraph{A finite implementation of linear-depth approximate synthesis.}
The continuous representation explains the desired action.
To obtain a quantum circuit, we discretize both the encoding and
the coordinate operations. The product structure of the Gaussian
allows the encoding to be implemented efficiently. Each encoded
basis state differs from the common product Gaussian in only one
coordinate. We transfer the basis label into a unary encoding and
then prepare the local coordinate states in parallel. A symmetric
grid keeps the local even and odd profiles exactly orthogonal,
so the finite encoding is an exact isometry.

For the coordinate operation, we enlarge the coordinate space
once more and decompose the required orthogonal transformation
into three shears and a signed swap. A typical shear has the form
\[
    (\vct x,\vct y)\longmapsto
    (\vct x+H\vct y,\vct y),
\]
where $H$ is determined by the target unitary. It updates one
coordinate block while preserving the other. After a Fourier
transform on the block being updated, the shear becomes a diagonal
bilinear phase. The Fourier transforms on different coordinates
run in parallel. 

There are only a constant number of these stages. Gaussian decay
and Fourier estimates control the discretization error across all
coordinates, giving an exponentially accurate circuit of depth
$O(n)$. The remaining difficulty is that implementing every phase
term separately uses too many gates and ancillary qubits.

\paragraph{Reducing the size and ancilla-count of approximate synthesis.}
We improve the implementation of the diagonal phase gates by
dividing their coefficients into strong and weak terms. A term is called strong if the magnitude of its coefficient exceeds a chosen threshold, and weak otherwise.
Orthogonality bounds the total squared magnitude of the underlying
matrix entries, so sufficiently large coefficients cannot be
numerous. We implement the strong terms exactly. The weak
coefficients are individually small, but there may be many of them,
and their combined contribution cannot simply be omitted.

For the weak terms, we group the coefficients within each row.
Weighted state preparation encodes their relative magnitudes,
and reversible coordinate reads provide the corresponding
coordinate values. Together, these operations encode the weighted
contribution of an entire group into a circuit amplitude.
The grouping is chosen so that each group's total phase is small.
A fixed-degree polynomial, followed by amplitude amplification,
then approximates the corresponding phase operation accurately
enough for the total error over all groups to remain small.
This avoids implementing every weak interaction separately by
the direct construction.

The groups run in parallel, but their workspaces must also be
controlled. The state-preparation depth--ancilla tradeoff allows
each group to use fewer ancillary qubits when its preparation is
allowed depth $O(n)$. We also batch the coordinate reads so that
different groups share copies of the coordinate labels.
These choices give an approximate circuit with depth $O(n)$,
size $O(4^n)$, and $O(4^n/n)$ ancillary qubits.

\paragraph{Removing the approximation error.}
The next step corrects the error of the actual finite circuit.
Let $M$ be its matrix block corresponding to all-zero input and
output ancillary registers. The residual $U-M$ contains all
errors from discretization and from approximating the diagonal
phases. We choose the approximation accuracy so that
\[
    \norm{U-M}_{\F}\le 1.
\]
State preparation and controlled state preparation allow us to
encode the normalized residual into the corresponding block of
another unitary circuit. This correction is determined by $U$
and the chosen approximate circuit; it does not depend on learning
the input amplitudes.

We then combine the approximate circuit and the residual circuit
coherently by the linear combination of unitaries. Then we implement the block-encoding $W$ of $U$ satisfying
\[
    (\Id\otimes\bra0)\,W\,(\Id\otimes\ket0)=\frac U2,
\]
By amplitude-amplification, the unitary $U$ can be realized by using $3$ times of $W$ and $W^\dagger$. The controls and reflections in amplitude amplification also fit within the same depth, size, and ancillary bounds.

\paragraph{The tradeoff for any number of ancillary qubits.}
When fewer ancillary qubits are available, we recursively apply
the quantum Shannon decomposition to express $U$ as smaller
unitaries and uniformly controlled rotations. We stop when each
remaining unitary is small enough that its linear-depth
implementation fits within the available workspace. These
subcircuits are applied in the order specified by the decomposition.

\subsection{Organization of the manuscript}
\label{sec:organization}
The rest of this manuscript is organized as follows. Section~\ref{sec:models} introduces the notation, circuit model, and
basic circuit primitives. Section~\ref{sec:geometry} constructs a
linear-depth circuit for approximate unitary synthesis using Gaussian
encoding and finite coordinate transformations.
Section~\ref{sec:spectral} reduces its size and ancillary requirements
while preserving linear depth. Section~\ref{sec:exactification}
converts the approximate circuit into an optimal exact implementation without
changing the asymptotic resource bounds.
Finally, Section~\ref{sec:tradeoff} extends this construction to any number
of ancillary qubits and establishes the optimal depth--ancilla
tradeoff.

\paragraph{Relation to \cite{NehoranYuen2026}} We obtained all the results by early September 2026, primarily through extensive interactions with ChatGPT  (see the AI statement below). Subsequently, we focused on simplifying and rewriting the proofs to improve readability. During the final preparation of the manuscript, we became aware of a recent preprint by Nehoran and Yuen \cite{NehoranYuen2026}, posted on arXiv on September 30, 2026, which presents an \textit{approximate} synthesis using the same transform  to continuous functions. In comparison, our results achieve \textit{exact} synthesis (Section \ref{sec:exactification}) with \textit{simultaneously optimal} size and depth-ancilla tradeoff (Sections \ref{sec:spectral} and \ref{sec:tradeoff}).

\paragraph{AI statement.}
The authors acknowledge the use of OpenAI’s GPT-5.6-Sol and GPT-6-Astra for author-directed exploratory discussions, mathematical brainstorming, and manuscript preparation. We have investigated the unitary synthesis problem since 2021 \cite{SunEtAl2023}, and tried to use extensive prompts and intensive discussions with GPT-5.6-Sol in the first half of this year, but did not improve upon previous published result. On September 6, 2026, we engaged with the newly released GPT-6-Astra using the same prompt. Astra successfully proposed the crucial idea of transforming the problem via continuous functions (as detailed in Section 3.1), initially reducing the depth to $\text{poly}(n)$. Through several hours of subsequent directed interactions on the same day, we jointly refined the approach to achieve optimal linear depth. Then on September 8, 2026, we further prompted Astra to utilize the newly obtained optimal linear unitary circuit in the quantum cosine-sine decomposition framework, leading Astra to the derivation of the optimal depth-ancilla trade-off. 

All AI-generated mathematical claims, derivations, and initial drafts were independently verified, simplified, and largely rewritten by the human authors. The authors take full responsibility for the correctness, originality, and integrity of the mathematical proofs and final results presented in this paper.

\section{Preliminaries}
\label{sec:models}

\subsection{Notation}
\paragraph{Notation} All logarithms are base two. We write $\Id$ for the identity operator, with its dimension
determined by context, and $\mathrm U(N)$ for the group of $N\times N$
unitary matrices. We use $A^T$ for transpose and $A^\dagger$ for conjugate transpose. For Hermitian matrices, $A\preceq B$ means that $B-A$ is positive semidefinite.
The computational basis of an $r$-qubit register is indexed by
$x\in\{0,1\}^r$, or equivalently by integers $0\le x<2^r$.
We write $\ket{0^r}$ for its all-zero basis state, abbreviating it
as $\ket0$ when the register is clear. An orthonormal coordinate
basis is denoted by $\vct{e}_j$. For any $i,j\in\Z$, $\delta_{i,j}=1$ if $i=j$; $\delta_{i,j}=0$ otherwise. Let $\vct 1$ denote the all-one vector, with its dimension determined by context.
For a vector, $\norm{\vct{x}}$ is its $2$-norm. For a matrix $A\in\mathbb{C}^{N\times N}$, $\norm A_{\op}=\sup_{\norm{\vct{x}}=1}\norm{A\vct{x}}$ and $\norm A_{\F}=\left(\sum_{1\le i,j\le N}|A_{ij}|^2\right)^{1/2}$.
It is a standard fact that $\norm A_{\op}\le\norm A_{\F}
\le\sqrt{N}\norm A_{\op}$. For $f\in L^2(\R^\ell)$, $\norm f=\norm f_2$ denotes its norm with respect to Lebesgue measure.

The quantum fan-out gate $F_n$ is an $(n+1)$-qubit transformation satisfying
\begin{equation}
    F_n\ket{x_0}\ket{x_1,x_2,\ldots,x_n}=\ket{x_0}\ket{x_1\oplus x_0,x_2 \oplus x_0,\ldots,x_n\oplus x_0}
\end{equation}
for any $x_0,x_1,\ldots,x_n\in\{0,1\}$.

The global problem parameters are the number $n$ of input qubits
and the ancilla-count $m$. For the endpoint construction we use
\begin{equation}
 d=2^n,\qquad t=n+1,\qquad p=2d=2^t,\qquad \ell=2p=4d.
 \label{eq:construction-dimensions}
\end{equation}
Thus $d$ is the logical dimension, $p$ is the realified dimension,
$t=\log p$, and $\ell$ is the number of coordinates in the continuous representation.

\subsection{Circuit models and unitary synthesis problem}
We adopt the following circuit models in this paper.
 A \textit{standard circuit} consists of arbitrary single- and two-qubit gates. Mid-circuit measurements and classical feedforward are not allowed in this model. Ancillary qubits are allowed, which are initialized in the all-zero state and are required to be restored to all-zero state at the end of the circuit.

\begin{definition}[Unitary synthesis (US) problem]
\label{def:exact-synthesis}
Given integers $n\ge1$, $m\ge0$, and $U\in\mathrm U(2^n)$,
the unitary synthesis problem is to construct a circuit $T$
in the standard model on $n+m$ qubits that deterministically
implements
\begin{equation}
 T(\ket{\vct{\psi}}\ket{0^m})=
 (U\ket{\vct{\psi}})\ket{0^m}
 \label{eq:exact-goal}
\end{equation}
for every input state $\ket{\vct{\psi}}$, using $m$ ancillary qubits.

As mentioned in Section \ref{sec:intro}, a key objective is to bound the quantum circuit depth for a given ancilla-count $m$.
\end{definition}

We also introduce the approximate unitary synthesis problem, which serves as an intermediate step in solving the US problem.
\begin{definition}[Approximate unitary synthesis problem]
\label{def:approximate-synthesis}
Let $n\ge1$, $m\ge0$, $U\in\mathrm U(2^n)$, and $\epsilon\ge0$.
A unitary circuit $T$ on $n+m$ qubits implements $U$ with error
at most $\epsilon$ if
\begin{equation}
    \norm{TJ_m-J_mU}_{\op}\le\epsilon, \text{ where } J_m \text{ satisfies }
    J_m\ket{\vct\psi}
      =\ket{\vct\psi}\ket{0^m} \text{ for every } n\text{-qubit state } \ket{\vct \psi}.
    \label{eq:approximate-goal}
\end{equation}
\end{definition}
Intuitively, this definition implies that $T(\ket{\vct \psi}\ket{0^m}) = T(J_m\ket{\vct\psi}) \approx_\epsilon J_m U \ket{\vct\psi} = U \ket{\vct\psi}\ket{0^m}$. The error controls the entire output, including all ancillary
qubits, and retains the prescribed global phase of $U$.
The same bound holds for inputs entangled with an arbitrary reference system.
When $\epsilon=0$, this reduces to exact synthesis with every ancillary qubit restored to zero.

\subsection{Circuit primitives}
In this section, we introduce some known circuit primitives which are utilized in our circuit construction.

The optimal depth--ancilla tradeoff of quantum state preparation (QSP) and its controlled version (CQSP) are given in \cite{YuanZhang2023} as stated in the following two lemmas, where the size cost is also optimal.
\begin{lemma}[Quantum state preparation {\cite{YuanZhang2023}}]
\label{lem:state-tradeoff}
Let $r\ge1$ and $m\ge0$ be integers. Any $r$-qubit state
$\ket{\vct{\psi}}=\sum_{x=0}^{2^r-1}\alpha_x\ket x$ can be prepared by a standard circuit of size $O(2^r)$ and depth $O(r+2^r/(r+m))$, using at most
$m$ ancillary qubits.
\end{lemma}

\begin{lemma}[Controlled quantum state preparation (CQSP) {\cite{YuanZhang2023}}]
\label{lem:controlled-state-tradeoff}
Let $k,m\ge0$ and $r\ge1$ be integers, and let
$\{\ket{\vct{\psi}_a}:a\in\{0,1\}^k\}$ be a set of $r$-qubit states. The $(k,r)$-CQSP
\[
 \ket a\ket{0^r}
 \longmapsto \ket a\ket{\vct{\psi}_a},
 \qquad a\in\{0,1\}^k
\]
can be implemented by a standard circuit of size $O(2^{k+r})$ and depth
$O(k+r+2^{k+r}/(k+r+m))$, using at most $m$ ancillary qubits.
\end{lemma}

\begin{lemma}[Uniformly controlled one-qubit gates, \cite{SunEtAl2023}]
\label{lem:tradeoff-ucg}
Let $j\ge1$ and $m\ge0$ be integers and $\{u_x\in \mathrm{U}(2):x\in\{0,1\}^{j-1}\}$ be a set of single-qubit gates. Any $j$-qubit uniformly controlled one-qubit gate $\mathcal M_j=\sum_{x\in\{0,1\}^{j-1}}\ket{x}\bra{x}\otimes u_x$ can be implemented by a standard circuit of size $O(2^j)$ and depth $O(j+2^j/(j+m))$, using at most $m$ ancillary qubits.
\end{lemma}

\begin{lemma}[Unitary without ancillary qubits, \cite{SunEtAl2023}]
\label{lem:US-no-ancilla}
Let $n\ge1$. Any $n$-qubit unitary $U\in{\rm U}(2^n)$ can be implemented by a standard circuit of size $O(4^n)$ and depth $O(4^n/n)$, without ancillary qubits.
\end{lemma}

The following lemma transforms one-hot encoding to binary encoding. More transform methods among one-hot, unary and binary encodings can be found in \cite{YuanZhang2025}.
\begin{lemma}[\cite{SunEtAl2023}]\label{lem:one-hot}
    The following $2^n$-qubit transformation
    \begin{equation}\label{eq:one-hot}
        \ket{\vct e_j}\mapsto\ket{j}\ket{0^{2^n-n}},\qquad 0\le j <2^n,
    \end{equation}
    can be implemented by a standard circui of depth $O(n)$ and size $O(2^n)$, using $O(2^n)$ ancillary qubits.
\end{lemma}

\begin{lemma}[Fan-out gate, \cite{LowEtAl2024}]\label{lem:fanout-standard}
    An $n$-qubit fan-out gate $F_{n-1}$ can be implemented by a standard circuit of depth $O(\log(n))$ and size $O(n)$, without using ancillary qubits.
\end{lemma}
The Fourier transform on
$L^2(\mathbb R^p)$ can be specified by its action on generalized position states as follows.
\begin{equation}
    \mathcal F\ket{\vct x}
    =
    \frac{1}{(2\pi)^{p/2}}
    \int_{\mathbb R^p}
    e^{-\iu\vct\zeta^T\vct x}
    \ket{\vct\zeta}\,d\vct\zeta.
\end{equation}
The discrete Fourier transform is well known to admit an efficient quantum circuit.

\begin{lemma}[Quantum Fourier transform (QFT), \cite{MooreNilsson2001}]\label{lem:QFT-standard-circuit}
    Any $q$-qubit quantum Fourier transform $\mathcal{F}_{2^q}$ defined by
    \begin{equation}\label{eq:QFT-standard-circuit}
        \mathcal{F}_{2^q}\ket{r}=\frac{1}{\sqrt{2^q}}\sum_{k=0}^{2^q-1} e^{2\pi\iu kr/2^q}\ket{k},\qquad 0\le r <2^q,
    \end{equation}
    can be implemented by a standard circuit of depth $O(q)$ and size $O(q^2)$ without using ancillary qubtis.
\end{lemma}
All the ancillary qubits are restored to zeros at the end of the circuits.

\section{Linear-depth circuit for approximate unitary synthesis}
\label{sec:geometry}
In this section, we construct a linear-depth circuit for approximate unitary synthesis using exponentially many gates and ancillary qubits.

Our approach is to transfer the action of the target unitary from amplitudes to coordinates, where we can exploit parallel coordinate updates.
The key is an encoding that makes a coordinate rotation reproduce
the desired transformation of the original amplitudes.

Specifically, we associate the target unitary $U$ with a real
orthogonal matrix $R$ and choose a fixed encoding $C$ such that,
in the continuous description,
\[
    \mathcal P_R C = C U,
\]
where $\mathcal P_R$ is the action on wavefunctions induced by
the coordinate transformation $\vct z\mapsto R\vct z$.
We can therefore implement $U$ by encoding the input, transforming
the coordinates, and decoding. Three questions make this idea concrete.
\begin{enumerate}
    \item
    \textbf{How can amplitude transformations become coordinate
    transformations?}
    We encode each basis state as a Gaussian wavefunction whose
    behavior under rotations gives the required correspondence.
    Section~\ref{subsec:gaussian-coordinate} explains this mechanism;
    Section~\ref{subsec:finite-core} constructs the finite encoding
    circuit.

    \item
    \textbf{How do we represent a complex unitary by a real
    orthogonal matrix?}
    We double the coefficient space while preserving the full
    action of $U$. The coefficient embedding uses only one ancilla
    and a fixed one-qubit gate, as described in
    Section~\ref{subsec:gaussian-coordinate}.

    \item
    \textbf{How do we implement the coordinate transformation
    in small depth?}
    We decompose it into three shears and a signed swap.
    Section~\ref{subsec:coordinate-shears} presents a relatively
    simple and intuitive spectral implementation that already
    achieves the asymptotically optimal depth $O(n)$.
    Section~\ref{sec:spectral} later improves its size and
    ancilla-count requirements.
\end{enumerate}

Combining these ingredients with the error analysis in
Section~\ref{subsec:core-error} gives the following bounds.

\begin{theorem}[Linear-depth approximate unitary synthesis]
\label{thm:core-construction}
Let $n\ge2$ and $0<\epsilon\le1/2$. Any $n$-qubit unitary $U\in\mathrm{U}(2^n)$ can be $\epsilon$-approximately implemented by a circuit of 
\begin{align*}
    & \text{depth:~~ $O(n+\log \log(1/\epsilon))$},\\ 
    & \text{size,~ancilla-count:~~ $O(4^n \log^2 (n+\log(1/\epsilon))+2^n\log(1/\epsilon))$,}
\end{align*}
\end{theorem}

In particular, for $\epsilon=2^{-cn}$ with fixed $c>0$,
these bounds simplify to
\[
    \text{depth:}\quad O(n),\qquad
    \text{size, ancilla-count:}\qquad O(4^n\log^2 n),.
\]
This simple core construction already achieves optimal worst-case
depth at exponential accuracy. Its size and ancillary count are
within factors $O(\log^2 n)$ and $O(n\log^2 n)$ of their respective
worst-case lower bounds. These gaps will be closed in Section \ref{sec:spectral}.

\subsection{From amplitudes to coordinate transformations}
\label{subsec:gaussian-coordinate}
 In this section, we represent a unitary operation as a continuous coordinate transformation acting on suitably chosen continuous-variable quantum states.
 
The encoding represents each computational-basis state by a
wavefunction. For a single qubit, consider two real coordinates
$(x,y)$ and the common Gaussian
\[
    G_2(x,y)=\pi^{-1/2}e^{-(x^2+y^2)/2}.
\]
We encode the two basis states as
\begin{equation}
    \ket0\longmapsto\phi_0(x,y)=\sqrt2\,xG_2(x,y),
    \qquad
    \ket1\longmapsto\phi_1(x,y)=\sqrt2\,yG_2(x,y).
    \label{eq:gaussian-two-coordinate-code}
\end{equation}
These functions give probability amplitudes at each coordinate;
their squared magnitudes give probability densities.
As Figure~\ref{fig:gaussian-basis-functions} illustrates,
$\phi_0$ changes sign across $x=0$, whereas $\phi_1$ changes sign
across $y=0$. Both are normalized, and their inner product vanishes
by parity. Their common Gaussian factor will allow coordinate
rotations to mix them in a particularly simple way.

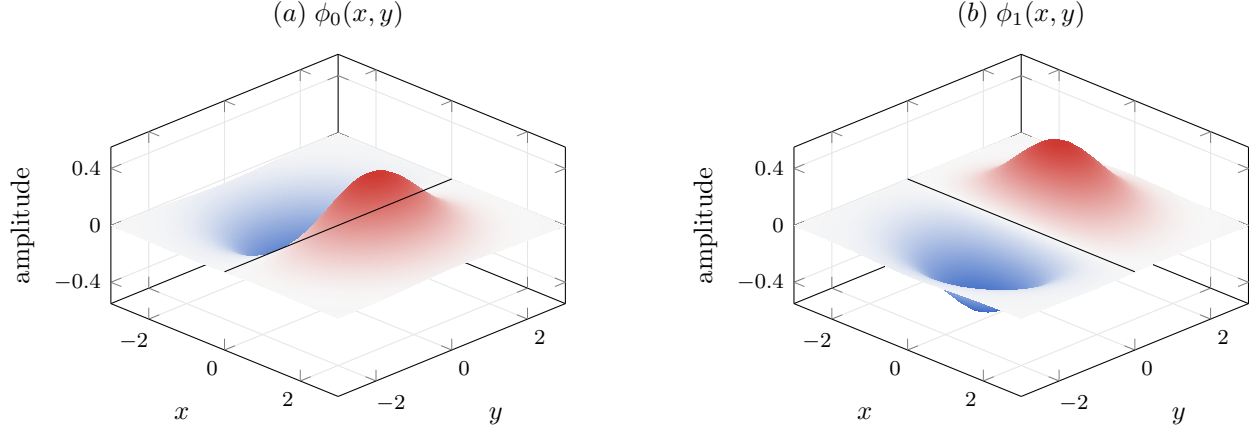
\begin{figure}[htbp]
    \centering
    \pgfplotsset{gaussianbasis/.style={
        width=.46\textwidth, height=.37\textwidth,
        view={45}{40}, z buffer=sort,
        domain=-3:3, y domain=-3:3,
        samples=35, samples y=35,
        xmin=-3, xmax=3, ymin=-3, ymax=3,
        zmin=-.55, zmax=.55,
        xtick={-2,0,2}, ytick={-2,0,2}, ztick={-.4,0,.4},
        xlabel={$x$}, ylabel={$y$}, zlabel={amplitude},
        tick label style={font=\scriptsize},
        label style={font=\small}, title style={font=\small},
        point meta min=-.5, point meta max=.5,
        colormap={signedamp}{
            rgb(0cm)=(.15,.35,.75);
            rgb(1cm)=(.97,.97,.97);
            rgb(2cm)=(.80,.18,.15)
        },
        grid=major, grid style={gray!20},
    }}
    \begin{tikzpicture}
        \begin{axis}[gaussianbasis,title={$(a)\;\phi_0(x,y)$}]
            \addplot3[surf,shader=interp]
                {sqrt(2/pi)*x*exp(-(x^2+y^2)/2)};
            \addplot3[black,thin] coordinates {(0,-3,0) (0,3,0)};
        \end{axis}
    \end{tikzpicture}\hfill
    \begin{tikzpicture}
        \begin{axis}[gaussianbasis,title={$(b)\;\phi_1(x,y)$}]
            \addplot3[surf,shader=interp]
                {sqrt(2/pi)*y*exp(-(x^2+y^2)/2)};
            \addplot3[black,thin] coordinates {(-3,0,0) (3,0,0)};
        \end{axis}
    \end{tikzpicture}
    \caption{The two encoded basis wavefunctions, shown on identical
    coordinate and amplitude scales. Red and blue indicate positive
    and negative amplitudes, respectively; the black lines mark their
    zeros at $x=0$ and $y=0$. The basis label determines which
    coordinate multiplies the common Gaussian.}
    \label{fig:gaussian-basis-functions}
\end{figure}

The definition extends to $n$ qubits, for which the $d=2^n$ computational basis states are encoded as $d$-variate Gaussian multiplied by coordinate variables.
\begin{definition}[Gaussian encoder]
Let $d=2^n$ and $\vct z=(z_0,\ldots,z_{d-1})^T$ be the real
coordinates.
Define the Gaussian encoder on computational basis states by
\begin{equation}
    \mathsf C_d\ket j=\phi_j, \text{ defined by }
    \phi_j(\vct z)=\sqrt2\,z_jG_d(\vct z),
    \quad 0\le j<d,
\label{eq:gaussian-coordinate-code}
\end{equation}  
where 
\begin{equation}
     G_d(\vct z)=\pi^{-d/4}e^{-\norm{\vct z}^2/2}.
\end{equation}
\end{definition}

The density $G_d^2$ has independent centered coordinates of
variance $1/2$. Hence
\begin{equation}
    \langle\phi_i,\phi_j\rangle
    =2\int_{\R^d}z_i z_jG_d(\vct z)^2\,d\vct z
    =\delta_{ij}.
    \label{eq:gaussian-coordinate-orthogonality}
\end{equation}
Thus the encoding preserves normalization and inner products.

The encoding extends to all $n$-qubit states by linearity. More specifically, take an arbitrary input state
\[\ket{\vct\psi}=\sum_{j=0}^{d-1}c_j\ket j,\]
with coefficient vector
$\vct c = (c_0,c_1,\ldots,c_{d-1})^T\in\C^d$. All basis wavefunctions share the same coordinate
space, so their amplitudes add pointwise:
\begin{equation}
\begin{aligned}
    f_{\vct c}(\vct z)
    &=\sum_{j=0}^{d-1}c_j\phi_j(\vct z)
      =\sqrt2\sum_{j=0}^{d-1}c_jz_jG_d(\vct z)\\
    &=\sqrt2\left(c_0z_0+c_1z_1+\cdots+c_{d-1}z_{d-1}\right)
      G_d(\vct z)\\
    &=\sqrt2\,(\vct c^T\vct z)G_d(\vct z).
\end{aligned}
\label{eq:gaussian-linear-wavefunction}
\end{equation}
The input amplitudes are now the coefficients of a linear
function of the coordinates. The transpose in $\vct c^T\vct z$
preserves linearity in the complex amplitudes.

In the continuous coordinate basis, the encoded state is
\[
    \ket{\Phi_{\vct c}}
    =\mathsf C_d\ket{\vct\psi}
    =\int_{\R^d}f_{\vct c}(\vct z)\ket{\vct z}\,d\vct z.
\]
For an orthogonal matrix $R\in\R^{d\times d}$, let $\mathcal P_R$
denote the coordinate transformation
\begin{equation}\label{eq:orth-PR}
    \mathcal P_R\ket{\vct z}=\ket{R\vct z}.
\end{equation}
This moves each coordinate label while preserving its amplitude.
Consequently,
\begin{align}
    \mathcal P_R\ket{\Phi_{\vct c}}
    &=\int_{\R^d}
      \sqrt2\,(\vct c^T\vct z)G_d(\vct z)
      \ket{R\vct z}\,d\vct z
      \notag\\
    &=\int_{\R^d}
      \sqrt2\,((R\vct c)^T(R\vct z))G_d(R\vct z)
      \ket{R\vct z}\,d\vct z
      \notag\\
    &=\int_{\R^d}
      \sqrt2\,((R\vct c)^T\vct w)G_d(\vct w)
      \ket{\vct w}\,d\vct w
      =\ket{\Phi_{R\vct c}}.
    \label{eq:gaussian-coordinate-action}
\end{align}
The second equality uses $R^TR=\Id$ and the rotational invariance
of the Gaussian. The third simply relabels $R\vct z$ as $\vct w$;
orthogonality ensures that the volume element is unchanged.
Thus,
\begin{equation}
    \mathcal P_R\mathsf C_d=\mathsf C_dR.
    \label{eq:gaussian-orthogonal-intertwining}
\end{equation}
Decoding on the encoded subspace gives
$\sum_j(R\vct c)_j\ket j$: moving the coordinates has implemented
$R$ on the original amplitudes.

The intuition behind the encoding rests on three key properties:
\begin{enumerate}
    \item
    \textbf{A linear factor carries the input amplitudes.}
    Pairing each amplitude $c_j$ with the coordinate function $z_j$
    produces $\vct c^T\vct z$. Orthogonality gives
    $\vct c^T\vct z=(R\vct c)^T(R\vct z)$, allowing a coordinate
    transformation to become the same transformation of the
    amplitude vector.

    \item
    \textbf{A radial factor is unchanged by rotations.}
    The common Gaussian depends only on $\norm{\vct z}$, so it
    remains invariant under the coordinate transformation.
    The transformed wavefunction therefore stays within the
    same family of encoded states.

    \item \textbf{The encoding supports controlled discretization.} To obtain a finite-qubit implementation, we subsequently discretize the continuous encoding and coordinate transformations. Gaussian decay and smoothness control truncation and grid approximation errors, respectively. The product structure enables parallel local encoding, while parity-preserving discretization on a symmetric grid maintains exact orthogonality of the encoded basis states.
\end{enumerate}

Other suitably normalized radial profiles can provide the second
property. We choose the Gaussian for its additional product structure:
\[
    \phi_j(\vct z)=e(z_j)\prod_{\nu\ne j}g(z_\nu),
    \qquad
    g(t)=\pi^{-1/4}e^{-t^2/2},\quad e(t)=\sqrt2\,t\,g(t).
\]
This factorization enables parallel local preparation in
Section~\ref{subsec:finite-core}.
We next embed an arbitrary complex unitary into a real orthogonal
transformation, extending the correspondence to the general target $U$.

\paragraph{From complex unitaries to real orthogonal matrices.}
Write $U=A+\iu B$, with $A,B\in\R^{d\times d}$, and set $p=2d$.
Define
\begin{equation}
    R_U=\begin{pmatrix}A&-B\\ B&A\end{pmatrix}\in\R^{p\times p}.
    \label{eq:unitary-realification}
\end{equation}
The unitarity of $U$ implies $R_U^TR_U=\Id_p$, so $R_U$ is
real orthogonal.

To connect its action to $U$, prepare one ancillary qubit in
$(\ket0-\iu\ket1)/\sqrt2$. With this qubit ordered first,
the resulting encoding is
\begin{equation}
    V\ket{\vct\psi}
    =\frac{\ket0-\iu\ket1}{\sqrt2}\otimes\ket{\vct\psi},
    \qquad
    V=\frac1{\sqrt2}
      \begin{pmatrix}\Id_d\\-\iu\Id_d\end{pmatrix}.
    \label{eq:complex-coefficient-embedding}
\end{equation}
Direct multiplication gives
\begin{equation}
    R_UV
    =\frac1{\sqrt2}
      \begin{pmatrix}A+\iu B\\ B-\iu A\end{pmatrix}
    =\frac1{\sqrt2}
      \begin{pmatrix}U\\-\iu U\end{pmatrix}
    =VU.
    \label{eq:realification-intertwining}
\end{equation}
Thus $R_U$ applies $U$ to the input while leaving the ancillary
state unchanged. The preparation uses the fixed one-qubit gate
$W=S^\dagger H$, where $S=\diag(1,\iu)$; applying $W^\dagger$
afterward restores the ancilla to $\ket0$, as illustrated in
Figure~\ref{fig:unitary-realification}.

\begin{figure}[htbp]
    \centering
    \begin{tikzpicture}[
        x=1cm,y=1cm,
        gate/.style={draw,fill=white,minimum width=8mm,
                     minimum height=7mm},
        every node/.style={font=\normalsize}
    ]
        \draw (0,0) -- (6,0);
        \draw (0,-1.1) -- (6,-1.1);

        \node[left] at (0,0) {$\ket0$};
        \node[left] at (0,-1.1) {$\ket{\vct\psi}$};

        \node[gate] at (1.2,0) {$W$};
        \node[gate,minimum width=17mm,minimum height=19mm]
            at (3,-.55) {$R_U$};
        \node[gate] at (4.8,0) {$W^\dagger$};

        \node[right] at (6,0) {$\ket0$};
        \node[right] at (6,-1.1) {$U\ket{\vct\psi}$};

        \draw (.4,-1.2) -- (.55,-1);
        \node[above,font=\small] at (.48,-1) {$n$};
    \end{tikzpicture}
    \caption{One fixed ancillary state converts the action of $R_U$
    into that of $U$. The middle block is realized, at the continuous
    level, by Gaussian encoding, the coordinate transformation
    $\mathcal P_{R_U}$, and decoding.}
    \label{fig:unitary-realification}
\end{figure}
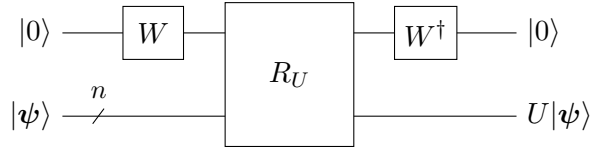

Combining $V$ with the Gaussian encoder on $p=2d$ coordinates gives
$C=\mathsf C_pV$. The two correspondences now yield
\begin{equation}
    \mathcal P_{R_U}C
    =\mathsf C_pR_UV
    =\mathsf C_pVU
    =CU.
    \label{eq:gaussian-unitary-intertwining}
\end{equation}
Thus the fixed encoding $C$ turns the action of an arbitrary $U$,
including its prescribed global phase, into a real orthogonal
transformation of coordinates.

\paragraph{From real orthogonal matrices to shears.}
Although $R_U$ is orthogonal, directly rounding its action on a finite grid may destroy reversibility. We instead use an alternative real orthogonal matrix
\begin{equation}
    F_U=\operatorname{diag}(R_U,R_U^T),
\end{equation}
which preserves the logical action of $R_U$ and decomposes into three shears and a signed swap. Each shear preserves its source block,
so rounded modular translations remain exactly reversible. This yields a finite-grid approximation implementable by a quantum circuit.

Recall that $d=2^n$ and $p=2d$. For any $\vct x,\vct y\in \R^p$, define the shears and signed swap
\begin{align}
    T_H(\vct x,\vct y) & =(\vct x+H\vct y,\vct y), \qquad H \in\R^{p\times p};\label{eq:shear-T}\\
    S_K(\vct x,\vct y) & =(\vct x,\vct y+K\vct x), \qquad K \in\R^{p\times p};\label{eq:shear-S}\\
    J_0(\vct x,\vct y) & =(-\vct y,\vct x).\label{eq:signed-swap}
\end{align}
By the definitions of shears and signed swap, we can verify that for any $\vct v, \vct y \in\R^{p}$,
\[
\begin{aligned}
   (\vct x,\vct y) &
    \stackrel{J_0}{\longmapsto}(-\vct y,\vct x)
     \stackrel{T_{R_U}}{\longmapsto}(R_U\vct x-\vct y,\vct x)
    \stackrel{S_{-R_U^T}}{\longmapsto}(R_U\vct x-\vct y,R_U^T\vct y)
     \stackrel{T_{R_U}}{\longmapsto}(R_U\vct x,R_U^T\vct y)=F_U(\vct x, \vct y).
\end{aligned}
\]
Then, the orthogonal matrix $F_U$ can be decomposed as three shears and a signed swap, i.e.,
    \begin{equation}
    F_U = T_{R_U}S_{-R_U^T}T_{R_U}J_0.
    \label{eq:shear-factorization}
\end{equation}

Based on the above discussion, the following Proposition \ref{prop:transform} shows that any unitary $U$ can be implemented by a coordinate transformation defined by $F_U$ on the encoded subspace.

Define the zero-padding isometry
\[
    V':\C^p\longrightarrow\C^\ell,\qquad
    V'\vct u=(\vct u,\vct 0_p).
\]
Thus the dimensionally typed continuous encoding relevant to the
two-block shear construction is
\[
  \vct c\in\C^d
  \xrightarrow{\;V\;}V\vct c\in\C^p
  \xrightarrow{\;V'\;}(V\vct c,\vct 0_p)\in\C^\ell
  \xrightarrow{\;\mathsf C_\ell\;}
  \mathsf C_\ell V' V\vct c\in L^2(\R^\ell).
\]
The zero-padding isometry $V'$ can be implemented by adding an ancillary qubit $\ket{0}$ as the most significant bit of the address register.
\begin{proposition}\label{prop:transform}
    Let $\ell=4d=2^{n+2}$. For any $U\in \mathrm{U}(d)$, there exist an isometry $C:\mathbb{C}^d\to L^2(\R^\ell)$ such that
    \begin{equation}
        \mathcal{P}_{F_U} C=CU, \qquad\text{with}\quad F_U=\diag(R_U,R_U^T), \qquad C=\mathsf C_\ell V' V,
    \end{equation}
    where $\mathsf C_\ell$ $(Eq.~\eqref{eq:gaussian-coordinate-code})$ is the Gaussian encoder, the coordinate transformation $\mathcal{P}_{F_U}$ $(Eq.~\eqref{eq:orth-PR})$ and real orthogonal matrix $R_U$ $(Eq.~\eqref{eq:unitary-realification})$ are defined by unitary $U$. 
\end{proposition}
\begin{proof}
 Let $\ket{\vct \psi}=\sum_{j=0}^{d-1}c_j\ket{j}$ and
$\vct c=(c_0,\ldots,c_{d-1})^T$.
The encoding $C$ maps state $\ket{\vct\psi}$ to
\[
  C\ket{\vct \psi}=\mathsf C_\ell V' V \ket{\vct\psi}=\int_{\R^\ell} f_{(V\vct c,\vct0_p)}(\vct z)\ket{\vct z}d{\vct z}.
\]
By Eqs.~\eqref{eq:gaussian-coordinate-action} and \eqref{eq:realification-intertwining}, we obtain
\[
\begin{aligned}
  \mathcal P_{F_U} C \ket{\vct\psi}=&\mathcal P_{F_U} \int_{\R^\ell} f_{(V\vct c,\vct0_p)}(\vct z)\ket{\vct z}d{\vct z}= \int_{\R^\ell}f_{F_U(V\vct c,\vct0_p)} (\vct z)\ket{\vct z} d \vct z\\
  =&\int_{\R^\ell}  f_{(R_UV\vct c,\vct0_p)} (\vct z)\ket{\vct z} d\vct z=\int_{\R^\ell}f_{(VU\vct c,\vct0_p)} (\vct z)\ket{\vct z} d\vct z\\
   =&\mathsf C_\ell V' V U\ket{\vct\psi}=CU\ket{\vct \psi},
\end{aligned}
\]
where $U\vct c$ is the amplitude vector of $U\ket{\vct\psi}$.
Thus, the coordinate transformation implements $U$
on the encoded subspace.
\end{proof}

We next discretize the coordinates and construct the Gaussian
encoding on finite qubit registers. Its product structure allows
the local encoding operations to be implemented in parallel.

\subsection{Implementing the Gaussian encoding on a finite grid}
\label{subsec:finite-core}

The product structure of the Gaussian reduces the encoding to a simple task: put one coordinate in the excited profile and all others in the ground profile. We implement this task by storing each coordinate in a finite register, converting the input label to a single excitation, and encoding all coordinates in parallel.

\paragraph{Finite coordinates and local profiles.}
To construct the finite Gaussian encoding, we first introduce the parameters in Table~\ref{tab:finite-grid-notation}.

\begin{table}[htbp]
    \centering
    \caption{Notation for the finite coordinate grid.}
    \label{tab:finite-grid-notation}
    \begin{tabularx}{\linewidth}{@{}lXl@{}}
        \toprule
        Symbol & Meaning & Definition or relation \\
        \midrule
        $B$
        & Half-width of the truncated coordinate interval $[-B,B)$.
        & A power of two. \\

        $h$
        & Grid spacing between adjacent coordinate values.
        & $h=B^{-1}$ \\

        $Q$
        & Number of grid points and dimension of one coordinate register.
        & $Q=2B/h=2B^2$ \\

        $q$
        & Number of qubits in one coordinate register.
        & $q=\log_2 Q$ \\

        $\mathcal I$
        & Set of signed integer labels for the grid points.
        & $\{-Q/2,\ldots,Q/2-1\}$ \\

        $k$
        & Integer label of a grid point, stored as a $q$-bit
          two's-complement word.
        & $k\in\mathcal I$ \\

        $t_k$
        & Coordinate value represented by the label $k$.
        & $t_k=h(k+\tfrac12)$ \\
        \bottomrule
    \end{tabularx}
\end{table}

A coordinate register contains $q$ qubits and represents the $Q$ equally spaced points in $[-B,B)$.
Here $k$ is an integer label, $\ket k$ is its computational-basis
representation, and $t_k$ is the coordinate value it represents.
We store $k$ in two's complement: a bit string
$b_{q-1}\cdots b_0\in\{0,1\}^q$ represents
\[
    k=-2^{q-1}b_{q-1}+\sum_{\alpha=0}^{q-2}2^\alpha b_\alpha.
\]
Then, an integer $k\in\mathcal{I}$ can be encoded as
\begin{equation}\label{eq:k-binary}
    \ket{k}:=\ket{b_{q-1},b_{q-2},\ldots,b_0}, \qquad b_{q-1},b_{q-2},\ldots,b_0\in\{0,1\}.
\end{equation}
In particular,
\begin{equation}
    t_{-k-1}=-t_k,
    \qquad
    \ket k\longmapsto\ket{-k-1}=X^{\otimes q}\ket k.
    \label{eq:finite-coordinate-negation}
\end{equation}
Thus coordinate negation requires only one layer of $X$ gates.

The following lemma defines finite-grid analogues of the two states below.
\begin{align}
    \ket{g}=&\int_{\R} g(t)\ket{t}dt,\qquad g(t)=\pi^{-1/4}e^{-t^2/2},\\
    \ket{e}=&\int_{\R} e(t)\ket{t}dt,\qquad e(t)=\sqrt{2}t\pi^{-1/4}e^{-t^2/2}.
\end{align}
\begin{lemma}
Let $t_k=h(k+\frac{1}{2})$ for $k\in\mathcal{I}$, $Z_0=\sum_{k\in\mathcal{I}}e^{-t_k^2}$ and $Z_1=\sum_{k\in\mathcal{I}}t_k^2e^{-t_k^2}$. The two $q$-qubit finite coordinate states
\begin{equation}
 \ket{\vct g_Q}
   =\sum_{k\in\mathcal I}\frac{e^{-t_k^2/2}}{\sqrt{Z_0}}\ket k,
 \qquad
 \ket{\vct e_Q}
   =\sum_{k\in\mathcal I}\frac{t_ke^{-t_k^2/2}}{\sqrt{Z_1}}\ket k
\end{equation}
are normalized and orthogonal.  Moreover, 
\[
 \begin{aligned}
 X^{\otimes q}\ket{\vct g_Q}&=\ket{\vct g_Q},&
 X^{\otimes q}\ket{\vct e_Q}&=-\ket{\vct e_Q}.
 \end{aligned}
\]
\end{lemma}

\begin{proof}
If \(k\in\mathcal I\), then \(-k-1\in\mathcal I\) and $t_{-k-1}=-t_k$.
Then we can obtain that
\[
 \langle \vct g_Q | \vct g_Q\rangle=\langle \vct e_Q | \vct e_Q\rangle=1,\qquad\langle\vct g_Q\mid\vct e_Q\rangle
   =\frac{\sum_{k\in\mathcal I}t_ke^{-t_k^2}}
          {\sqrt{Z_0Z_1}}=0.
\]
Since $\ket{-k-1}=X^{\otimes q}\ket{k}$, we have $ X^{\otimes q}\ket{\vct g_Q}=\ket{\vct g_Q}$, and $
 X^{\otimes q}\ket{\vct e_Q}=-\ket{\vct e_Q}$.
\end{proof}


\paragraph{From the input label to coordinate registers.}
Recall $d=2^n$ and $p=2d$. For the shear construction detailed in
Section~\ref{subsec:coordinate-shears}, we use two blocks of $p$
coordinates, denoted $x$ and $y$, so that $\ell=2p=4d$.
The $x$ block carries the encoded input; the additional $y$ block
starts in the ground profile in every coordinate. At the continuous
level, this simply appends the factor $G_p(\vct y)$ to the
wavefunction constructed above.

\begin{lemma}[Exact finite local profiles]\label{lem:LQ}
The transformation \(\mathcal L_Q\) satisfying
\begin{equation}\label{eq:LQ}
 \begin{aligned}
 \mathcal L_Q\ket0\ket{0^{q-1}}
   =\ket{\vct g_Q},\qquad \mathcal L_Q\ket1\ket{0^{q-1}}
   =\ket{\vct e_Q}
 \end{aligned}
\end{equation}
can be implemented by a standard circuit of depth \(O(q)\) and size \(O(Q)\), using \(O(Q)\) clean ancillary
qubits.
\end{lemma}
\begin{proof}
For integer $u\in\{0,1,\ldots,Q/2-1\}$, let $\tau_u=h(u+1/2)>0$. Define $(q-1)$-qubit quantum states
\begin{equation}
    \ket{\vct v_0}
      =\sum_{u=0}^{Q/2-1}\frac{\sqrt2 e^{-\tau_u^2/2}}{\sqrt{Z_0}}\ket u,
    \qquad
    \ket{\vct v_1}
      =\sum_{u=0}^{Q/2-1}
         \frac{\sqrt2\tau_u e^{-\tau_u^2/2}}{\sqrt{Z_1}}\ket u.
    \label{eq:finite-magnitude-profiles}
\end{equation}
Both $\vct v_0$ and $\vct v_1$ are normalized.
Define a $(1,q-1)$-CQSP $P$
\begin{equation}
    P\ket{b}\ket{0^{q-1}}=\ket{b}\ket{\vct v_b},\qquad b\in\{0,1\}.
\end{equation}

We next establish a decomposition of \(\mathcal L_Q\).
First, apply the $(1,q-1)$-CQSP $P$ on the initial state $\ket{b}\ket{0^{q-1}}$ for any $b\in\{0,1\}$ and obtain
\begin{equation}
    P\ket{b}\ket{0^{q-1}}=\ket{b}\ket{\vct v_b}.
\end{equation}
Second, apply a Hadamard gate $H$ to the first qubit $\ket{b}$:
\begin{align*}
\frac{1}{\sqrt{2}}(\ket{0}\ket{\vct v_b}+(-1)^b\ket{1}\ket{\vct v_b}) = \sum_{u=0}^{Q/2-1}\frac{(\tau_u)^b e^{-\tau_u^2/2}}{\sqrt{Z_b}}(\ket {0,u}+(-1)^b\ket{1,u})
\end{align*}
Finally, apply a fan-out gate $F_{q-1}$.  Conditioned on the first qubit being $\ket{1}$, $F_{q-1}$ maps $\ket{u}$ to $\ket{\overline{u}}$ by flipping all \(q-1\) qubits; if the first qubit is $\ket{0}$, it leaves $\ket{u}$ unchanged. Then we have
\begin{equation}\label{eq:afterF_q-1}
    \sum_{u=0}^{Q/2-1}\frac{(\tau_u)^b e^{-\tau_u^2/2}}{\sqrt{Z_b}}(\ket {0,u}+(-1)^b\ket{1,\overline{u}}).
\end{equation}
For $u\in\{0,1,\ldots,Q/2-1\}$ and $k\in\mathcal{I}$, according to ~Eq.~\eqref{eq:finite-coordinate-negation}, we have
\begin{align*}
&\ket{k}=\ket{0,u},\quad \text{if }k=u;\\
&\ket{k}=\ket{1,\overline{u}},\quad \text{if }k=-u-1.
\end{align*}
Therefore, Eq.~\eqref{eq:afterF_q-1} can be represented as
\begin{align} 
     \sum_{k=0}^{Q/2-1}\frac{e^{-\tau_k^2/2}}{\sqrt{Z_0}}
      \bigl(\ket k+\ket{-k-1}\bigr)
    =&\ket{\vct g_Q},     \quad \text{if }b=0;\\
     \sum_{k=0}^{Q/2-1}\frac{\tau_k e^{-\tau_k^2/2}}{\sqrt{Z_1}}
      \bigl(\ket k-\ket{-k-1}\bigr)
    =&\ket{\vct e_Q},    \quad\text{if }b=1.
\end{align}
Based on the above discussion, $\mathcal{L}_Q$ can be decomposed as 
\begin{equation}
    \mathcal{L}_Q=F_{q-1}\cdot (H\otimes I^{\otimes (q-1)})\cdot P.
\end{equation}

According to Lemmas \ref{lem:controlled-state-tradeoff} and \ref{lem:fanout-standard}, $P$ can be implemented in depth $O(q)$ and size $O(Q)$, using $O(Q)$ ancillary qubits, and $F_{q-1}$ can be implemented in depth $O(\log(q))$ and size $O(q)$ without ancillary qubits. Hence, the total depth, size and ancilla-count of $\mathcal{L}_Q$ are $O(q)+1+O(\log(q))=O(q)$, $O(Q)+1+O(q)=O(Q)$ and $O(Q)$, respectively.
\end{proof}

We now combine the local profiles with a one-hot representation of the input label to construct the finite Gaussian encoder.
\begin{theorem}\label{thm:CQ}
  Let $n\ge 2$, $d=2^n$, $p=2d$ and $\ell=2p$. Define the $(q\ell)$-qubit transformation $C_Q$ by
  \begin{equation}\label{eq:CQ-encode}
      C_Q\ket{j,0^{q\ell-n}} =\frac{1}{\sqrt2}(\ket{\vct\Phi_{j,Q}} -\iu\ket{\vct\Phi_{d+j,Q}}),
 \qquad 0\le j<d ,
  \end{equation}
  where 
  \begin{equation}
      \ket{\vct\Phi_{t,Q}} =\ket{\vct e_Q}_t \bigotimes_{\substack{0\le\nu<\ell\\ \nu\ne t}} \ket{\vct g_Q}_\nu, \qquad 0\le t<\ell.
  \end{equation}
  Then $C_Q$ can be implemented by a standard circuit of depth $O(q+n)$ and size $O(Qd)$, using $O(Qd)$ ancillary qubits.
\end{theorem}
\begin{proof}
    First, for the $q\ell$ input qubits $\ket{j,0^{q\ell-n}}$, we swap the state of first $n$ qubits $\ket{j}$ and the $(n+1)$th qubit $\ket{0}$ by an $(n+1)$-qubit swap network $\mathsf{SWAP}_{n+1}$:
    \begin{equation}
        (\mathsf{SWAP}_{n+1}\otimes I^{\otimes (q\ell-n-1)})\ket{j,0^{q\ell-n}}=\ket{0,j,0^{q\ell-n-1}},\qquad 0\le j <d.
    \end{equation}
    Let the first $\ell$ qubits $\ket{0,j,0^{\ell-n-1}}$ be the register $\mathsf A$.
    Second, apply $W=S^{\dagger}H$ to the first qubit of register $\mathsf A$. Then we have
    \begin{align}
        (W\otimes I^{\otimes (q\ell-1)})\ket{0,j,0^{\ell-n-1}}_{\mathsf A}\ket{0^{q\ell-\ell}}=&\frac{\ket{0,j,0^{\ell-n-1}}_{\mathsf{A}}-\iu\ket{1,j,0^{\ell-n-1}}_{\mathsf{A}}}{\sqrt{2}}\ket{0^{q\ell-\ell}},\label{eq:after-W1}\\
        =& \frac{\ket{j,0^{\ell-n-1}}_{\mathsf{A}}-\iu\ket{j+d,0^{\ell-n-1}}_{\mathsf{A}}}{\sqrt{2}}\ket{0^{q\ell-\ell}},\qquad 0\le j <d.\label{eq:after-W2}
    \end{align}
    In Eqs.~\eqref{eq:after-W1} and \eqref{eq:after-W2}, $\ket{j}$ denotes the $n$- and $(n+1)$-bit binary representations of $j$, respectively.
    Third, a binary-to-unary transformation $\mathcal{B}_\ell$ on register $\mathsf A$ implements
    \begin{equation}
        \frac{\mathcal{B}_\ell(\ket{j,0^{\ell-n-1}}_{\mathsf{A}}-\iu\ket{j+d,0^{\ell-n-1}}_{\mathsf{A}})}{\sqrt{2}}\ket{0^{q\ell-\ell}}=\frac{\ket{\vct e_j}_{\mathsf{A}}-\iu\ket{\vct e_{j+d}}_{\mathsf{A}}}{\sqrt{2}}\ket{0^{q\ell-\ell}},\qquad 0\le j <d,
    \end{equation}
    Next, we divide $\ket{0^{q\ell-\ell}}$ into $\ell$ parts, each of which consists of $q-1$ qubits. For $1\le k\le \ell$, we use a $q\ell$-qubit swap network $\mathsf{SWAP}_{q\ell}$ to move the $k$-th part immediately after the $k$-th qubit of register $\mathsf A$, and we can obtain that
    \begin{align}
        & \mathsf{SWAP}_{q\ell}\frac{\ket{\vct e_j}_{\mathsf{A}}-\iu\ket{\vct e_{j+d}}_{\mathsf{A}}}{\sqrt{2}}\ket{0^{q\ell-\ell}}\\
        =&\frac{1}{\sqrt{2}}\Big(\ket{10^{q-1}}_j\bigotimes_{\substack{0\le\nu<\ell\\ \nu\ne j}} \ket{0^q}_{\nu}-\iu \ket{10^{q-1}}_{d+j}\bigotimes_{\substack{0\le\nu<\ell\\ \nu\ne j+d}} \ket{0^q}_\nu \Big),\qquad 0\le j <d.
    \end{align}
    Finally, apply $(\mathcal{L}_Q)^{\otimes \ell}$ (Eq.~\eqref{eq:LQ}) and we have
    \begin{align*}
        & \frac{1}{\sqrt{2}}(\mathcal{L}_Q)^{\otimes \ell}\Big(\ket{10^{q-1}}_j\bigotimes_{\substack{0\le\nu<\ell\\ \nu\ne j}} \ket{0^q}_{\nu}-\iu \ket{10^{q-1}}_{d+j}\bigotimes_{\substack{0\le\nu<\ell\\ \nu\ne j+d}} \ket{0^q}_\nu \Big)\\
    =   &\frac{1}{\sqrt{2}}\Big(\ket{\vct e_Q}_j\bigotimes_{\substack{0\le\nu<\ell\\ \nu\ne j}} \ket{\vct g_Q}_{\nu}-\iu \ket{\vct e_Q}_{d+j}\bigotimes_{\substack{0\le\nu<\ell\\ \nu\ne j+d}} \ket{\vct g_Q}_\nu \Big)\\
    =&\frac{1}{\sqrt2}(\ket{\vct\Phi_{j,Q}} -\iu\ket{\vct\Phi_{d+j,Q}}),\qquad 0\le j <d.
    \end{align*}
    Based on the above discussion, $C_Q$ can be decomposed as follows:
    \begin{equation*}
        C_Q= (\mathcal{L}_Q)^{\otimes \ell}\cdot(\mathsf{SWAP}_{q\ell})\cdot (\mathcal{B}_\ell\otimes I^{\otimes (q\ell-\ell)})\cdot (W\otimes I^{\otimes (q\ell-1)})\cdot (\mathsf{SWAP}_{n+1}\otimes I^{\otimes(q\ell-n-1)}).
    \end{equation*}
    Since any qubit permutation can be implemented using at most two layers of disjoint SWAP gates, the networks $\mathsf{SWAP}_{n+1}$ and $\mathsf{SWAP}_{q\ell}$ admit standard-circuit implementations without ancillary qubits, of depth $O(1)$ and sizes $O(n)$ and $O(q\ell)$, respectively. $W$ is a single-qubit gate. The $\ell$-qubit binary-to-unary transformation $\mathcal{B}_\ell$ is an inverse transformation of Eq.~\eqref{eq:one-hot}. Then $\mathcal{B}_\ell$ can be implemented by a standard circuit of depth $O(n)$ and size $O(\ell)$, using $O(\ell)$ ancillary qubits by Lemma \ref{lem:one-hot}. Lemma \ref{lem:LQ} shows that $(\mathcal{L}_Q)^{\otimes \ell}$ can be implemented by a standard circuit of depth \(O(q)\) and size \(O(\ell Q)\), using \(O(\ell Q)\) clean ancillary qubits. In summary, the depth, size and ancilla-count of $C_Q$ are
    \begin{align*}
        &\text{depth:}\quad O(1)+1+O(n)+O(1)+O(q)=O(q+n),\\
        &\text{size:}\quad O(n)+1+O(\ell)+O(q\ell)+O(\ell Q)=O(Qd),\\
        &\text{ancilla-count:}\quad O(\ell Q)=O(Qd).
    \end{align*}
\end{proof}

\subsection{Implementing the coordinate shears on a finite grid}
\label{subsec:coordinate-shears}
In this section, we discretize the shears $T_H,S_K$ and the signed swap $J_0$, and construct quantum circuits implementing their discrete counterparts. Recall that these operations combined give $\mathcal{P}_{F_U}$ by Eq.~\eqref{eq:shear-factorization}.
 
\paragraph{Implementation of $T_H$ and $S_K$.}
Let $\mathcal F$ denote the unitary Fourier transform on
$L^2(\mathbb R^p)$, whose action on generalized position
states is
\begin{equation}
    \mathcal F\ket{\vct x}
    =
    \frac{1}{(2\pi)^{p/2}}
    \int_{\mathbb R^p}
    e^{-\iu\vct\zeta^T\vct x}
    \ket{\vct\zeta}\,d\vct\zeta.
\end{equation}
At the continuous-variable level, the shear
$T_H(\vct x,\vct y)=(\vct x+H\vct y,\vct y)$,
where $H\in\mathbb R^{p\times p}$, is implemented on
position states by the following three operations:
\begin{enumerate}
    \item Apply $\mathcal F$ to the first register:
    \begin{equation}\label{eq:TH-step1}
        (\mathcal F\otimes I)\ket{\vct x,\vct y}
        =
        \frac{1}{(2\pi)^{p/2}}
        \int_{\mathbb R^p}
        e^{-\iu\vct\zeta^T\vct x}
        \ket{\vct\zeta,\vct y}\,d\vct\zeta.
    \end{equation}

    \item Apply the diagonal unitary specified by
    \begin{equation}\label{eq:TH-step2}
        \ket{\vct\zeta,\vct y}
        \longmapsto
        e^{-\iu\vct\zeta^TH\vct y}
        \ket{\vct\zeta,\vct y}.
    \end{equation}
    By linearity, the resulting state is
    \begin{align}
        &\frac{1}{(2\pi)^{p/2}}
        \int_{\mathbb R^p}
        e^{-\iu\vct\zeta^T\vct x}
        e^{-\iu\vct\zeta^TH\vct y}
        \ket{\vct\zeta,\vct y}\,d\vct\zeta
        =
        \frac{1}{(2\pi)^{p/2}}
        \int_{\mathbb R^p}
        e^{-\iu\vct\zeta^T(\vct x+H\vct y)}
        \ket{\vct\zeta,\vct y}\,d\vct\zeta.
    \end{align}

    \item Apply $\mathcal F^\dagger$ to the first register:
    \begin{equation}\label{eq:TH-step3}
        (\mathcal F^\dagger\otimes I)
        \left(
        \frac{1}{(2\pi)^{p/2}}
        \int_{\mathbb R^p}
        e^{-\iu\vct\zeta^T(\vct x+H\vct y)}
        \ket{\vct\zeta,\vct y}\,d\vct\zeta
        \right)
        =
        \ket{\vct x+H\vct y,\vct y}.
    \end{equation}
\end{enumerate}
The second register remains in the coordinate basis
throughout. Similarly, the lower shear
$S_K(\vct x,\vct y)=(\vct x,\vct y+K\vct x)$
is implemented by applying $\mathcal F$ to the second
register, applying the diagonal phase operation
\begin{equation}
    \ket{\vct x,\vct\zeta}
    \longmapsto
    e^{-\iu\vct\zeta^TK\vct x}
    \ket{\vct x,\vct\zeta},
\end{equation}
and then applying $\mathcal F^\dagger$ to the second register.

\paragraph{Implementation of finite grid version of $T_H$ and $S_K$.}
To obtain a finite-qubit realization of the above construction, we discretize both the Fourier transform and the diagonal
unitary using the coordinate grid specified in
Table~\ref{tab:finite-grid-notation}.
Recall that $Q=2^q=2B^2$, $h=B^{-1}$, and
$\mathcal I=\{-Q/2,\ldots,Q/2-1\}$.
The coordinate and frequency labels are identified with
\begin{equation}
    \vct x=h(\vct r+\tfrac12\vct1),
    \qquad
    \vct y=h(\vct \ell+\tfrac12\vct1),
    \qquad
    \vct\zeta=\widehat h\,\vct k,
\end{equation}
where $\vct r,\vct \ell,\vct k\in\mathcal I^p$ and
\begin{equation}
    \widehat h=\frac{2\pi}{Qh}=\frac{\pi}{B}.
\end{equation}
Each component label is stored in a $q$-qubit register
using two's-complement representation.

First, the one-coordinate Fourier transform is discretized
as the unitary matrix $D$ defined by
\begin{equation}\label{eq:one-coordinateFT}
    D\ket r
    =
    \frac{1}{\sqrt Q}
    \sum_{k\in\mathcal I}
    e^{-2\pi\iu k(r+1/2)/Q}\ket k,
    \qquad r\in\mathcal I.
\end{equation}
Its phase is obtained by evaluating the continuous Fourier
kernel at the coordinate and frequency grid points, since
$\widehat h k\,h(r+1/2)=2\pi k(r+1/2)/Q$.
The normalization satisfies
$Q^{-1/2}=(h\widehat h/(2\pi))^{1/2}$.
Thus, the $p$-coordinate transform $\mathcal F$ is replaced
by 
\begin{equation}
    D^{\otimes p} 
\end{equation}
acting on the $pq$ qubits of the first
coordinate block.

Second, evaluating the continuous phase
$e^{-\iu\vct\zeta^TH\vct y}$ on the same grids gives the
finite diagonal unitary
\begin{align}\label{eq:Delta-H}
    \Delta^{\uparrow
}_H\ket{\vct k,\vct \ell}
    =&
    \exp\!\left[
        -\frac{2\pi\iu}{Q}
        \vct k^TH(\vct \ell+\tfrac12\vct1)
    \right]
    \ket{\vct k,\vct \ell}.
\end{align}
Here the first block contains frequency labels, while
the second block retains coordinate labels.
Applying $D^{\otimes p}$ to the first block, followed by
$\Delta^{\uparrow}_H$ and then $(D^\dagger)^{\otimes p}$, yields the
finite-grid Fourier shear 
\begin{equation}\label{eq:PH}
     {\widehat{T}_H}=(D^{\otimes p}\otimes I^{\otimes p})^\dagger\Delta^\uparrow_H(D^{\otimes p}\otimes I^{\otimes p}).
\end{equation}
The lower shear
is obtained analogously
by Fourier transforming the second block and using the
finite diagonal unitary
\begin{align}\label{eq:Delta-K}
\Delta^{\downarrow}_K\ket{\vct \ell,\vct k}
 =&\exp\!\left[
        -\frac{2\pi\iu}{Q}
        \vct k^TK(\vct \ell+\tfrac12\vct1)
    \right]
    \ket{\vct \ell,\vct k}.
\end{align}
Then $ {\widehat{S}_K}$ can be implemented by
\begin{equation}\label{eq:PK}
     {\widehat{S}_K}=(I^{\otimes p}\otimes D^{\otimes p})^\dagger\Delta_K^\downarrow (I^{\otimes p}\otimes D^{\otimes p}).
\end{equation}

The transformations $D$, $\Delta^\uparrow_H$ and $\Delta_K^\downarrow$ are unitary on the finite
registers. Their composition is a finite-grid analogue of
the continuous shear, rather than an exact coordinate
permutation for arbitrary real displacements.
Its approximation to the continuous transformation on
Gaussian-encoded states is analyzed in
Section~\ref{subsec:core-error}.
 

\begin{lemma}
\label{lem:fourier}
The $q$-qubit unitaries $D,D^\dagger\in\C^{Q\times Q}$ (Eq. \eqref{eq:one-coordinateFT}) can be implemented exactly
by standard circuits of depth $O(q)$ and size $O(q^2)$, without ancillary qubits.
\end{lemma}
\begin{proof}
The two's-complement word $\ket r$ is the ordinary unsigned binary representation of $r\bmod Q$. Apply the inverse quantum Fourier transform $\mathcal F_Q^\dagger$:
\begin{align}
   \mathcal{F}^\dagger_Q\ket{r}=&\left\{ \begin{array}{ll}
     \frac1{\sqrt Q}(\sum_{k=0}^{Q/2-1}
 e^{-2\pi\iu kr/Q}\ket{k}+\sum_{k=-Q/2}^{-1}
 e^{-2\pi\iu (k+Q)r/Q}\ket{k}),   & \text{if }r\ge 0;  \\
   \frac1{\sqrt Q}(\sum_{k=0}^{Q/2-1}
 e^{-2\pi\iu k(r+Q)/Q}\ket{k}+\sum_{k=-Q/2}^{-1}
 e^{-2\pi\iu (k+Q)(r+Q)r/Q}\ket{k}), &  \text{if }r<0.
    \end{array}\right.\\
    = & \frac{1}{\sqrt{Q}}\sum_{k\in\mathcal{I}} e^{-2\pi\iu k/Q}\ket{k},\qquad r\in\mathcal{I}.
\end{align}
Second, apply the $q$-qubit diagonal unitary
\[
 \Gamma_Q\ket k=e^{-\pi\iu k/Q}\ket k,\qquad k\in\mathcal{I}.
\]
The resulting coefficient is
$Q^{-1/2}e^{-2\pi\iu k(r+1/2)/Q}$, so $D=\Gamma_Q\mathcal{F}_Q^\dagger$.

To implement $\Gamma_Q$, write $k=\sum_{\alpha=0}^{q-1}w_\alpha b_\alpha$,
where $b_\alpha\in\{0,1\}$ and
\[
 w_\alpha=2^\alpha\quad(0\le\alpha<q-1),\qquad
 w_{q-1}=-2^{q-1}.
\]
Apply $\operatorname{Ph}(-\pi w_\alpha/Q)$ to the $\alpha$-th qubit, where $\operatorname{Ph}(\theta)=\diag(1,e^{\iu\theta})$.
These $q$ gates form one parallel layer. According to Lemma \ref{lem:QFT-standard-circuit}, $\mathcal{F}_Q^\dagger$ can be implemented by a standard circuit of depth $O(q)$ and size $O(q^2)$, without ancillary qubits. This completes the proof.

\end{proof}

\begin{lemma}
\label{lem:diag-finite}
Let $H,K\in\R^{p\times p}$. The finite 
$(q\ell)$-qubit diagonal unitaries $\Delta^{\uparrow}_H$ and $\Delta^{\downarrow}_K$ (Eqs.~\eqref{eq:Delta-H} and \eqref{eq:Delta-K}) can be implemented by standard circuits of depth $O(\log(pq))$ and size $O(p^2q^2)$, using $O(p^2q^2)$ ancillary qubits.
\end{lemma}

\begin{proof}
For any $\vct k,\vct \ell\in\mathcal{I}^p$, let registers $\mathsf X=(\mathsf X_0,\ldots,\mathsf X_{p-1})$ and $\mathsf Y=(\mathsf Y_0,\ldots,\mathsf Y_{p-1})$ denote the first and second $pq$ qubits of the input state $\ket{\vct k,\vct\ell}$, respectively. Each register $\mathsf X_j$ and $\mathsf Y_j$ contains $q$ qubits. Write $k_i=\sum_\alpha w_\alpha\kappa_{i\alpha}$ and
$\ell_j=\sum_\beta w_\beta\lambda_{j\beta}$, with bits
$\kappa_{i\alpha},\lambda_{j\beta}\in\{0,1\}$, where
$0\le i,j<p$ and $0\le\alpha,\beta<q$.
The phase exponent separates into linear and bilinear terms, giving the following product of one-qubit and two-qubit phase gates.
\begin{equation}\label{eq:circuit-Delta-H1}
 \begin{aligned}
 \Delta^\uparrow_H&=
 \left(\prod_{i,\alpha}(\operatorname{Ph}
       (\theta_{i\alpha}))_{\mathsf X_{i,\alpha}}\right)
 \left(\prod_{i,j,\alpha,\beta}
       (\operatorname{CP}
       (\theta_{ij\alpha\beta}))_{\mathsf X_{i,\alpha},\mathsf Y_{j,\beta}}\right),
 \end{aligned}
\end{equation}
where
\begin{equation}
    \theta_{i\alpha}=-\frac{\pi w_\alpha}{Q}\sum_{j=0}^{p-1}H_{ij},\qquad
 \theta_{ij\alpha\beta}=-\frac{2\pi}{Q}H_{ij}w_\alpha w_\beta,
\end{equation}
and $\mathsf X_{i,\alpha}$ is the qubit $\alpha$ of $\mathsf X_i$, and similarly for $\mathsf Y_{j,\beta}$.

The first term of Eq.~\eqref{eq:circuit-Delta-H1} consists of single-qubit gates, which can be implemented by a circuit of depth $O(1)$ and size $O(pq)$. Now we reduce the depth of the second term of Eq.~\eqref{eq:circuit-Delta-H1} by fan-out gates and ancillary qubits. For each of the $q\ell=2pq$ original qubits, we first create $pq$ copies by fan-out gates $(F_{pq})^{\otimes 2pq}$ using $2p^2q^2$ ancillary qubits:
\begin{equation}
    (F_{pq})^{\otimes 2pq}\ket{\vct k,\vct \ell}\underbrace{\ket{0^{2pq}}\cdots\ket{0^{2pq}}}_{pq~\text{copies}}=\ket{\vct k,\vct \ell}\underbrace{\ket{\vct k,\vct \ell}\cdots\ket{\vct k,\vct \ell}}_{pq~\text{copies}}
\end{equation}
For $0\le u<p$ and $0\le v<q$, apply the $1$-depth circuit $\bigotimes_{u,v}L_{u,v}$ on the $pq$ copies of $\ket{\vct k,\vct \ell}$, where
\begin{equation}
    L_{u,v}:=\prod_{0\le i<p \atop 0\le \alpha<q}\left(\operatorname{CP}(\theta_{i,j_i,\alpha,\beta_\alpha})\right)_{\mathsf X_{i,\alpha},\mathsf Y_{j,\beta_\alpha}},\quad j_i=(i+u)\mod p,\text{~and~} \beta_\alpha=(\alpha+v)\mod p
\end{equation}
and $L_{u,v}$ acts on the $(uv)$-th copy of $\ket{\vct k,\vct\ell}$.
Finally, restore the ancillary qubits by $(F_{pq})^{\otimes 2pq}$. The phase gates $\operatorname{Ph}$ and $\operatorname{CP}$ leave the states of the computational basis unchanged, so all ancillary qubits are restored to $\ket{0}$. Since $\bigotimes_{u,v}L_{u,v}$ consists of the two-qubit phase gates in Eq.~\eqref{eq:circuit-Delta-H1}, it introduces the same phase as $\Delta_H^\uparrow$.
By Lemma~\ref{lem:fanout-standard}, $\Delta_H^\uparrow$ has the stated resource bounds. The same argument applies to $\Delta_K^\downarrow$.
\end{proof}

The following lemma shows the circuit implementations of $ {\widehat{T}_H}$ and $ {\widehat{S}_K}$ (Eqs.~\eqref{eq:PH} and \eqref{eq:PK}), the discretized counterparts of $T_H$ and $S_K$.
\begin{lemma}[Finite coordinate shears]
\label{lem:shear}
For any $H,K\in \R^{p\times p}$, $ {\widehat{T}_H}$ and $ {\widehat{S}_K}$ can be implemented by standard circuits of depth $O(q+\log(pq))$ and size $O(p^2q^2)$, using $O(p^2q^2)$ ancillary qubits. 
\end{lemma}
\begin{proof}
Recall that $ {\widehat{T}_H}=(D^{\otimes p}\otimes I^{\otimes p})^\dagger\Delta^\uparrow_H(D^{\otimes p}\otimes I^{\otimes p})$ and $ {\widehat{S}_K}=(I^{\otimes p}\otimes D^{\otimes p})^\dagger\Delta_K^\downarrow (I^{\otimes p}\otimes D^{\otimes p})$. Combining with Lemmas~\ref{lem:fourier} and \ref{lem:diag-finite}, the total depth, size and ancilla-count are $O(q)+O(\log(pq))=O(q+\log(pq))$, $p\cdot O(q^2)+O(p^2q^2)=O(p^2q^2)$ and $O(p^2q^2)$. 
\end{proof}

\paragraph{Implementation of $J_0$.} Recall that the signed swap $J_0(\vct x,\vct y)=(-\vct y,\vct x)$. We discretize $J_0$ using the coordinate grid introduced in Table \ref{tab:finite-grid-notation}. The coordinate labels are $\vct x=h(\vct r +\vct1/2)$ and $\vct y=h(\vct \ell +\vct1/2)$.  According to Eq.~\eqref{eq:finite-coordinate-negation}, $-\vct y=h(\vct\ell +\vct1/2)=-h((-\vct\ell-\vct 1)+\vct1/2)$. Then the discretized counterpart of $J_0$ is defined by
\begin{equation}
     {\widehat J_0}\ket{\vct r,\vct \ell}
       =\ket{-\vct \ell-\vct1,\vct r},\qquad \vct r,\vct\ell\in\mathcal{I}^p.
    \label{eq:shear-finite-signed-swap}
\end{equation}
The following lemma shows the circuit complexity of $\widehat J_0$.
\begin{lemma}\label{lem:finite-signed-swap}
 The $2pq$-qubit finite signed swap $\widehat J_0$ (Eq. \eqref{eq:shear-finite-signed-swap}) can be implemented by a standard circuit of depth $O(1)$ and size $O(pq)$, without ancillary qubits.
\end{lemma}
\begin{proof}
For any input state $\ket{\vct r,\vct \ell}$, swap the two $pq$-qubit blocks, using a circuit of depth $O(1)$ and size $O(pq)$. Then apply $X^{\otimes pq}$ to the first block, as in Eq.~\eqref{eq:finite-coordinate-negation}.. By Eq.~\eqref{eq:finite-coordinate-negation}, we apply $X^{\otimes pq}$ on the first $pq$ qubits and obtain $(X^{\otimes pq}\otimes I^{\otimes pq})\ket{\vct \ell,\vct r}=\ket{-\vct\ell-\vct1,\vct r}$.
The total depth and size are $O(1)$ and $O(pq)$.
\end{proof}

Recall that $F_{U}=T_{R_U}S_{-R_U^T}T_{R_U}J_0$. Then the discretized counterpart of $F_U$ is 
\begin{equation}\label{eq:dir-FU}
    \mathcal{P}_{\rm  {discrete}}= {\widehat{T}_{R_U}\widehat{S}_{-R_U^T}\widehat{T}_{R_U}\widehat J_0} 
\end{equation}
\begin{theorem}\label{thm:finiyr-coor-operation}
    The finite coordinate operation $\mathcal{P}_{\rm  {discrete}}$ (Eq.~\eqref{eq:dir-FU}) can be implemented by a standard circuit of depth $O(q+\log(pq))$ and size $O(p^2q^2)$, using $O(p^2q^2)$ ancillary qubits.
\end{theorem}
\begin{proof}
    The proof follows directly from Lemmas \ref{lem:shear} and \ref{lem:finite-signed-swap}.
\end{proof}

\subsection{Error analysis and resource bounds for approximate unitary synthesis}
\label{subsec:core-error}
Let
\begin{equation}
    B^2\ge2^{20}(n+20).
\end{equation}
Recall
\[
    \ell=4\cdot2^n, \quad
    Q=2B^2,\quad h=B^{-1},\quad \widehat h=\frac{2\pi}{Qh}=\frac{\pi}{B}.
\]

We first establish two bounds for the error analysis. The proofs are deferred to Appendix \ref{append:core-lattice}.
\begin{lemma}
\label{lem:core-lattice}
{ Let $A\in\mathbb C^{\ell\times\ell}$ be symmetric, with
$G=\operatorname{Re}A$ satisfying $\frac1{64}\Id\preceq G\preceq64\Id$ and let $\vct b\in\mathbb C^\ell\setminus\{\vct0\}$.
Define
\[
    f(\vct z)
       =\alpha(\vct b^T\vct z)e^{-\vct z^TA\vct z/2},
    \qquad \vct z\in\mathbb R^\ell,
\]
where $\alpha\in\mathbb C$ is a normalization constant such that $\int_{\mathbb R^\ell}|f(\vct z)|^2\,d\vct z=1$.
For each $0\le j<\ell$, choose a spacing
$a_j\in\{h,\widehat h\}$ and an offset $\sigma_j\in\mathbb R$.
Define the infinite grid points by
\[
    \vct z_{\vct k}
       =\bigl(a_0(k_0+\sigma_0),\ldots,
               a_{\ell-1}(k_{\ell-1}+\sigma_{\ell-1})\bigr)^T,
    \qquad \vct k\in\mathbb Z^\ell.
\]
Then the following bounds hold:
\begin{align}
    & \left|
        \left(\prod_{j=0}^{\ell-1}a_j\right)
        \sum_{\vct k\in\mathbb Z^\ell}
             |f(\vct z_{\vct k})|^2
        -1
        \right|
        \le18\ell e^{-B^2/128}.\label{eq:dir-norm-estimate}\\
        & \left(\prod_{j=0}^{\ell-1}a_j\right)
        \sum_{\substack{\vct k\in\mathbb Z^\ell\\
                        |z_{\vct k,i}|\ge R}}
             |f(\vct z_{\vct k})|^2
        \le6e^{-R^2/128},\qquad i\in\{0,1\ldots,\ell\},~\forall R\ge0.
        \label{eq:core-sampled-tail}
\end{align}}
\end{lemma}
{Note that the same tail bound in Eq.~\eqref{eq:core-sampled-tail} also holds if, for any $j\ne i$, the index
    $k_j$ is restricted to a finite subset of $\mathbb Z$, and $k_i$ still ranges over $\mathbb Z$.} 

Continuous transformations act on wavefunctions defined on a continuous domain, whereas finite-qubit circuits act on finite-dimensional amplitude vectors. Their outputs therefore belong to different spaces and cannot be compared directly. We introduce the sampling map $\mathsf S_{\vct a,\vct\sigma}$ to associate each smooth wavefunction with a finite-dimensional vector on a specified grid.
\begin{definition}[Sampling map]\label{def:sampling-map}
For spacing vector $\vct a=(a_0,\ldots,a_{\ell-1})^T\in\{h,\widehat h\}^{\ell}$ and offset vector $\vct \sigma=(\sigma_0,\ldots,\sigma_{\ell-1})^T\in\R^\ell$, define the sampling map
\begin{equation}
    \mathsf S_{\vct a,\vct\sigma}f
       =\left(\prod_{j=0}^{\ell-1}\sqrt{a_j}\right)
         \sum_{\vct k\in\mathcal I^\ell}
         f\bigl(\vct a\mathbin\odot(\vct k+\vct\sigma)\bigr)
         \ket{\vct k}
    \label{eq:core-mixed-samples}
\end{equation}
where $\odot$ denotes the entrywise multiplication.
\end{definition}

This allows us to compare the vector obtained by sampling first and then applying the circuit with the vector obtained by applying the continuous transformation first and then sampling on the corresponding output grid in the following Lemma \ref{lem:core-one-fourier}. The norm of their difference quantifies the discretization error.


Recall that the discrete and continuous one-coordinate quantum Fourier transform:
\begin{align}
    D\ket{r}&=\frac{1}{\sqrt{Q}}\sum_{k\in\mathcal{I}}e^{-2\pi\iu k(r+1/2)/Q}\ket{k},\qquad r\in\mathcal{I},\\
    \mathcal F\ket{ x}&=
    \frac{1}{(2\pi)^{1/2}}
    \int_{\mathbb R}
    e^{-\iu \zeta x}
    \ket{\zeta}\,d \zeta,\qquad x\in\R.
\end{align}
\begin{lemma}[Error of one-coordinate finite Fourier transform]
\label{lem:core-one-fourier}
Let $D_i$ and $\mathcal{F}_i$ denote the discrete and continuous one-coordinate Fourier transform acting on the coordinate $i$. Suppose $f$ and $\mathcal F_i f$ satisfy Lemma~\ref{lem:core-lattice}.
Changing coordinate $i$ from $(h,1/2)$ to $(\widehat h,0)$ gives
\begin{equation}
    \norm{D_i\mathsf S_{\vct a,\vct\sigma}f
          -\mathsf S_{\vct a',\vct\sigma'}\mathcal F_i f}
       \le12e^{-B^2/256}.
    \label{eq:core-one-fourier-error}
\end{equation}
where $\vct a,\vct a',\vct \sigma, \vct \sigma'$ satisfy
\begin{equation}
    (a_i,\sigma_i)=(h,1/2),
    \qquad
    (a_i',\sigma_i')=(\widehat h,0),
    \qquad
    (a_j',\sigma_j')=(a_j,\sigma_j)
    \quad \forall j\ne i.
\end{equation}
The analogous bound
holds for $D_i^\dagger$ and the inverse continuous transform.
\end{lemma}
The proof is shown in Appendix \ref{append:core-one-fourier}.
\begin{remark}[Why use different sampling maps for $f$ and $\mathcal{F}_if$?]
Consider a one-dimensional function $f$, sampled at
$x_r=h(r+\tfrac12)$ for $r\in\mathcal I$.
For the finite Fourier representation, we use Fourier modes with period $L=Qh$. Their frequencies satisfy
\[
e^{\iu\zeta(x+L)}=e^{\iu\zeta x}
\quad\Longleftrightarrow\quad
\zeta L\in2\pi\mathbb Z.
\]
Accordingly, we sample $\mathcal Ff$ at
\[
\zeta_k=\frac{2\pi k}{L}=\widehat h k,
\qquad
\widehat h=\frac{2\pi}{Qh},
\qquad k\in\mathcal I.
\]
The frequency grid therefore has spacing $\widehat h$
and zero offset. The half-cell offset of the position grid
is retained in the Fourier phase:
\[
D_{k,r}
=\frac{1}{\sqrt Q}e^{-\iu\zeta_kx_r}
=\frac{1}{\sqrt Q}
 e^{-2\pi\iu kr/Q}e^{-\pi\iu k/Q}.
\]
\end{remark}

Define sampling map $\mathsf S$ by
\begin{equation}\label{eq:samp-map-S}
    \mathsf S f=h^{\ell/2}\sum_{\vct k\in\mathcal{I}^\ell}f(h(\vct k+\frac{\vct 1}{2}))\ket{\vct k}.
\end{equation}
\begin{lemma}[Normalization of the sampled encoder]
\label{lem:red34-normalization}
For any $n$-qubit state $\ket{\psi}=\sum_{j=0}^{d-1}c_j\ket{j}$ with $\vct c=(c_0,c_1,\ldots,c_{d-1})^T\in\C^d$, define
\begin{equation}
       \ket{\Psi_Q(\vct c)} =\frac{1}{\sqrt2}\sum_{j=0}^{d-1}c_j(\ket{\vct\Phi_{j,Q}} -\iu\ket{\vct\Phi_{d+j,Q}}),
  \end{equation}
  where 
  \begin{equation}
      \ket{\vct\Phi_{t,Q}} =\ket{\vct e_Q}_t \bigotimes_{\substack{0\le\nu<\ell\\ \nu\ne t}} \ket{\vct g_Q}_\nu, \qquad 0\le t<\ell.
  \end{equation}
Let $f_{\vct b}$ be the wavefunction defined in Eq.~\eqref{eq:gaussian-linear-wavefunction}. Then 
\begin{equation}
 \mathsf S f_{V'V\vct c}
      =\rho_Q\ket{\Psi_Q(\vct c)},\qquad
 \rho_Q^2=2h^\ell\pi^{-\ell/2}Z_1Z_0^{\ell-1},
\label{eq:red34-normalization}
\end{equation}
where $\rho_Q>0$  satisfies
\begin{equation}
 \rho_Q^2\ge1-24\ell e^{-B^2/128}\ge\tfrac14.
\label{eq:red34-normalization-lower}
\end{equation}
\end{lemma}
The proof is deferred to Appendix \ref{append:red34-normalization}.

 Recall that $F_U$ can be decomposed as follows (Eq.~\eqref{eq:shear-factorization}):
\[F_U=T_{R_U}S_{-R_U^T}T_{R_U}J_0.\]
Moreover, each of $T_{R_U}$ and $S_{-R_U^T}$ decomposes into a Fourier transform, a diagonal phase operation, and the inverse Fourier transform; see Eqs.~\eqref{eq:TH-step1}--\eqref{eq:TH-step3}.

The following lemma shows that, in the continuous implementation described above, the wavefunctions before and after each transformation satisfy the hypotheses of Lemma~\ref{lem:core-lattice}, allowing us to apply that lemma in the error analysis.

{\begin{lemma}
\label{lem:red34-gaussian-closure}
Let $\vct c\in\mathbb C^d$ satisfy
$\norm{\vct c}=1$, and let $f_{\vct b}(\vct z)=\sqrt2\,(\vct b^T\vct z)G_\ell(\vct z)$ be the input wavefunction of $F_U$ for $\vct b=V'V\vct c$.
For each $M\in \mathcal{S}:=\{ \Id,\, J_0,\, T_{R_U}J_0,\,
S_{-R_U^T}T_{R_U}J_0,\, F_U \}$,
define
\[
f_M(\vct z) =f_{\vct b}(M^{-1}\vct z).
\]
Define a partial Fourier transform
\[
\mathcal F_S=\prod_{i\in S}\mathcal F_i,
\qquad
\mathcal F_{\varnothing}=\Id,\qquad S\subseteq\{0,\ldots,p-1\} \quad \text{or}\quad S\subseteq\{p,\ldots,\ell-1\}. 
\]
Then there exist $\alpha\in\mathbb C\setminus\{0\}$,
$\vct\beta\in\mathbb C^\ell\setminus\{0\}$, and
$A\in\mathbb C^{\ell\times\ell}$ such that
\[
(\mathcal F_Sf_M)(\vct w)
=
\alpha(\vct\beta^T\vct w)
e^{-\vct w^TA\vct w/2},
\]
where $\vct w=(\vct\zeta_S,\vct z_{T})$ and
$T=\{0,\ldots,\ell-1\}\setminus S$.
Moreover,
\[
\norm{\mathcal F_Sf_M}=1,
\qquad
A^T=A,
\qquad
\Id/64\preceq\operatorname{Re}A\preceq64\Id.
\]
\end{lemma}
The proof is deferred to Appendix \ref{append:red34-gaussian-closure}.
\begin{remark}
    According to Eqs.~\eqref{eq:TH-step1}, \eqref{eq:TH-step2} and \eqref{eq:TH-step3}, $T_{R_U}$ and $S_{-R_U^T}$ can be implemented by applying a partial Fourier transform, followed by a phase gate and then the inverse partial Fourier transform. Therefore, an intermediate state in the implementation of $T_{R_U}$ and $S_{-R_U^T}$ can be obtained by applying a partial Fourier transform to \(f_M\) for some \(M\in \mathcal{S}\). By the discussion above, such intermediate states satisfy the assumption of Lemma \ref{lem:core-lattice}.
\end{remark}

\begin{lemma}[Finite coordinate error on Gaussian-encoded states]
\label{lem:red34-coordinate-error}
Let $U\in\mathrm U(d)$ and $\vct c\in\mathbb C^d$ satisfy
$\norm{\vct c}=1$.  The finite coordinate operation $\mathcal{P}_{\rm  {discrete}}$ (Eq.~\eqref{eq:dir-FU}) satisfies
\begin{equation}
\norm{
\mathcal P_{\rm  {discrete}}\ket{\Psi_Q(\vct c)}
-\ket{\Psi_Q(U\vct c)}
}
\le72\ell e^{-B^2/256}
\le 100\ell e^{-B^2/256}.
\label{eq:red34-coordinate-error}
\end{equation}
The bound is uniform in $U$ and in the unit vector $\vct c$.
\end{lemma}

\begin{proof}
First, we specify the continuous and finite operations
used in the comparison.
Split each block Fourier transform into its $p$
one-coordinate transforms.
Treat each diagonal phase block and the signed swap
as a single operation.
This gives $m=6p+4$ operations in chronological order.
Let $A_j$ denote the $j$th finite operation and
$\mathcal B_j$ its continuous counterpart, so that
\[
\mathcal P_{\rm  {discrete}}=A_m\cdots A_1.
\]
Define the ideal continuous wavefunctions by
\[
f_0=f_{V'V\vct c},
\qquad
f_j=\mathcal B_jf_{j-1},
\qquad 1\le j\le m.
\]
Let $\mathsf S_j$ be the sampling map for the position
and frequency variables of $f_j$.
Both endpoints use the position grid, so $\mathsf S_0=\mathsf S_m=\mathsf S$, where $\mathsf S$ is defined in Eq.~\eqref{eq:samp-map-S}.
The continuous identity $\mathcal P_{F_U}C=CU$ gives
\[
f_m=f_{V'VU\vct c}.
\]

Second, we show that the diagonal phases and signed swap
have zero sampling error.
For the $T_H$ with $H=R_U$, the mixed-grid points are $\vct\zeta=\widehat h\vct k$ and $\vct y=h(\vct \ell+\tfrac12\vct1)$.
Using $h\widehat h=2\pi/Q$, we obtain
\begin{align}
e^{-\iu\vct\zeta^TH\vct y}= e^{-\iu h\widehat h \vct k^TH(\vct \ell+\tfrac12\vct1)}=\exp\left[
-\frac{2\pi\iu}{Q}
\vct k^TH(\vct \ell+\tfrac12\vct1)
\right].
\end{align}
This is exactly the multiplier defining $\Delta_H^\uparrow$ (Eq.~\eqref{eq:Delta-H}). For the $S_K$ with $K=-R_U^T$, the mixed-grid points are $\vct x=h(\vct r+\tfrac12\vct1),$ and $
\vct\eta=\widehat h\vct k.$
Thus
\begin{equation}
e^{-\iu\vct\eta^TK\vct x} = \exp\left[ -\frac{2\pi\iu}{Q}
\vct k^TK(\vct r+\tfrac12\vct1) \right],    
\end{equation}
which is the multiplier defining $\Delta_K^\downarrow$ (Eq.~\eqref{eq:Delta-K}).
The input and output grids of each phase operation are
identical. Multiplying the sampled amplitudes therefore
agrees exactly with sampling the multiplied wavefunction.

For the signed swap, recall that $\vct t_{\vct r}=h(\vct r+\tfrac12\vct1)$ for $r\in\mathcal{I}^p$.
Then we have
\begin{align}
     {\widehat{J}_0}\mathsf S f=& {\widehat{J}_0} h^p\sum_{\vct r,\vct \ell\in\mathcal I^p}f(\vct t_{\vct r},\vct t_{\vct \ell}) \ket{\vct r,\vct \ell}\\
      = & h^p\sum_{\vct r,\vct \ell\in\mathcal I^p} f(\vct t_{\vct r},\vct t_{\vct \ell})\ket{-\vct \ell-\vct1,\vct r}= h^p\sum_{\vct r',\vct \ell'\in\mathcal I^p} f(\vct t_{\vct \ell'},-\vct t_{\vct r'}) \ket{\vct r',\vct \ell'}\\
      = & \mathsf S(f\circ J_0^{-1}),
\end{align}

Since $\ell=2p$, position sampling gives
\[
\mathsf S f
=
h^p\sum_{\vct r,\vct \ell\in\mathcal I^p}
f(\vct t_{\vct r},\vct t_{\vct \ell})
\ket{\vct r,\vct \ell}.
\]
Applying $ {\widehat{J}_0}$ yields
\[
 {\widehat{J}_0}\mathsf S f
=
h^p\sum_{\vct r,\vct \ell\in\mathcal I^p}
f(\vct t_{\vct r},\vct t_{\vct \ell})
\ket{-\vct \ell-\vct1,\vct r}.
\]
The substitution
\[
\vct r'=-\vct \ell-\vct1,
\qquad
\vct \ell'=\vct r
\]
is a bijection on $\mathcal I^p\times\mathcal I^p$.
Using $\vct t_{-\vct r'-\vct1}=-\vct t_{\vct r'}$,
we obtain
\[
\begin{aligned}
 {\widehat{J}_0}\mathsf S f=
h^p\sum_{\vct r',\vct \ell'\in\mathcal I^p}
f(\vct t_{\vct \ell'},-\vct t_{\vct r'})
\ket{\vct r',\vct \ell'}
=\mathsf S(f\circ J_0^{-1}),
\end{aligned}
\]
because $J_0^{-1}(\vct x,\vct y)=(\vct y,-\vct x)$.
Hence the finite signed swap also agrees exactly with
sampling its continuous action.

Consequently, whenever step $j$ is a diagonal phase
or the signed swap,
\[
A_j\mathsf S_{j-1}f_{j-1}
=\mathsf S_jf_j.
\]

Third, consider a one-coordinate Fourier step.
Lemma~\ref{lem:red34-gaussian-closure} shows that
$f_{j-1}$ and $f_j$ satisfy the assumptions required by
Lemma~\ref{lem:core-lattice}.
The latter lemma therefore gives
\[
\norm{
A_j\mathsf S_{j-1}f_{j-1}
-\mathsf S_jf_j
}
\le12e^{-B^2/256}.
\]
This estimate applies to both forward and inverse transforms.
Each shear contains $p$ forward and $p$ inverse transforms.
There are therefore $6p$ Fourier steps in total.

Fourth, we accumulate these local errors.
For each $j$, adding and subtracting
$A_j\mathsf S_{j-1}f_{j-1}$ gives
\begin{equation}
A_j\cdots A_1\mathsf S_0f_0-\mathsf S_jf_j=
A_j\bigl(A_{j-1}\cdots A_1\mathsf S_0f_0
-\mathsf S_{j-1}f_{j-1}\bigr)+\bigl(A_j\mathsf S_{j-1}f_{j-1}-\mathsf S_jf_j
\bigr).
\end{equation}
Every empty product is interpreted as the identity.
Iterating this equality yields
\begin{equation}
A_m\cdots A_1\mathsf S_0f_0-\mathsf S_mf_m=\sum_{j=1}^m A_m\cdots A_{j+1}
\bigl(A_j\mathsf S_{j-1}f_{j-1}-\mathsf S_jf_j\bigr).
\label{eq:red34-telescope}
\end{equation}
By the triangle inequality,
\begin{align}
\norm{A_m\cdots A_1\mathsf S_0f_0-\mathsf S_mf_m} &\le
\sum_{j=1}^m\norm{A_m\cdots A_{j+1}\bigl(A_j\mathsf S_{j-1}f_{j-1}-\mathsf S_jf_j\bigr)}\\
&= \sum_{j=1}^m \norm{A_j\mathsf S_{j-1}f_{j-1}-\mathsf S_jf_j}\\
&\le6p\cdot12e^{-B^2/256}\\
&=72p\,e^{-B^2/256}.
\end{align}
The last inequality uses the $6p$ Fourier-step estimates;
the diagonal phases and signed swap contribute zero.
Then we have
\begin{equation}
  \norm{
\mathcal P_{\rm  {discrete}}\mathsf S f_{V'V\vct c}
-\mathsf S f_{V'VU\vct c}
}
\le72p\,e^{-B^2/256}.  
\end{equation}
Each local difference is evaluated on the ideal wavefunction
$f_{j-1}$, rather than on an already perturbed circuit state.
The sampled vectors $\mathsf S_jf_j$ need not be normalized.

Finally, since $U$ is unitary,
$\norm{U\vct c}=\norm{\vct c}=1$.
Lemma~\ref{lem:red34-normalization} therefore gives
\[
\mathsf S f_{V'V\vct c}
=\rho_Q\ket{\Psi_Q(\vct c)},
\qquad
\mathsf S f_{V'VU\vct c}
=\rho_Q\ket{\Psi_Q(U\vct c)},
\qquad
\rho_Q\ge\frac12.
\]
The same factor occurs at both endpoints. Hence
\[
\begin{aligned}
&\rho_Q
\norm{
\mathcal P_{\rm  {discrete}}\ket{\Psi_Q(\vct c)}
-\ket{\Psi_Q(U\vct c)}
}=\norm{ \mathcal P_{\rm  {discrete}}\mathsf S f_{V'V\vct c} -\mathsf S f_{V'VU\vct c}
} \le72p\,e^{-B^2/256}.
\end{aligned}
\]
Dividing by $\rho_Q$ and using $\ell=2p$, we obtain
\begin{equation}
\norm{\mathcal P_{\rm  {discrete}}\ket{\Psi_Q(\vct c)}
-\ket{\Psi_Q(U\vct c)}}\le\frac{72p}{\rho_Q}e^{-B^2/256} \le100\ell e^{-B^2/256}.
\end{equation}
This proves Eq.~\eqref{eq:red34-coordinate-error}.
\end{proof}

\paragraph{Full-output accuracy and quantum resources.}
We combine the finite Gaussian encoder with the circuit implementing the finite coordinate operation. The Gaussian profiles and the elementary gates are implemented exactly. The approximation error therefore comes only from the finite coordinate grid.

\begin{theorem}[Restatement of Theorem \ref{thm:core-construction}]
\label{thm:core-full-resources}
Let $n\ge2$ and $0<\epsilon\le1/2$. Any $n$-qubit unitary $U\in\mathrm{U}(2^n)$ can be $\epsilon$-approximately implemented by a circuit of 
\begin{align*}
    & \text{depth:~~ $O(n+\log \log(1/\epsilon))$},\\ 
    & \text{size:~~ $O(4^n \log^2 (n+\log(1/\epsilon))+2^n\log(1/\epsilon))$,}
\end{align*}
 using $m=O(4^n \log^2 (n+\log(1/\epsilon))+2^n\log(1/\epsilon))$ ancillary qubits. 

More specifically, there is a circuit $T_{\rm  {discrete}}$ on $n+m$ qubits satisfying
\begin{equation}
    \norm{T_{\rm  {discrete}}J_{\rm  {discrete}}-J_{\rm  {discrete}}U}_{\op}
       < \epsilon,\quad \text{where}\qquad J_{\rm  {discrete}}\ket{\vct \psi}
       =\ket{\vct \psi}\ket{0^m},
\label{eq:core-strong-accuracy}
\end{equation}
\end{theorem}

\begin{proof}
First, we construct the complete circuit.
Let $E$ denote the unitary circuit implementing $C_Q$ (Theorem \ref{thm:CQ}). By Eq.~\eqref{eq:CQ-encode} and linearity,
\begin{equation}
    EJ_{\rm  {discrete}}\ket{\vct \psi}
       =\ket{\Psi_Q(\vct c)}\ket{0^w},\qquad \text{where}\quad\ket{\vct\psi}=\sum_{j=0}^{d-1}c_j\ket j,\quad \text{with}\quad \norm{\vct c}=1
\label{eq:core-literal-encoding}
\end{equation}
Let $\mathsf P_{\rm  {discrete}}$ denote the circuit implementing
$\mathcal P_{\rm  {discrete}}$.
The phase circuits in Lemma~\ref{lem:diag-finite} restore
their helpers to zero for every coordinate input.
Define
\[
    T_{\rm  {discrete}}=E^\dagger\mathsf P_{\rm  {discrete}}E.
\]

Second, we bound the full error.
For any input state $\ket{\vct\psi}$, we have
\begin{align}
\norm{T_{\rm  {discrete}}J_{\rm  {discrete}}\ket{\vct \psi}
                 -J_{\rm  {discrete}}U\ket{\vct \psi}}=&
       \norm{\mathsf P_{\rm  {discrete}}EJ_{\rm  {discrete}}\ket{\vct \psi}
                 -EJ_{\rm  {discrete}}U\ket{\vct \psi}}\\
                 = &\norm{\bigl(\mathcal P_{\rm  {discrete}}\ket{\Psi_Q(\vct c)}
                 -\ket{\Psi_Q(U\vct c)}\bigr)\otimes\ket{0^w}}\\
                 = &\norm{\mathcal P_{\rm  {discrete}}\ket{\Psi_Q(\vct c)}
                 -\ket{\Psi_Q(U\vct c)}}.
\end{align}
Then the total error bound is
\begin{equation}
    \norm{T_{\rm  {discrete}}J_{\rm  {discrete}}-J_{\rm  {discrete}}U}_{\op}=\sup_{\norm{\vct c}=1}\norm{\mathcal P_{\rm  {discrete}}\ket{\Psi_Q(\vct c)}-\ket{\Psi_Q(U\vct c)}}\le 100\ell e^{-B^2/256}.
\label{eq:core-full-output-transfer}
\end{equation}
The inequality follows from
Lemma~\ref{lem:red34-coordinate-error}. The circuit complexity and error bound of $T_{\rm  {discrete}}$ are determined by the parameter $B$. Choose $B=2^{\lceil\frac12\log_2(2^{20}r_0)\rceil}$ with $r_0=n+\lceil\log(1/\epsilon)\rceil+20$.
The definition of $B$ gives $2^{20}r_0\le B^2<2^{22}r_0$, which yields
\begin{equation}
    \text{Eq.~\eqref{eq:core-full-output-transfer}}\le100\cdot2^{n+2}e^{-4096r_0}\le2^{n+9-4096r_0}\le2^{-40-\lceil\log(1/\epsilon)\rceil-n/2}<\epsilon.
\end{equation}
This proves Eq.~\eqref{eq:core-strong-accuracy}.

Third, we count the quantum resources.
Section~\ref{subsec:finite-core} implements the encoder in
depth $O(n+q)$ and size $O(\ell Q)$, using
$O(\ell Q)$ preparation-workspace qubits.
Reversing the circuit gives the decoder with the same costs.
Recall that $d=2^n,~p=2d,~\ell=2p$. According to Theorems \ref{thm:CQ} and \ref{thm:finiyr-coor-operation}, the depth, size and ancilla-count to implement $U\in\mathrm{U}(d)$ are 
\begin{align*}
    \text{depth:}\quad   &2O(n+q)+O(q+\log(pq))
         =O(n+q),\\
    \text{size:}\quad
       &2O(Qd)+O(p^2q^2)
         =O(d Q+p^2q^2),\\
      \text{ancilla-count:}\quad &(q\ell-n)+O(Qd)+O(p^2q^2)
         =O(d Q+p^2q^2).
\end{align*}
Finally, the bounds on $B^2$ imply $Q=2B^2=O(r_0)$ and $q=\log Q=O(\log r_0)$. Substituting $Q$ and $q$ into above circuit complexities complete the proof.
\end{proof}
\begin{remark}
With $\Pi=J_{\rm  {discrete}}J_{\rm  {discrete}}^\dagger$, we can verify that
\begin{align}
    \norm{(\Id-\Pi)T_{\rm  {discrete}}J_{\rm  {discrete}}}_{\op}
       &=
       \norm{(\Id-\Pi)
       (T_{\rm  {discrete}}J_{\rm  {discrete}}-J_{\rm  {discrete}}U)}_{\op}\\&\le
       \norm{T_{\rm  {discrete}}J_{\rm  {discrete}}-J_{\rm  {discrete}}U}_{\op}
        \le\epsilon.
\end{align}
Consequently, the probability of a nonzero ancillary output
is at most $\epsilon^2$ for every input state.
\end{remark}

\section{Reducing the size and ancilla-count for linear-depth approximate unitary synthesis }
\label{sec:spectral}
In this section, we reduce the size and ancilla-count of the unitary approximation circuit to optimal while keeping the circuit depth linear with respect to $n$.
We improve the circuits for $\Delta_H^\uparrow$ and
$\Delta_K^\downarrow$ without changing the Gaussian encoder, the
Fourier transforms, and the signed swap from Section \ref{sec:geometry}.

To reduce circuit size and ancilla-count, we approximately implement \(\Delta_H^\uparrow\) and \(\Delta_K^\downarrow\) as opposed to their exact implementation in Section \ref{sec:geometry}. We divide their phase terms into strong and weak ones according to their coefficient magnitudes $|\theta_{ij}|$. The orthogonality of $H$ or $K$ bounds the number of strong terms, which we implement exactly using phase and controlled-phase gates. Weak terms are grouped by row. State preparation encodes each group’s normalized phase sum as an amplitude, and a fixed-degree polynomial approximates the corresponding phase factor at the target accuracy. The groups run in parallel, while the state-preparation depth–ancilla tradeoff reduces their ancilla-count. With the grid and accuracy used in this section, the resulting circuit has depth \(O(n)\), size \(O(4^n)\), and ancillary count \(O(4^n/n)\), which are all optimal.

Recall that $d=2^n$, $t=n+1$, $p=2d=2^t$, and $\ell=2p$.
Each coordinate register has $q$ qubits, with $Q=2^q$ and labels
in $\mathcal I$. For an upper shear, let
$\mathsf X_0,\ldots,\mathsf X_{p-1}$ store the frequency labels
$\vct k$ and let $\mathsf Y_0,\ldots,\mathsf Y_{p-1}$ store the coordinate labels
$\vct l$. Define
\begin{equation}
 a(k)=\frac{2k}{Q},\qquad b(l)=\frac{2l+1}{Q},\qquad
 \theta_{ij}=-\frac{\pi Q}{2}H_{ij}.
 \label{eq:s4c-couplings}
\end{equation}
Then $|a(k)|,|b(l)|\le1$, and the required action is
\begin{equation}
 \Delta_H^\uparrow\ket{\vct k,\vct l}
 =\exp\!\left(\iu\sum_{i,j=0}^{p-1}
       \theta_{ij}a(k_i)b(l_j)\right)\ket{\vct k,\vct l}.
 \label{eq:s4c-phase-target}
\end{equation}
This is equivalent to Eq.~\eqref{eq:Delta-H}, since
$\theta_{ij}a(k_i)b(l_j)=-2\pi H_{ij}k_i(l_j+1/2)/Q$.
For $\Delta_H^\downarrow$, exchange the roles of the two coordinate
blocks. 

\subsection{Reducing the size and ancilla-count for diagonal phase gates}

Call the $(i,j)$ entry \textit{strong} if $|\theta_{ij}|>p^{-1/4}/2$, and weak
otherwise. 

\paragraph{Implementation of the strong terms.} Let $N_{\rm str}$ be the number of strong entries.
Since $H$ is orthogonal,
\begin{align}
    &\sum_{i,j}|H_{ij}|^2=p,\\
    &N_{\rm str}\frac{p^{-1/2}}{\pi^2Q^2} \le\sum_{i,j}|H_{ij}|^2.
\end{align}
Consequently,
\begin{equation}
 N_{\rm str}\le\min\{p^2,\pi^2p^{3/2}Q^2\}.
 \label{eq:s4c-strong-count}
\end{equation}

\begin{lemma}[Circuit for the strong terms]
\label{lem:s4c-strong}
Define the $2pq$-qubit diagonal unitary
\begin{equation}\label{eq:DeltaH-str}
 \Delta_{H,\rm str}^\uparrow\ket{\vct k,\vct l}
 =\exp\!\left(\iu\sum_{(i,j)\,\mathrm{strong}}
      \theta_{ij}a(k_i)b(l_j)\right)\ket{\vct k,\vct l}.
      \end{equation}
It can be implemented exactly by a standard circuit of depth
$O(\log(pq))$ and size $O(N_{\rm str}q^2)$, using
$O(N_{\rm str}q^2)$ ancillary qubits.
\end{lemma}
\begin{proof}
First, write the signed coordinate labels in two's complement:
\[
k_i=\sum_{\alpha=0}^{q-1}w_\alpha x_{i,\alpha},
\qquad
l_j=\sum_{\beta=0}^{q-1}w_\beta y_{j,\beta},
\]
where $x_{i,\alpha},y_{j,\beta}\in\{0,1\}$ and
\[
w_\alpha=
\begin{cases}
2^\alpha,&0\le\alpha<q-1,\\
-2^{q-1},&\alpha=q-1.
\end{cases}
\]
For each strong pair $(i,j)$, we have
\begin{equation}
\theta_{ij}a(k_i)b(l_j)
=
\sum_{\alpha=0}^{q-1}
\frac{2\theta_{ij}w_\alpha}{Q^2}x_{i,\alpha}
+
\sum_{\alpha,\beta=0}^{q-1}
\frac{4\theta_{ij}w_\alpha w_\beta}{Q^2}
x_{i,\alpha}y_{j,\beta}.
\label{eq:s4c-strong-bits}
\end{equation}
Thus its phase can be implemented using $q$ single-qubit
phase gates and $q^2$ controlled-phase gates.

Second, assign ancillary qubit only to the coordinate qubits used
by these gates. For each strong pair $(i,j)$, assign one ancillary qubit $\ket{0}$ to each of its $q$ single-qubit phase gates and $\ket{00}$ to each of its $q^2$ controlled-phase gates.
An ancillary qubit stores only the value of the coordinate bit used by its assigned gate.
For example, a controlled-phase gate on
$(\mathsf X_{i,\alpha},\mathsf Y_{j,\beta})$ uses $\ket{00}$ initialized by
\[
\ket{x_{i,\alpha},y_{j,\beta}}\ket{0,0}
\longmapsto
\ket{x_{i,\alpha},y_{j,\beta}}
\ket{x_{i,\alpha},y_{j,\beta}}.
\]
The total number of ancillary qubits is therefore
\[
N_{\rm str}(q+2q^2).
\]
The fan-out gate can be unitilized to make copies of a qubit.
For each strong pair involving row $i$, the qubit
$\mathsf X_{i,\alpha}$ participates in one single-qubit phase gate
and $q$ controlled-phase gates. Hence it has at most $p(q+1)$ assigned ancillary qubits. Similarly, each $\mathsf Y_{j,\beta}$ has at most $pq$ assigned ancillary qubits.
The fan-out gates for different original bits act on disjoint
qubits and run in parallel. According to Lemma \ref{lem:fanout-standard}, Their total depth and size of these fan-out gates are
$O(\log(pq))$ and $N_{\rm str}(q+2q^2)=O(N_{\rm str}q^2)$, respectively. 

Third, apply each phase gate in
Eq.~\eqref{eq:s4c-strong-bits} to its assigned ancillary qubits.
Different gates use disjoint ancillary qubits, so all these phase
gates can be applied in one layer.
Their combined phase on a coordinate basis state is
\[
\exp\!\left(
\iu\sum_{(i,j)\,\mathrm{strong}}
\theta_{ij}a(k_i)b(l_j)
\right).
\]
Finally, reverse all fan-out gates.
The fan-out gates and their inverse each use
$O(N_{\rm str}q^2)$ gates. The phase layer has depth $1$ and uses $N_{\rm str}(q+q^2)$ gates. This completes the proof.
\end{proof}

\paragraph{Implementation of the weak terms.} Choose
\begin{equation}
 L=2^{\lfloor\min\{t/16,4\log_2(t+1)\}\rfloor},\qquad
 G=p^2/L,\qquad \kappa=\log_2(2L).
 \label{eq:s4c-groups}
\end{equation}
For each row $i$ of the coefficient matrix $(\theta_{ij})$,
partition its column indices into consecutive blocks
$\mathcal B$ of $L$ indices. Namely, the size of $\mathcal{B}$ is $L$. There are $G$ row-block pairs. A group $g=(i,\mathcal B)$ contains only the weak entries of this row block. In formulas for
that group, set the strong coefficients to zero. Define
\begin{equation}
 f_g(\vct k,\vct l)
   =a(k_i)\sum_{j\in\mathcal B}\theta_{ij}b(l_j),\qquad
 \tau_g=\sum_{j\in\mathcal B}|\theta_{ij}|,\qquad
 A_g(\vct k,\vct l)=f_g(\vct k,\vct l)/\tau_g.
 \label{eq:s4c-group-functions}
\end{equation}
Omit groups with $\tau_g=0$. Define the corresponding operators $f_g$ and
$A_g$ satisfying
\begin{equation}
    f_g\ket{\vct k,\vct l}=f_g(\vct k,\vct l)\ket{\vct k,\vct l}, \qquad A_g\ket{\vct k,\vct l}=A_g(\vct k,\vct l)\ket{\vct k,\vct l}\qquad \vct k,\vct l\in \mathcal{I}^p.
\end{equation}
Neither operator is unitary.
Since $|a|,|b|\le1$ and $L\le p^{1/16}$,
\begin{equation}
 |A_g(\vct k,\vct l)|\le1,\qquad
 \tau_g\le\tfrac12Lp^{-1/4}\le\tfrac12p^{-3/16}.
 \label{eq:s4c-group-small}
\end{equation}
The required diagonal phase is therefore
\begin{equation}
 \Delta_H^\uparrow
   =\Delta_{H,\rm str}^\uparrow\prod_g e^{\iu f_g}.
 \label{eq:s4c-phase-product}
\end{equation}
We next give the gates needed to implement each factor
$e^{\iu f_g}$ and show how the groups run in parallel.

\begin{lemma}[Read operator]
\label{lem:s4c-read}
Let $L$ be a power of two and $\kappa=\log_2(2L)$.
Let $\mathsf Y=(\mathsf Y_0,\ldots,\mathsf Y_{L-1})$, $\mathsf Z=(\mathsf Z_0,\ldots,\mathsf Z_{L-1})$ and $\mathsf A=(\mathsf A_0,\ldots,\mathsf A_{L-1})$ be registers for data, address and workspace, respectively.
For $0\le j,r<L$, let $\mathsf Y_j,\mathsf Z_r$ be $q$-qubit
registers and let $\mathsf A_r$ contain $\log_2L$ qubits.
Define the transformation $\mathcal R$ by
\begin{equation}
 \mathcal R\ket{\vct Y}_{\mathsf Y}\ket{\vct A}_{\mathsf A}\ket{\vct Z}_{\mathsf Z}
  =\ket{\vct Y}_{\mathsf Y}\ket{\vct A}_{\mathsf A}\bigotimes_{r=0}^{L-1}\ket{Z_r\mathbin\oplus Y_{A_r}}.
 \label{eq:s4c-read}
\end{equation}
Then $\mathcal R$ can be implemented by a standard circuit
of depth $O((q+\kappa)\kappa^2)$ and size $O(L(q+\kappa)\kappa^2)$, using $O(L(q+\kappa)\kappa^2)$ ancillary qubits.
\end{lemma}
\begin{proof}
If $L=1$, apply $q$ parallel CNOTs from $\mathsf Y_0$ to
$\mathsf Z_0$. Suppose $L\ge2$ and write $N=2L$.

First, use $X$ and CNOT gates to initialize $N$ additional
registers as
\[
 \mathsf R_j=(j,0,j,Y_j),\qquad
 \mathsf R_{L+r}=(A_r,1,L+r,0^q),\qquad 0\le j,r<L.
\]
The four entries store an address, one qubit distinguishing a
coordinate word from a query, an original position, and a
$q$-bit word. The original $\mathsf Y,\mathsf A,\mathsf Z$
registers are unchanged. Each $\mathsf R_i$ uses $O(q+\kappa)$
qubits.

Second, sort the $\mathsf R_i$ registers by their first three
entries, in lexicographic order. A coordinate word with address
$j$ now precedes all queries with address $j$. The position
entries ensure distinct comparison keys. Use a bitonic sorting
circuit: on $2^m$ registers, sort its two halves in opposite
directions and merge them in $m$ layers. Recursion gives
$m(m+1)/2$ layers of disjoint comparisons, here with $m=\kappa$.
For each comparison, scan the $O(\kappa)$ key bits, retain
the comparison results, and conditionally swap the two entire
registers. Computing each Boolean step into fresh ancillary qubits and
serializing the swaps gives size, depth, and space
$O(q+\kappa)$ per comparison. All comparison ancillas and swap
flags are kept until the sorting circuit is reversed.

Third, compute the coordinate word belonging to each sorted
query. At sorted position $i$, initialize $f_i^{(0)}=1$ for
a coordinate-word register and $f_i^{(0)}=0$ for a query;
let $v_i^{(0)}$ be its $q$-bit fourth entry. At step
$k=0,\ldots,\kappa-1$, for $i\ge2^k$ compute
\[
 \begin{aligned}
 f_i^{(k+1)}&=f_i^{(k)}\lor f_{i-2^k}^{(k)},\\
 v_i^{(k+1)}&=
 \begin{cases}
 v_i^{(k)},&f_i^{(k)}=1,\\
 v_{i-2^k}^{(k)},&f_i^{(k)}=0.
 \end{cases}
 \end{aligned}
\]
For $i<2^k$, copy the preceding values. All outputs are written
into new ancillas and all earlier values are retained.
After step $k$, these values identify the last coordinate word
among the preceding $2^{k+1}$ positions, including position $i$.
Thus $v_i^{(\kappa)}$ is the correct word for every query at
position $i$. Each update costs $O(q+1)$ gates, depth, and
ancillas. At one step, the pairs $i-2^k,i$ form disjoint paths,
so their updates can be divided into two rounds of disjoint
pairs, followed by one round of copies. This gives total size
and workspace $O(N(q+1)\kappa)$ and depth $O((q+1)\kappa)$.

Fourth, copy the $v_i^{(\kappa)}$ into $N$ new $q$-qubit
registers $\mathsf B_i$. Apply the saved sorting switches in
reverse order to these $\mathsf B_i$ alone. They now follow
the original order of the $\mathsf R_i$, so
\[
 \mathsf B_{L+r}=Y_{A_r},\qquad 0\le r<L.
\]
Call the four steps so far $\mathcal U_{\rm read}$. They have
the form
\[
 \ket{\vct Y,\vct A,\vct Z,0_W}
 \longmapsto
 \ket{\vct Y,\vct A,\vct Z,w(\vct Y,\vct A),
                    \vct B(\vct Y,\vct A)},
\]
where $w$ contains all retained values and comparison ancillas.
Neither $w$ nor $\vct B$ depends on $\vct Z$.

Finally, apply CNOTs from $\mathsf B_{L+r}$ to $\mathsf Z_r$
and reverse $\mathcal U_{\rm read}$. All additional registers
return to zero, whereas $Z_r$ has changed to $Z_r\oplus Y_{A_r}$.
In particular,
\begin{equation}
 \mathcal R=\mathcal U_{\rm read}^{\dagger}
   \left(\prod_{r=0}^{L-1}\prod_{\alpha=0}^{q-1}
       \operatorname{CNOT}_{\mathsf B_{L+r,\alpha},\mathsf Z_{r,\alpha}}
   \right)\mathcal U_{\rm read}.
\end{equation}
When reversing a comparison, undo its register swap before
erasing its comparison ancillas. The original comparison inputs
are then restored, so every ancilla is erased exactly.

Initialization and copying cost $O(N(q+\kappa))$ size and
constant depth. Sorting costs $O(N(q+\kappa)\kappa^2)$ size
and work and $O((q+\kappa)\kappa^2)$ depth. Routing the answer
words costs $O(Nq\kappa^2)$ size and $O(q\kappa^2)$ depth.
Including the computations above and their inverses gives
\begin{align*}
 \text{depth:}\quad
 &2O((q+\kappa)\kappa^2)+1=O((q+\kappa)\kappa^2),\\
 \text{size:}\quad
 &2O(N(q+\kappa)\kappa^2)+Lq=O(L(q+\kappa)\kappa^2),\\
 \text{ancilla-count:}\quad&O(L(q+\kappa)\kappa^2).
\end{align*} 
\end{proof}

For $u\in[-1,1]$, define the single-qubit gate
\begin{equation}
 B_u=\begin{pmatrix}u&\sqrt{1-u^2}\\
                    \sqrt{1-u^2}&-u\end{pmatrix}.
 \label{eq:s4c-scalar-gate}
\end{equation}

\begin{lemma}[Circuit for a weighted coordinate sum]
\label{lem:s4c-sum}
Let $g=(i,\mathcal B)$ be a group with $\tau_g>0$.
Then there exists a quantum circuit $\mathcal V_g$ satisfying
\begin{equation}
 (\bra{\vct k,\vct l}\otimes\bra{0_{\rm aux}})
       \mathcal V_g
       (\ket{\vct k,\vct l}\otimes\ket{0_{\rm aux}})
       =A_g(\vct k,\vct l),
 \label{eq:s4c-sum-amplitude}
\end{equation}
where 
\begin{equation}\label{eq:Vg-aux}
\ket{0_{\rm aux}} := \ket{ 0^{\log (L)}}_{\mathsf I_g} \ket0_{\mathsf S_g} 
\ket0_{\mathsf u_g}\ket0_{\mathsf v_g}\ket{0^{m_g}}_{\mathsf H_g}.
\end{equation}
Here $m_g$, the size of register $\mathsf H_g$, can be chosen arbitrarily.
When a control qubit disables $\mathcal V_g$, the circuit is the identity on every input of
these auxiliary registers.
\end{lemma}

\begin{proof}
Let $j_0=\min\mathcal B$, so that $\mathcal B=\{j_0,j_0+1,\ldots,j_0+L-1\}$.
Fix a computational-basis input $\ket{\vct k,\vct l}_{\mathsf X\mathsf Y}=\left(\bigotimes_{\nu=0}^{p-1}\ket{k_\nu}_{\mathsf X_\nu}\right)\otimes\left(\bigotimes_{\nu=0}^{p-1}\ket{l_\nu}_{\mathsf Y_\nu}\right)$ for any $\vct k,\vct l\in\mathcal I^p$.

For the fixed group $g=(i,\mathcal B)$, the required value is
$A_g(\vct k,\vct l)=\frac{a(k_i)}{\tau_g} \sum_{j\in\mathcal B}\theta_{ij}b(l_j)$, with $\tau_g=\sum_{j\in\mathcal B}|\theta_{ij}|>0$.
As in the definition of the group, strong coefficients are
set to zero in these sums.
Only $\mathsf X_i$ and the registers
$(\mathsf Y_j)_{j\in\mathcal B}$ are needed; all other
coordinate registers undergo the identity.

The ancillary registers retained in the definition of
$\mathcal V_g$ are $\mathsf I_g,\mathsf S_g,\mathsf u_g,\mathsf v_g$ and $\mathsf H_g$.
Their joint initial state is
\begin{equation}
\ket{0_{\rm aux}} = \ket{0^{\log (L)}}_{\mathsf I_g} \ket0_{\mathsf S_g} 
\ket0_{\mathsf u_g}\ket0_{\mathsf v_g}\ket{0^{m_g}}_{\mathsf H_g}.
\end{equation}
Thus $\ket{0_{\rm aux}}$ contains
$\log_2L+3+m_g$ qubits.
Additional registers used for coordinate reading are
initialized and returned to zero, and are not included in
$\ket{0_{\rm aux}}$.
All unspecified tensor factors below are identities.

First, choose an exact state-preparation circuit $P_g$
on $\mathsf I_g \mathsf S_g \mathsf H_g$ satisfying
\begin{equation}
\begin{aligned}
&P_g \ket{0^{\log(L)}}_{\mathsf I_g} \ket0_{\mathsf S_g}\ket{0^{m_g}}_{\mathsf H_g}=
\sum_{j\in\mathcal B} \sqrt{\frac{|\theta_{ij}|}{\tau_g}}\, \ket{j -j_0}_{\mathsf I_g} \ket{s_{ij}}_{\mathsf S_g}\ket{0^{m_g}}_{\mathsf H_g},
\end{aligned}
\label{eq:s4c-weight-preparation}
\end{equation}
where $ (-1)^{s_{ij}}=\operatorname{sgn}(\theta_{ij})$ if $\theta_{ij}\ne0$, and $s_{ij}=0$ if $\theta_{ij}=0$.
The latter choice does not affect the state because its coefficient is zero.
The state in Eq.~\eqref{eq:s4c-weight-preparation} is
normalized since $\sum_{j\in\mathcal B}\frac{|\theta_{ij}|}{\tau_g}=1$.
Here $j-j_0$ is the zero-based index of column $j$ within $\mathcal B$.
If $L=1$, the index register is empty. Apply $P_g$ to $\mathsf I_g\mathsf S_g\mathsf H_g$, leaving
$\mathsf X,\mathsf Y,\mathsf u_g, \mathsf v_g$ unchanged.

Second, we define precisely the coordinate-controlled gates.
For $x\in[-1,1]$, $B_x$ is a single-qubit gate defined in Eq.~\eqref{eq:s4c-scalar-gate}.
When its parameter depends on a coordinate register,
the corresponding operation is the following
$(q+1)$-qubit uniformly controlled one-qubit gate:
\begin{equation}
    B_{a(\mathsf X_i)}^{(\mathsf u_g)}
    :=
    \sum_{r\in\mathcal I}
    \ket r\bra r_{\mathsf X_i}
    \otimes (B_{a(r)})_{\mathsf u_g}.
\end{equation}
Apply $Z$ to $\mathsf S_g$ and
$B_{a(\mathsf X_i)}^{(\mathsf u_g)}$ to $\mathsf X_i \mathsf u_g$.
For the branch indexed by $j$, these operations contribute
the factor $(-1)^{s_{ij}}$ and replace
$\ket0_{\mathsf u_g}$ by
\[
    a(k_i)\ket0_{\mathsf u_g}
    +\sqrt{1-a(k_i)^2}\ket1_{\mathsf u_g}.
\]

Third, we construct the read operation required to select
the coordinate word $l_j$.
Write
\begin{equation}\label{eq:YB-register}
    \mathsf Y_{\mathcal B}
    =(\mathsf Y_{j_0},\ldots,\mathsf Y_{j_0+L-1}),
    \qquad
    \ket{\vct l_{\mathcal B}}_{\mathsf Y_{\mathcal B}}
    =
    \bigotimes_{j\in\mathcal B}
    \ket{l_j}_{\mathsf Y_j}.
\end{equation}
Introduce a $q$-qubit all-zero register $\mathsf T$. We use read operator $\mathcal R$ (Eq.~\eqref{eq:s4c-read}) to read the coordinate word selected by $\mathsf I_g$ from $\mathsf Y_{\mathcal B}$ into $\mathsf T$. Specifically, we set $\mathsf A_0=\mathsf I_g$ and $\mathsf Z_0=\mathsf T$,
and initialize all remaining address and target registers
to zero. Therefore, we have 
\begin{equation}
\mathcal{R}\ket{\vct l_{\mathcal B}}_{\mathsf Y_{\mathcal B}}
 \ket{j-j_0}_{\mathsf I_g}\ket{0^q}_{\mathsf T}  =\ket{\vct l_{\mathcal B}}_{\mathsf Y_{\mathcal B}}
 \ket{j-j_0}_{\mathsf I_g}\ket{l_j}_{\mathsf T}
\end{equation}
Consequently, its depth is $O((q+\kappa)\kappa^2)$, and its size and ancilla-count are $O(L(q+\kappa)\kappa^2)$ by Lemma \ref{lem:s4c-read}.
Now apply the strictly defined word-controlled unitary
\begin{equation}
    B_{b(\mathsf T)}^{(\mathsf v_g)} :=
    \sum_{r\in\mathcal I}
    \ket r\bra r_{\mathsf T}\otimes (B_{b(r)})_{\mathsf v_g}.
\end{equation}
Then its action is
\begin{equation}
    \ket{l_j}_{\mathsf T}\ket 0 _{\mathsf v_g}
    \longmapsto
    \ket{l_j}_{\mathsf T} \bigl(b(l_j)\ket0+\sqrt{1-b(l_j)^2}\ket1\bigr)_{\mathsf v_g}.
\end{equation}
This is a $(q+1)$-qubit unitary, implemented using
Lemma~\ref{lem:tradeoff-ucg}.
Apply $\mathcal R^{\dagger}$ and the register $\mathsf T$ are restored to all-zero state.

At this point, the complete state is
\begin{multline}
\ket{\vct k,\vct l}_{\mathsf X\mathsf Y}
\otimes
\sum_{j\in\mathcal B}
\sqrt{\frac{|\theta_{ij}|}{\tau_g}}
(-1)^{s_{ij}}
\ket{j-j_0}_{\mathsf I_g}
\ket{s_{ij}}_{\mathsf S_g}
\\
\otimes
\bigl(a(k_i)\ket0+\sqrt{1-a(k_i)^2}\ket1\bigr)_{\mathsf u_g}
\otimes
\bigl(b(l_j)\ket0+\sqrt{1-b(l_j)^2}\ket1\bigr)_{\mathsf v_g}.
\label{eq:s4c-sum-state}
\end{multline}
Here we omit the ancillary registers which are all-zeros.

Finally, apply $P_g^\dagger$ to $\mathsf I_g\mathsf S_g\mathsf H_g$.
Eq.~\eqref{eq:s4c-weight-preparation} implies
\begin{equation}
\bra{0^{\log(L)+1+m_g}}
P_g^\dagger
\ket{j-j_0,s_{ij},0^{m_g}}
=
\sqrt{\frac{|\theta_{ij}|}{\tau_g}}.
\end{equation}
Therefore, the coefficient of
$\ket{\vct k,\vct l}\ket{0_{\rm aux}}$
in the final state is
\begin{align}
&\sum_{j\in\mathcal B}
\sqrt{\frac{|\theta_{ij}|}{\tau_g}}
(-1)^{s_{ij}}
a(k_i)b(l_j)
\sqrt{\frac{|\theta_{ij}|}{\tau_g}}=
\frac{a(k_i)}{\tau_g}
\sum_{j\in\mathcal B}
|\theta_{ij}|(-1)^{s_{ij}}b(l_j)=
\frac{a(k_i)}{\tau_g}
\sum_{j\in\mathcal B}\theta_{ij}b(l_j)
=
A_g(\vct k,\vct l).
\end{align}
This proves Eq.~\eqref{eq:s4c-sum-amplitude}.
\end{proof}

The next lemma gives the additional gates that approximate $e^{\iu f_g}$ and bounds the error of the entire output.}

\begin{lemma}
\label{lem:s4c-phase}
Let $g=(i,\mathcal{B})$ be a nonzero weak group and let $K\ge1$ be an integer. There is a circuit $\mathcal G_g$ such that, for any $2pq$-qubit state $\ket{\vct \psi}$,
\begin{equation}
 \left\|\mathcal G_g(\ket{\vct \psi}\ket{\vct 0_{\rm aux}})
       -(e^{\iu f_g}\ket{\vct \psi})\ket{\vct 0_{\rm aux}}\right\|
 \le\left(\frac{10e^{\tau_g}\tau_g^{K+1}}{(K+1)!}\right)^{1/2}.
 \label{eq:s4c-local-error}
\end{equation}
The ancillary register $\ket{\vct 0_{\rm aux}}$ consists of $K$ separate copies of the register in Eq.~\eqref{eq:Vg-aux}, a $\lceil \log(K+1)\rceil$-qubit degree register together with the workspace of its state preparation on it,
and one additional qubit.
\end{lemma}
\begin{proof}
First, define
\[
 \alpha_g=\sum_{r=0}^K\frac{\tau_g^r}{r!},\qquad
 P_K(z)=\sum_{r=0}^K\frac{(\iu z)^r}{r!}.
\]
Since $0<\tau_g\le1/2$, we have $1\le\alpha_g<2$.
Prepare quantum state on the $\lceil\log (K+1)\rceil$-qubit degree register
\begin{equation}
 \sum_{r=0}^K\sqrt{\frac{\tau_g^r}{\alpha_g r!}}\ket r.
 \label{eq:s4c-degree-state}
\end{equation}
Its unused basis values have amplitude zero.

Second, for each $s=1,\ldots,K$, apply $\mathcal V_g$
to the coordinate registers $\mathsf X,\mathsf Y$ and
the $s$-th ancillary copy defined in Eq.~\eqref{eq:Vg-aux}, conditioned on $r\ge s$.
The other copies are unchanged. To implement this condition, reversibly compute
$[r\ge s]$ into an ancillary qubit initialized to $\ket0$, use this qubit to control $\mathcal V_g$, and reverse the condition computation. The ancillary qubit and
the comparison workspace return to zero and can be reused for the next value of $s$.
All applications share the same coordinate registers but use separate auxiliary copies.

Next, apply the phase $\iu^r$ to $\ket r$.
Writing $r=r_0+2r_1+4r_2+\cdots$, where $r_b\in\{0,1\}$, gives $\iu^r=\iu^{r_0}(-1)^{r_1}$. Hence apply $S$ to the least significant qubit and, when present, $Z$ to the second least significant qubit. If the register has only one qubit, only $S$ is needed.

Fix a coordinate input $\ket{\vct k,\vct l}$ and write
$z=f_g(\vct k,\vct l)$.
Since $\mathcal V_g$ preserves the coordinate labels,
Eq.~\eqref{eq:s4c-sum-amplitude} gives the same
all-zero output amplitude for every auxiliary copy:
\[
A_g(\vct k,\vct l)=\frac{z}{\tau_g}.
\]
In a branch $0\le r\le K$, the first $r$ auxiliary
copies are acted on separately, while the remaining
$K-r$ copies stay in their all-zero states.
Before applying the phase $\iu^r$, the joint all-zero
output amplitude of these copies is therefore
\[
\prod_{s=1}^{r}A_g(\vct k,\vct l)
=
A_g(\vct k,\vct l)^r
=
\left(\frac{z}{\tau_g}\right)^r,
\]
where the empty product for $r=0$ equals one.
Including the phase, this step contributes the factor $\left(\frac{\iu z}{\tau_g}\right)^r$,
to the all-zero ancillary component of the $r$ branch. After these operations, projecting the $K$ auxiliary copies onto their all-zero states leaves the following unnormalized vector on the degree register.
\begin{equation}
    \sum_{r=0}^K\sqrt{\frac{\tau_g^r}{\alpha_g r!}}\left(\frac{\iu z}{\tau_g}\right)^r\ket r.
\end{equation}

Third, undo the preparation in Eq.~\eqref{eq:s4c-degree-state}.
The coefficient of the all-zero state of this register, its
preparation workspace, and the $K$ copies is
\begin{equation}
 \sum_{r=0}^K\frac{\tau_g^r}{\alpha_g r!}
           \iu^r(z/\tau_g)^r
       =\frac{P_K(z)}{\alpha_g}.
 \label{eq:s4c-polynomial-amplitude}
\end{equation}
Apply $B_{\alpha_g/2}$ to the additional qubit, initially zero.
By Eq.~\eqref{eq:s4c-scalar-gate}, this multiplies the all-zero
coefficient by $\alpha_g/2$. Denote this complete circuit by
$W_g$. Its all-zero coefficient is $P_K(z)/2$. Namely
\begin{equation}\label{eq:Wg-block-encoding}
    (\bra{\vct k,\vct l}\otimes \bra{\vct 0_{\rm aux}})W_g(\ket{\vct k,\vct l}\otimes \ket{\vct 0_{\rm aux}})=\frac{P_K(z)}{2},
\end{equation}
where $\frac{P_K(z)}{2}\approx \frac{e^{\iu z}}{2}$.

Next, let $R_g$ change the sign exactly when all the ancillary
registers of $W_g$ are zero:
\begin{equation}
 R_g\ket{\vct k,\vct l}\ket{\vct b}
 =\begin{cases}
   -\ket{\vct k,\vct l}\ket{\vct b},&\vct b=\vct0,\\
   \phantom{-}\ket{\vct k,\vct l}\ket{\vct b},&\vct b\ne\vct0.
  \end{cases}
 \label{eq:s4c-zero-sign}
\end{equation}
The tested bits include all state-preparation workspace, the
integer register, the $K$ ancillary copies, and the additional
qubit. Read and control-computation helpers are cleared before
this operation. Computing the all-zero condition, applying a
phase $-1$, and reversing that computation implements $R_g$.

Finally, apply $W_g$, $R_g$, $W_g^\dagger$, $R_g$, and $W_g$
in this order, followed by a single-qubit $-I$ gate. Thus
\begin{equation}
 \mathcal G_g=-W_gR_gW_g^\dagger R_gW_g.
 \label{eq:s4c-group-phase-circuit}
\end{equation}
For the fixed coordinate input, put
$\Pi=\ket{\vct 0_{\rm aux}}\bra{\vct 0_{\rm aux}}$, so that $R_g=I-2\Pi$.
Expanding the two factors $I-2\Pi$ gives
\[
 -W_gR_gW_g^\dagger R_gW_g
 =-W_g+2W_g\Pi+2\Pi W_g-4W_g\Pi W_g^\dagger\Pi W_g.
\]
Taking its all-zero matrix element and using
Eq.~\eqref{eq:Wg-block-encoding}, we obtain
\begin{equation}
 (\bra{\vct k,\vct l}\otimes \bra{\vct 0_{\rm aux}})\mathcal G_g(\ket{\vct k,\vct l}\otimes\ket{\vct 0_{\rm aux}})
 =\frac32P_K(z)-\frac12|P_K(z)|^2P_K(z).
 \label{eq:s4c-phase-amplitude}
\end{equation}
This identity holds for each fixed coordinate input $\ket{\vct k,\vct l}$.

We now bound the error. Since $|z|\le\tau_g$, the exponential
series gives
\begin{align*}
 |P_K(z)-e^{\iu z}|
 &\le\sum_{r=K+1}^{\infty}\frac{\tau_g^r}{r!}\le\frac{e^{\tau_g}\tau_g^{K+1}}{(K+1)!}
   =:\eta_g.
\end{align*}
The second inequality uses
$(K+1+s)!\ge(K+1)!s!$. Also $\eta_g\le e^{1/2}/8<1$,
and hence $|P_K(z)|\le1+\eta_g$. Subtracting
$e^{\iu z}\overline{e^{\iu z}}e^{\iu z}=e^{\iu z}$
from the cubic product in Eq.~\eqref{eq:s4c-phase-amplitude}
one factor at a time gives
\begin{align*}
 \left|\frac32P_K(z)-\frac12|P_K(z)|^2P_K(z)-e^{\iu z}\right|
 &\le\frac32\eta_g
  +\frac12\eta_g\bigl((1+\eta_g)^2+(1+\eta_g)+1\bigr)\le5\eta_g.
\end{align*}
Both $\mathcal G_g\ket{\vct k,\vct l,\vct 0_{\rm aux}}$ and
$e^{\iu z}\ket{\vct k,\vct l,\vct 0_{\rm aux}}$ are normalized.
Their squared distance is therefore
\begin{align*}
 &\left\|\mathcal G_g\ket{\vct k,\vct l,\vct 0_{\rm aux}}
       -e^{\iu z}\ket{\vct k,\vct l,\vct 0_{\rm aux}}\right\|^2=2-2\operatorname{Re}\!\left[
      e^{-\iu z}\left(\frac32P_K(z)-\frac12|P_K(z)|^2P_K(z)\right)\right]
       \le10\eta_g.
\end{align*}
This expression includes every auxiliary output, not just its
all-zero component. Different coordinate inputs remain in
orthogonal coordinate basis states. For a superposition, its
squared error is consequently the weighted sum of these squared
errors, and is at most $10\eta_g$. Taking square roots proves
Eq.~\eqref{eq:s4c-local-error}.

The decomposition in Eq.~\eqref{eq:s4c-group-phase-circuit}
uses $3K$ controlled calls to $\mathcal V_g$ or its inverse,
six applications of the state preparation or its inverse,
three single-qubit gates $B_{\alpha_g/2}$ or their inverses,
and two gates $R_g$.
\end{proof}

{
Recall that
\begin{equation}
 L=2^{\lfloor\min\{t/16,4\log_2(t+1)\}\rfloor},\qquad
 G=p^2/L,\qquad \kappa=\log(2L).
\end{equation}
\begin{lemma}[Circuit for the weak terms]
\label{lem:s4c-schedule}
Let $m\ge 0$ be some integer. Define a transformation $\mathcal{G}_{\rm all}$ satisfying
\[
 \mathcal G_{\rm all}
    \bigl(\ket{\vct X,\vct Y}\bigr)\ket{0^m}_{\mathsf A}
 =\left[(\prod_g\mathcal G_g)
       \ket{\vct X,\vct Y}\ket{0^m}_{\mathsf A}\right]
\]
Then $\mathcal{G}_{\rm all}$ can be implemented by a standard circuit of 
\begin{align}
 \text{depth:}\qquad&O\bigl(t+Q+\log(K+1)+(q+\kappa)\kappa^2\bigr),\notag\\
 \text{size:}\qquad&O(Kp^2+KGF),\\
 \text{ancilla-count:}\qquad &m=O(Kp^2/t+KGF),
 \label{eq:s4c-schedule}
\end{align}
where $F=Q+(q+\kappa)\kappa^2$.
\end{lemma}
\begin{proof}
First, implement each $P_g$ (Eq.~\eqref{eq:s4c-weight-preparation}) on its $\kappa$-qubit address and
sign registers with $m_g=\lceil L/t\rceil$ ancillary qubits.
Lemma~\ref{lem:state-tradeoff} gives size $O(L)$ and depth
$O(\kappa+2L/(\kappa+m_g))=O(t)$.
Including its two scalar qubits, one $\mathcal{V}_g$
has $O(L/t+\kappa)$ qubits. Since there are $G$ groups and each group use $K$ times of $\mathcal{V}_g$ in one $\mathcal{G}_g$. Thus all $KG$ preparations use
$O(L)\cdot KG= O(Kp^2)$ gates and $KG\cdot O(L/t+\kappa)= O(Kp^2/t)$ qubits. Here the prescribed
$L$ satisfies $t(\kappa+1)/L=O(1)$: on its exponential branch,
writing $r=\lfloor t/16\rfloor$ bounds this ratio by
$16(r+1)(r+2)/2^r$; on its polynomial branch use
$L\ge(t+1)^4/2$. These preparation ancillary qubits remain part of the group registers throughout the phase circuit.

{Second, construct the coordinate copies required for
the parallel reads.
Recall that $\mathsf Y_{\mathcal B}$ contains the
$L$ coordinate registers indexed by $\mathcal B$,
and hence consists of $Lq$ qubits.

For a fixed column block $\mathcal B$ and a fixed
auxiliary-copy index among the $K$ copies, the groups
$(i,\mathcal B)$, with $0\le i<p$, produce $p$ queries.
Partition them into $p/L$ batches of $L$ queries.
Assign each batch a separate copy of
$\mathsf Y_{\mathcal B}$ (Eq.~\eqref{eq:YB-register}), so that its $L$ queries can be handled together by one instance of
$\mathcal R$. Across all $K$ ancillary-copy indices, each
$\mathsf Y_{\mathcal B}$ therefore requires $Kp/L$
copies. Using CNOT-tree fan-out on zero-initialized
targets, implement
\begin{equation}
\ket{\vct l_{\mathcal B}}_{\mathsf Y_{\mathcal B}}
\bigl(\ket{0^{Lq}}\bigr)^{\otimes Kp/L}
\longmapsto
\ket{\vct l_{\mathcal B}}_{\mathsf Y_{\mathcal B}}
\bigl(\ket{\vct l_{\mathcal B}}\bigr)^{\otimes Kp/L},
\end{equation}
where the output copies occupy distinct registers.
In addition, assign a separate copy of the row
coordinate register $\mathsf X_i$ to each group
$(i,\mathcal B)$ and each auxiliary-copy index.
Each $\mathsf X_i$ also requires $Kp/L$ copies,
because there are $p/L$ column blocks.

Let $\mathcal C$ denote all these fan-out gates.
The column-coordinate copies use 
$ 
\frac pL\cdot\frac{Kp}{L}\cdot Lq
=
KGq
$ 
ancillary qubits, and the row-coordinate copies use
$ 
p\cdot\frac{Kp}{L}\cdot q
=
KGq.
$ 
Thus $\mathcal C$ has size $O(KGq)$ and uses
$O(KGq)$ ancillary qubits.
A fan-out gate make $s$ copies of one qubit has
depth $O(\log (s))$ by Lemma \ref{lem:fanout-standard}.
Since the trees for distinct coordinate bits
run in parallel, the depth of $\mathcal C$ is
\[
O\!\left(\log_2\!(Kp/L)\right)
=
O\bigl(t+\log_2(K+1)\bigr),
\qquad p=2^t.
\]
After the subsequent operations, which preserve
the computational-basis values of these coordinate
registers, applying $\mathcal C^\dagger$ restores
all copy registers to zero.}

Third, apply the parallel read circuit $\mathcal R_{\rm all}$
from Lemma~\ref{lem:s4c-read}. For $j\in\mathcal B$, the stored
address is $j-\min\mathcal B$. Define $e\in\{0,1\}$ as the enable bit for this copy before using $B_{b(l_j)}^e$. The following sequence reads,
uses, and erases the selected coordinate word:
\[
 \begin{aligned}
 \ket{e}\ket{j-\min\mathcal B,0^q}\ket v
 &\longmapsto\ket{e}\ket{j-\min\mathcal B,l_j}\ket v\\
 &\longmapsto\ket{e}\ket{j-\min\mathcal B,l_j}B_{b(l_j)}^e\ket v\\
 &\longmapsto\ket{e}\ket{j-\min\mathcal B,0^q}B_{b(l_j)}^e\ket v.
 \end{aligned}
\]
The middle step is a uniformly controlled one-qubit gate define in Lemma \ref{lem:tradeoff-ucg}; the last is the
inverse read. Apply $B_{a(X_i)}^e$ to the other scalar qubit
and the enabled sign gate to the sign qubit. Call all these
parallel scalar and sign gates $\mathcal L$, and put
$\mathcal P=\bigotimes_g P_g^{\otimes K}$, one preparation for each group and each power factor. A complete parallel implementation of the copies of $\mathcal V_g$ is
\begin{equation}
 \mathcal C^\dagger\mathcal P^\dagger
 \mathcal R_{\rm all}^\dagger\mathcal L
 \mathcal R_{\rm all}\mathcal P\mathcal C.
 \label{eq:s4c-block-schedule}
\end{equation}
The reads cost $O(KG(q+\kappa)\kappa^2)$ gates and ancillas
and depth $O((q+\kappa)\kappa^2)$ by Lemma \ref{lem:s4c-read}. The scalar gates cost
$O(KGQ)$ gates and depth $O(Q+q)$, with no extra ancillas by Lemma \ref{lem:tradeoff-ucg}.
The read words, addresses, and coordinate copies retain their
values during $\mathcal L$, so the inverse reads and fan-out gates clear their ancillas for every input of the group registers.

Fourth, implement the controls and the gates $R_g$ in
$\mathcal G_g$. From the stored integer $r$, compute the $K$ enable
bits by
\[
 \ket{r,0^K}\longmapsto
 \ket{r,\mathbf1[1\le r],\ldots,\mathbf1[K\le r]}.
\]
Copy the register storing $r$ and apply the binary-to-one-hot conversion
of Lemma~\ref{lem:one-hot}, padding its dimension to a power
of two. Computing cumulative ORs on a balanced binary tree
then gives these enables with $O(K)$ gates and ancillas
and $O(\log(K))$ depth. Reverse this computation after the
enabled circuits. These circuits leave $r$ unchanged.
Preparing the register storing $r$ also has size and workspace $O(K)$
and depth $O(\log(K))$ by Lemma~\ref{lem:state-tradeoff}.
Its preparation ancillas, unlike the ancillary enable ancillas,
are retained and included in the group's zero test.
To reflect a group of $b$ auxiliary qubits, compute its all-zero
predicate with a balanced Toffoli tree, apply its sign, and
reverse the tree. This costs $O(b)$ gates and ancillary ancillas
and $O(\log(b))$ depth. Test every private group register,
integer register, preparation ancilla, and additional qubit required by
$\mathcal G_g$, and perform the tests separately for each $g$.

The complete operations preserve the coordinate words
and act on separate auxiliary copies. Thus they commute
and run in parallel using the copies above, including on
nonzero or entangled group-register inputs. The complete
phase circuits of distinct groups similarly run in parallel.
Each phase circuit uses only a constant
number of the parallel stages in Eq.~\eqref{eq:s4c-block-schedule},
their literal inverses, and group reflections. Consequently
\begin{align*}
 \text{depth:}\quad
 &O(t)+O(Q+q)+O((q+\kappa)\kappa^2)
           +O(\log(K+1)),\\
 \text{size:}\quad
 &O(Kp^2)+O(KG(q+\kappa)\kappa^2)+O(KGQ)+O(Kp^2/t),\\
 \text{ancilla-count:}\quad
 &O(Kp^2/t)+O(KG(q+\kappa)\kappa^2)+O(KGq).
\end{align*} 
\end{proof}

\paragraph{Implementation of the diagonal phase gate.} Now we approximately implement the diagonal phase gates.
\begin{theorem}
\label{thm:s4c-diagonal}
Let $H\in\R^{p \times p}$ be an orthogonal matrix and $K\ge1$ be an integer.
There are unitary circuits $\widetilde\Delta_H^\uparrow$ such that, for any $\ket\varphi$ of the coordinate registers,
\begin{equation}
\norm{\widetilde\Delta_H^\uparrow
          (\ket\varphi\ket{0^m})
       -(\Delta_H^\uparrow\ket\varphi)\ket{0^m}}
\le\sum_g\sqrt{\frac{10e^{\tau_g}\tau_g^{K+1}}{(K+1)!}},
\label{eq:s4c-diagonal-error}
\end{equation}
and the same statement holds for the $\Delta_H^\downarrow$.
Here $m$ includes every ancillary qubit of the circuit.
Each circuit has
\begin{align}
\text{depth:}\quad
 &O\bigl(t+Q+\log(K+1)+(q+\kappa)\kappa^2\bigr),\notag\\
\text{size:}\quad
 &O\bigl(N_{\rm str}(q+1)^2+Kp^2
           +KG[Q+(q+\kappa)\kappa^2]\bigr),\notag\\
\text{ancilla-count:}\quad
 &O\bigl(N_{\rm str}(q+1)^2+Kp^2/t
           +KG[Q+(q+\kappa)\kappa^2]\bigr).
\label{eq:s4c-diagonal-resources}
\end{align}
\end{theorem}

\begin{proof}
We only construct the circuit for $\Delta_H^\uparrow$.
By Eq.~\eqref{eq:s4c-phase-product}, $\Delta_H^\uparrow$ has the following decomposition
\[
\Delta_H^\uparrow
 =\left(\prod_g e^{\iu f_g}\right)
       \Delta_{H,\rm str}^\uparrow.
\]
Then the circuit for $\Delta_H^\uparrow$ is
\begin{equation}
\widetilde\Delta_H^\uparrow
 =\left(\prod_g\mathcal G_g\right)
       \Delta_{H,\rm str}^\uparrow,
\label{eq:s4c-diagonal-circuit}
\end{equation}
The circuit $\Delta_{H,\rm str}^\uparrow$ does not introduce error. According to Lemma \ref{lem:s4c-phase}, each $\mathcal{G}_g$ introduce error at most $\sum_g\sqrt{\frac{10e^{\tau_g}\tau_g^{K+1}}{(K+1)!}},
$ which gives Eq.~\eqref{eq:s4c-diagonal-error}.
By Lemmas~\ref{lem:s4c-strong} and \ref{lem:s4c-schedule},
implies \eqref{eq:s4c-diagonal-resources} directly. The same results holds for $\Delta_H^\downarrow$.
\end{proof}

\subsection{Reducing the circuit size and ancilla-count}

\begin{theorem}[Implementation of $U$]
\label{thm:s4c-native}
Let $n\ge2$, $d=2^n$, $p=2d$, and $\ell=2p$.
For every $U\in\mathrm U(d)$, there is a unitary circuit
$T_U$ on $n+m$ qubits such that
\begin{equation}
\norm{T_UJ_T-J_TU}_{\op}<2^{-20}2^{-n/2},
\qquad
J_T\ket\psi=\ket\psi\ket{0^m}.
\label{eq:s4c-native-accuracy}
\end{equation}
The circuit has depth $O(n)$ and size $O(4^n)$, using
$m=O(4^n/n)$ ancillary qubits. 
\end{theorem}

\begin{proof}
Choose
\[
r_0=n+20,\qquad
B=2^{\lceil\frac12\log_2(2^{20}r_0)\rceil},
\qquad Q=2B^2,\qquad q=\log_2Q,
\]
and take $K=31$ in Theorem~\ref{thm:s4c-diagonal}.

First, apply the exact encoder $E$ from
Theorem~\ref{thm:CQ}. For
$\ket\psi=\sum_{j=0}^{d-1}c_j\ket j$ with
$\norm{\vct c}=1$, its output is
\[
EJ_T\ket\psi
=\ket{\Psi_Q(\vct c)}\ket{0^w},
\qquad w=n+m-q\ell.
\]

Second, apply $ {\widehat{J}_0}$, followed by the three finite shears
in Eq.~\eqref{eq:dir-FU}. Replace only their diagonal
phases by Theorem~\ref{thm:s4c-diagonal}. More explicitly,
\begin{align*}
\widetilde{ {\widehat{T}}}_H
 &=(D^{\otimes p}\otimes I)^\dagger
        \widetilde\Delta_H^\uparrow
        (D^{\otimes p}\otimes I),\\
\widetilde{ {\widehat S}}_H
 &=(I\otimes D^{\otimes p})^\dagger
        \widetilde\Delta_H^\downarrow
        (I\otimes D^{\otimes p}).
\end{align*}

The complete coordinate circuit is
\[
\widetilde{\mathcal P}_U
 =\widetilde{ {\widehat{T}}}_{R_U}^{(2)}
   \widetilde{ {\widehat{S}}}_{-R_U^T}
   \widetilde{ {\widehat{T}}}_{R_U}^{(1)} {\widehat{J}_0},
\]
where $(1)$ and $(2)$ distinguish the two occurrences
of the upper shear and their separate auxiliary registers.

Finally, apply $E^\dagger$. Thus
\begin{equation}
T_U=E^\dagger
   \widetilde{ {\widehat T}}_{R_U}^{(2)}
   \widetilde{ {\widehat S}}_{-R_U^T}
   \widetilde{ {\widehat T}}_{R_U}^{(1)} {\widehat{J}_0} E
   =E^\dagger\widetilde{\mathcal P}_U E.
\label{eq:s4c-U-circuit}
\end{equation}
All circuits in this equation act on the same $n+m$ qubits,
by identity extension on registers that they do not use.

We next prove the error bound. For $K=31$,
Lemma~\ref{lem:s4c-phase} bounds the error of one group by
$\sqrt{20/32!}\,\tau_g^{16}$. There are at most $G=p^2/L$
groups in each of the three shears, and
$\tau_g\le Lp^{-1/4}/2$. Consequently, for every unit vector
$\ket\varphi$ of the coordinate registers,
\begin{align*}
\norm{\widetilde{\mathcal P}_U(\ket\varphi\ket{0^w})
       -(\mathcal P_{\rm  {discrete}}\ket\varphi)\ket{0^w}}&\le3G\sqrt{\frac{20}{32!}}
                    \left(\frac{Lp^{-1/4}}2\right)^{16}\\
&=3\sqrt{\frac{20}{32!}}\,2^{-16}L^{15}p^{-2}\\
&\le3\sqrt{\frac{20}{32!}}\,2^{-16}p^{-17/16}
 <2^{-71}p^{-17/16}.
\end{align*}
The second inequality uses $L\le p^{1/16}$;
the last uses $180\cdot2^{110}<32!$.
The exact Fourier transforms and $ {\widehat{J}_0}$ introduce no
additional implementation error.

Also, $B^2\ge2^{20}(n+20)$ and $\ell=2^{n+2}$ imply
\begin{align*}
\log \bigl(100\ell e^{-B^2/256}\bigr)
&\le\log (400)+n-4096(n+20)\log e<-40-n/2.
\end{align*}
Using the encoder identity for both $\vct c$ and $U\vct c$,
unitarity of $E$, and Lemma~\ref{lem:red34-coordinate-error},
we obtain
\begin{align*}
\norm{T_UJ_T\ket\psi-J_TU\ket\psi}
&=\norm{\widetilde{\mathcal P}_U
       (\ket{\Psi_Q(\vct c)}\ket{0^w})
       -\ket{\Psi_Q(U\vct c)}\ket{0^w}}\\
&\le2^{-71}p^{-17/16}
     +\norm{\mathcal P_{\rm  {discrete}}\ket{\Psi_Q(\vct c)}
                       -\ket{\Psi_Q(U\vct c)}}\\
&<2^{-71}p^{-17/16}+2^{-40}d^{-1/2}\\
&\le(2^{-71}+2^{-40})d^{-1/2}
 <2^{-20}d^{-1/2}.
\end{align*}
Taking the supremum over $\norm{\vct c}=1$ proves
Eq.~\eqref{eq:s4c-native-accuracy}.

It remains to count the resources. Since $t=n+1$,
$Q=\Theta(t)$, $q=O(\log(t+1))$, and $K=31$,
Theorem~\ref{thm:s4c-diagonal} gives depth $O(t)$
for each diagonal phase. Furthermore,
\[
N_{\rm str}(q+1)^2
 =O\bigl(p^{3/2}t^2\log^2(t+1)\bigr)=o(p^2/t),
\]
and, since $L=\Theta(t^4)$ and $\kappa=O(\log(t+1))$,
\[
G[Q+(q+\kappa)\kappa^2]
 =O\left(\frac{p^2}{t^4}
                 [t+\log^3(t+1)]\right)=o(p^2/t).
\]
Thus one diagonal phase has size $O(p^2)$ and
ancilla-count $O(p^2/t)$. The complete circuit has
\begin{align*}
\text{depth:}\quad
 &2O(n+q)+6O(q)+3O(t)+O(1)=O(n),\\
\text{size:}\quad
 &2O(pQ)+6O(pq^2)+3O(p^2)+O(pq)=O(p^2),\\
\text{ancilla-count:}\quad
 &O(pQ)+3O(p^2/t)+O(pq)=O(p^2/t).
\end{align*}
The terms correspond to the encoder and decoder, Fourier
transforms, three diagonal phases, and signed swap.
The ancillary count also includes the coordinate padding
$q\ell-n$. The finitely many smaller $n$ are absorbed
in the constants. Finally, $p^2=\Theta(4^n)$ and
$t=\Theta(n)$ give the claimed bounds.
\end{proof}

\begin{corollary}[Arbitrary requested accuracy]
\label{cor:s4c-general}
Let $0<\epsilon\le1/2$. Then the circuit $T_U$
in Eq.~\eqref{eq:s4c-U-circuit} satisfies
\[
\norm{T_UJ_T-J_TU}_{\op}<\epsilon,
\]
which has resources
\begin{align*}
\text{depth:}\quad
 &O\left(n+\log(1/\epsilon)\right),\\
\text{size:}\quad
 &O\left(
 4^n\left[
 1+\frac{\log(1/\epsilon)}{n}
   +\frac{\log^2(1/\epsilon)}{n^5}
 \right]\right),\\
\text{ancilla-count:}\quad
 &O\left(
 \frac{4^n}{n}\left[
 1+\frac{\log(1/\epsilon)}{n}
   +\frac{\log^2(1/\epsilon)}{n^4}
 \right]\right).
\end{align*}
In particular, for $\epsilon=2^{-cn}$ with any fixed $c>0$,
the depth, size, and ancilla-count are
$O_c(n)$, $O_c(4^n)$, and $O_c(4^n/n)$, respectively.
\end{corollary}

\begin{proof}
Choose
\[
\lambda=\lceil\log_2(1/\epsilon)\rceil,
\qquad r_0=n+\lambda+20,
\qquad
K=\max\left\{31,\left\lceil
 \frac{32}{3}\left(2+\frac{\lambda+6}{t}\right)
                    \right\rceil\right\}.
\]
Choose $B,Q,q$ as in the preceding theorem with this $r_0$,
and keep the same $L,G,\kappa$.
First, the new grid choice gives
\[
100\ell e^{-B^2/256}<2^{-40-\lambda}d^{-1/2}<\epsilon/4.
\]
Second, $\tau_g\le p^{-3/16}/2$ and
Lemma~\ref{lem:s4c-phase} give total implementation error
at most
\begin{align*}
3G\sqrt{\frac{10e^{1/2}}{(K+1)!}}
       (p^{-3/16}/2)^{(K+1)/2}
&\le15p^{2-3(K+1)/32}\le15\,2^{-\lambda-6}
 \le\frac{15}{64}\epsilon<\epsilon/4.
\end{align*}
Here the first inequality uses $G\le p^2$ and
$\sqrt{10e^{1/2}}<5$; the second uses the choice of $K$
and $p=2^t$; the third uses $2^{-\lambda}\le\epsilon$.
Finally, the encoder--coordinate--decoder error calculation
in the preceding proof bounds the complete output error
by the sum of these two terms, which is less than $\epsilon$.

For the resources, add the encoder, Fourier, and signed-swap
costs to three copies of
Eq.~\eqref{eq:s4c-diagonal-resources}, and use
$N_{\rm str}=O(\min\{p^2,p^{3/2}Q^2\})$.
Since $q^2=O(Q)$, the Fourier and coordinate-padding costs
are absorbed in $O(pQ)$. Also,
$Q=\Theta(r_0)$, $q,\kappa=O(\log r_0)$, and
$K=O(1+\lambda/t)$, so
$(q+\kappa)\kappa^2=O(\log^3 r_0)=O(r_0)$
and $\log(K+1)=O(\log r_0)$. This proves the 
resource bounds:
\begin{align*}
\text{depth:}\qquad
 &O\bigl(t+Q+\log(K+1)+(q+\kappa)\kappa^2\bigr)=O(n+\log(1/\epsilon)),\\
\text{size:}\qquad
 &O\bigl(pQ+\min\{p^2,p^{3/2}Q^2\}(q+1)^2
            +Kp^2+KG[Q+(q+\kappa)\kappa^2]\bigr)\\
=& O\left(
 4^n\left[
 1+\frac{\log(1/\epsilon)}{n}
   +\frac{\log^2(1/\epsilon)}{n^5}
 \right]\right),\\
\text{ancilla-count:}\qquad
 &O\bigl(pQ+\min\{p^2,p^{3/2}Q^2\}(q+1)^2
            +Kp^2/t+KG[Q+(q+\kappa)\kappa^2]\bigr)\\
= & O\left(
 \frac{4^n}{n}\left[
 1+\frac{\log(1/\epsilon)}{n}
   +\frac{\log^2(1/\epsilon)}{n^4}
 \right]\right).
\end{align*}
If $\epsilon=2^{-cn}$ for fixed $c>0$,
then $Q=O_c(t)$ and $K=O_c(1)$, so the resource
calculation in the preceding theorem applies.
\end{proof}
}

\section{From approximate to exact linear-depth unitary synthesis}
\label{sec:exactification}

We convert the approximate circuit from Section~\ref{sec:spectral}
into an exact circuit without changing its asymptotic quantum
resources. First, we construct a circuit whose all-zero auxiliary
matrix elements equal a normalized residual matrix. We then combine
this circuit with the approximate circuit and apply two reflections.
The resulting circuit implements the prescribed unitary exactly and
restores every ancillary qubit to zero.

\subsection{A circuit for the residual matrix}
\label{subsec:residual-exactification}

\begin{lemma}[A circuit for a normalized residual matrix]
\label{lem:s5-residual}
Let $n\ge2$, $d=2^n$, and let $\Gamma\in\C^{d\times d}$ be
nonzero. Define
\[
\rho=\norm{\Gamma}_{\F},
\qquad b=\left\lceil d^2/n\right\rceil.
\]
There is a unitary circuit $V_\Gamma$ on $2n+b$ qubits satisfying
\begin{equation}
(\Id\otimes\bra{0^{n+b}})V_\Gamma
(\Id\otimes\ket{0^{n+b}})
=\frac{\Gamma}{\rho}.
\label{eq:s5-residual-matrix}
\end{equation}
It can be implemented by a quantum circuit of depth $O(n)$ and size $O(d^2)$, using $O(d^2/n)$ ancillary qubits.
\end{lemma}
\begin{proof}
Let $\mathsf X$ be the $n$-qubit input register, let $\mathsf B$
be an additional $n$-qubit register, and let $\mathsf H$ contain
$b$ working qubits. For $0\le i<d$, define
\begin{equation}
s_i=\left(\sum_{j=0}^{d-1}|\Gamma_{ij}|^2\right)^{1/2},
\qquad
\ket{\vct r_i}=
\begin{cases}
\displaystyle\frac{1}{s_i}\sum_{j=0}^{d-1}
       \overline{\Gamma_{ij}}\ket j,&s_i>0,\\[2mm]
\ket{0^n},&s_i=0.
\end{cases}
\label{eq:s5-row-states}
\end{equation}
Every $\ket{\vct r_i}$ is normalized, and
$\sum_i s_i^2=\rho^2$.

First, choose a state-preparation circuit $P_{\rm wt}$ on
$\mathsf B\mathsf H$ satisfying
\begin{equation}
P_{\rm wt}\ket{0^n}_{\mathsf B}\ket{0^b}_{\mathsf H}
=\sum_{i=0}^{d-1}\frac{s_i}{\rho}
       \ket i_{\mathsf B}\ket{0^b}_{\mathsf H}.
\label{eq:s5-weight-preparation}
\end{equation}
The coefficients in this state are normalized. Next, choose an
$(n,n)$-CQSP circuit $P_{\rm row}$ on
$\mathsf X\mathsf B\mathsf H$ such that
\begin{equation}
P_{\rm row}
\ket{0^n}_{\mathsf X}\ket i_{\mathsf B}\ket{0^b}_{\mathsf H}
=\ket{\vct r_i}_{\mathsf X}\ket i_{\mathsf B}\ket{0^b}_{\mathsf H},
\qquad 0\le i<d.
\label{eq:s5-row-preparation}
\end{equation}
The two preparation circuits may use the same register $\mathsf H$:
Eq.~\eqref{eq:s5-weight-preparation} returns it to zero before
the next operation. Both circuits are complete unitaries, and their
inverses act on the full registers, not only on the specified inputs.

Second, apply $P_{\rm wt}$, then $P_{\rm row}^{\dagger}$, and
finally swap $\mathsf X$ and $\mathsf B$. Thus
\begin{equation}
V_\Gamma
=(\mathsf{SWAP}_{\mathsf X,\mathsf B}\otimes\Id_{\mathsf H})
  P_{\rm row}^{\dagger}
  (\Id_{\mathsf X}\otimes P_{\rm wt}).
\label{eq:s5-residual-circuit}
\end{equation}
We compute its matrix elements with the ancillary registers zero.
For $0\le i,j<d$, Eqs.~\eqref{eq:s5-weight-preparation}
and~\eqref{eq:s5-row-preparation} give
\begin{align}
&(\bra i_{\mathsf X}\bra{0^n}_{\mathsf B}\bra{0^b}_{\mathsf H})
 V_\Gamma
 (\ket j_{\mathsf X}\ket{0^n}_{\mathsf B}\ket{0^b}_{\mathsf H})
 \notag\\
=&
\sum_{k=0}^{d-1}\frac{s_k}{\rho}
(\bra{0^n}_{\mathsf X}\bra i_{\mathsf B}\bra{0^b}_{\mathsf H})
P_{\rm row}^{\dagger}
(\ket j_{\mathsf X}\ket k_{\mathsf B}\ket{0^b}_{\mathsf H})
 \notag\\
=&
\sum_{k=0}^{d-1}\frac{s_k}{\rho}
\langle\vct r_i\mid j\rangle\langle i\mid k\rangle
=\frac{s_i}{\rho}\langle\vct r_i\mid j\rangle
=\frac{\Gamma_{ij}}{\rho}.
\label{eq:s5-residual-entry}
\end{align}
If $s_i=0$, then $\Gamma_{ij}=0$ for every $j$, and the
same equality holds without division by $s_i$.
This proves Eq.~\eqref{eq:s5-residual-matrix}.
In particular, the complex conjugation in
Eq.~\eqref{eq:s5-row-states} gives $\Gamma_{ij}$, rather
than $\overline{\Gamma_{ij}}$, in the final matrix element.

Finally, we count the resources.
Lemma~\ref{lem:state-tradeoff} implements $P_{\rm wt}$ with
size $O(d)$ and depth $O(n+d/(n+b))$.
Lemma~\ref{lem:controlled-state-tradeoff} implements
$P_{\rm row}$ with size $O(d^2)$ and depth
$O(2n+d^2/(2n+b))$.
Since $b\ge d^2/n$, both depths are $O(n)$.
The inverse $P_{\rm row}^{\dagger}$ has the same costs as
$P_{\rm row}$, and the $n$ disjoint SWAP gates have
constant depth. Therefore,
\begin{align*}
\text{depth:}\quad&
 O\!\left(n+\frac{d}{n+b}\right)
 +O\!\left(2n+\frac{d^2}{2n+b}\right)+O(1)=O(n),\\
\text{size:}\quad& O(d)+O(d^2)+O(n)=O(d^2),\\
\text{ancilla-count:}\quad& n+b=O(d^2/n).
\end{align*}
\end{proof}

\subsection{Linear-depth exact unitary synthesis with optimal ancilla-count}
\label{subsec:exact-bounds}

\begin{lemma}[Exact correction of an approximate circuit]
\label{lem:s5-exactification}
Let $n\ge2$, $d=2^n$, and $U\in\mathrm U(d)$.
Suppose that a standard unitary circuit $T$ on $n+a$ qubits
has depth $D$ and size $S$, and satisfies
\begin{equation}
\norm{TJ_a-J_aU}_{\op}\le d^{-1/2},
\qquad
J_a\ket{\vct\psi}=\ket{\vct\psi}\ket{0^a}.
\label{eq:s5-core-hypothesis}
\end{equation}
Then there is a standard circuit $\mathcal A$ on $n+m$ qubits
such that
\begin{equation}
\mathcal A J_m=J_mU,
\qquad
J_m\ket{\vct\psi}=\ket{\vct\psi}\ket{0^m}.
\label{eq:s5-exact-action}
\end{equation}
Its resources are
\begin{align}
\text{depth:}\quad& O(D+n+\log(a+2)),\notag\\
\text{size:}\quad& O(S+d^2+a),\notag\\
\text{ancilla-count:}\quad& m=O(a+d^2/n).
\label{eq:s5-exactification-resources}
\end{align}
The identity in Eq.~\eqref{eq:s5-exact-action} includes every
ancillary qubit and preserves the prescribed phase of $U$.
It also holds when the input is entangled with a reference system.
\end{lemma}

\begin{proof}
First, define the matrix obtained by taking the all-zero auxiliary
matrix elements of the actual circuit $T$:
\begin{equation}
M=J_a^\dagger TJ_a,\qquad
\Gamma=U-M,\qquad
\rho=\norm{\Gamma}_{\F}.
\label{eq:core-contract}
\end{equation}
Since $J_a^\dagger J_a=\Id$ and $\norm{J_a^\dagger}_{\op}=1$,
\begin{align}
\norm{\Gamma}_{\op}
&=\norm{J_a^\dagger(J_aU-TJ_a)}_{\op}
 \le\norm{J_aU-TJ_a}_{\op}\le d^{-1/2},\\
\rho^2
&=\sum_{j=0}^{d-1}\norm{\Gamma\ket j}^2
 \le d\,\norm{\Gamma}_{\op}^2\le1.
\label{eq:s5-residual-small}
\end{align}
If $\rho=0$, then $M=U$ and $T$ already gives an exact
implementation. Indeed, with $\Pi_a=J_aJ_a^\dagger$,
\[
\Pi_a TJ_a\ket{\vct\psi}=J_aU\ket{\vct\psi}
\]
has the full norm of $TJ_a\ket{\vct\psi}$, so its orthogonal
component is zero. In this case take $\mathcal A=T$ and $m=a$.
We henceforth assume $0<\rho\le1$.

Second, construct $V_\Gamma$ by Lemma~\ref{lem:s5-residual}.
Let the input register be $\mathsf X$ and let $\mathsf E$
contain the $a$ ancillary qubits of $T$.
The residual circuit uses the registers $\mathsf B,\mathsf H$
from the preceding lemma, with $b=\lceil d^2/n\rceil$.
Extend $T$ by the identity on $\mathsf B\mathsf H$, and extend
$V_\Gamma$ by the identity on $\mathsf E$.
Their all-zero auxiliary matrices on this common space are
$M$ and $\Gamma/\rho$, respectively.
Add a selector qubit $\mathsf c$ and choose a single-qubit gate
$P_\rho$ satisfying
\begin{equation}
P_\rho\ket0_{\mathsf{c}}
=\frac{\ket0_{\mathsf c}+\sqrt{\rho}\ket1_{\mathsf c}}{\sqrt{1+\rho}}.
\label{eq:s5-selector-preparation}
\end{equation}
Apply $P_\rho$, then the controlled operation
\begin{equation}
\ket0\bra0_{\mathsf c}\otimes T+\ket1\bra1_{\mathsf c}\otimes V_\Gamma,
\label{eq:s5-selected-circuit}
\end{equation}
and then $P_\rho^\dagger$.  
Taking the all-zero matrix elements of
$\mathsf E,\mathsf B,\mathsf H,\mathsf c$ gives
\begin{equation}
\frac{1}{1+\rho}M+
\frac{\rho}{1+\rho}\frac{\Gamma}{\rho}
=\frac{M+\Gamma}{1+\rho}
=\frac{U}{1+\rho}.
\label{eq:s5-combined-matrix}
\end{equation}
This is an exact matrix identity.

Third, add a qubit $\mathsf u$ initialized to zero and apply
$B_{(1+\rho)/2}$ from Eq.~\eqref{eq:s4c-scalar-gate}.
It is well defined because $0<\rho\le1$.
Let $W$ denote the complete circuit:
\begin{equation}
W=
(B_{(1+\rho)/2})_{\mathsf u}
P_\rho^\dagger
\left(\ket0\bra0_{\mathsf c}\otimes T+
      \ket1\bra1_{\mathsf c}\otimes V_\Gamma\right)P_\rho.
\label{eq:s5-half-unitary-circuit}
\end{equation}
Let $w=a+n+b+2$ count the qubits in
$\mathsf E,\mathsf B,\mathsf H,\mathsf c,\mathsf u$, and define
\[
J_W\ket{\vct\psi}=\ket{\vct\psi}\ket{0^w}.
\]
All preparation workspaces are included among these $w$ qubits.
Since $\bra0 B_{(1+\rho)/2}\ket0=(1+\rho)/2$,
Eq.~\eqref{eq:s5-combined-matrix} yields
\begin{equation}
J_W^\dagger WJ_W=\frac{U}{2}.
\label{eq:s5-half-unitary}
\end{equation}

Fourth, define the reflection
\begin{equation}
\Pi=J_WJ_W^\dagger
=\Id_{\mathsf X}\otimes\ket{0^w}\bra{0^w},
\qquad R_0=\Id-2\Pi.
\label{eq:s5-zero-reflection}
\end{equation}
It changes the sign exactly when all $w$ ancillary qubits
are zero and does not test the data register.
Apply $W$, $R_0$, $W^\dagger$, $R_0$, and $W$ in this order,
followed by a single-qubit $-\Id$ gate. The resulting logical
operation is
\begin{equation}
\mathcal A=-WR_0W^\dagger R_0W.
\label{eq:s5-exact-circuit}
\end{equation}
The scalar minus sign is retained to obtain the prescribed
matrix $U$, rather than $-U$.

Expanding the two reflections and using $W^\dagger W=\Id$ gives
\[
\mathcal A=-W+2W\Pi+2\Pi W-4W\Pi W^\dagger\Pi W.
\]
As $\Pi J_W=J_W$ and $J_W^\dagger\Pi=J_W^\dagger$,
\begin{align}
J_W^\dagger\mathcal A J_W
&=3J_W^\dagger WJ_W
 -4(J_W^\dagger WJ_W)
    (J_W^\dagger W^\dagger J_W)
    (J_W^\dagger WJ_W)\notag\\
&=\frac32U
 -4\left(\frac U2\right)
    \left(\frac{U^\dagger}{2}\right)
    \left(\frac U2\right)
=U.
\label{eq:s5-amplified-matrix}
\end{align}
This matrix identity also forces every auxiliary output to be zero.
For any unit vector $\ket{\vct\psi}$,
\begin{align}
\norm{(\Id-\Pi)\mathcal A J_W\ket{\vct\psi}}^2
&=\norm{\mathcal A J_W\ket{\vct\psi}}^2
 -\norm{\Pi\mathcal A J_W\ket{\vct\psi}}^2=1-\norm{J_WU\ket{\vct\psi}}^2=0.
\label{eq:s5-exact-cleanup}
\end{align}
Thus $\mathcal A J_W=J_WU$ on the subspace with all signal ancillas initialized to zero.

Finally, we give the implementations of the controls and
reflections and count every working qubit.
A circuit of depth $D_0$ and size $S_0$ on $N$ qubits can be
controlled by one additional qubit with depth
$O(D_0+\log(N+1))$, size $O(S_0+N)$, and $O(N)$ ancillary
qubits. To see this, use a fan-out gate to distribute the control
value to $N$ zero qubits. In each original layer, assign distinct
control copies to its disjoint gates. A controlled one- or
two-qubit gate acts on at most three qubits and has an exact
constant-size, constant-depth standard decomposition
\cite{Barenco1995}. The same control copies can be reused in
subsequent layers. Reverse the fan-out gate at the end.
The copies retain their computational-basis values, so they
return to zero for every input of the circuit registers.
Conjugating the control by $X$ implements the branch enabled
by $c=0$.

Applying this construction to $T$ and $V_\Gamma$ implements
Eq.~\eqref{eq:s5-selected-circuit} exactly, including on
nonzero inputs of their ancillary registers. With the
preparations and the additional one-qubit gates, $W$ has
\begin{align*}
\text{depth:}\quad&
 O(D+\log(n+a+1))+O(n+\log(2n+b+1))+O(1)\\
&=O(D+n+\log(a+2)),\\
\text{size:}\quad&
 O(S+n+a)+O(d^2+2n+b)+O(1)=O(S+d^2+a).
\end{align*}
The ancillary control copies require $O(n+a+b)$ qubits.

To implement $R_0$, compute the predicate that all $w$ signal
qubits are zero into a balanced tree of Toffoli gates on
zero-initialized working qubits. Apply $Z$ to the resulting
predicate qubit, then reverse the tree.
This has size $O(w)$, depth $O(\log(w+1))$, and uses
$O(w)$ ancillary qubits. It restores those qubits to zero
for every signal input. These ancillary qubits are separate
from the $w$ qubits tested by $R_0$.
They are untouched by the logical $W$; the control copies are
also cleaned at the end of every controlled operation.
Consequently, the logical identities above apply on the
invariant subspace where these ancillary registers are zero.
Adding them to $J_W$ gives the full embedding $J_m$ in
Eq.~\eqref{eq:s5-exact-action}.

The complete circuit uses two copies of $W$, one literal inverse
$W^\dagger$, two reflections, and one $-\Id$ gate. In particular,
there are three controlled uses of $T$ or $T^\dagger$ and
three controlled uses of $V_\Gamma$ or $V_\Gamma^\dagger$.
Their internal preparations and all inverse gates are included
in these costs. Since $w=O(a+d^2/n)$, we obtain
\begin{align*}
\text{depth:}\quad&
 3O(D+n+\log(a+2))+2O(\log(w+1))+O(1)\\
&=O(D+n+\log(a+2)),\\
\text{size:}\quad&
 3O(S+d^2+a)+2O(w)+O(1)=O(S+d^2+a),\\
\text{ancilla-count:}\quad&
 w+O(n+a+b)+O(w)=O(a+d^2/n).
\end{align*}
Control and reflection workspace can be reused after it is
cleaned; no signal register is discarded or assumed to be
zero between operations.
Equation~\eqref{eq:s5-exact-action} is an equality of
isometries. Tensoring it with the identity on any reference
system proves correctness for entangled inputs as well.
\end{proof}

\begin{theorem}[Restatement of Theorem \ref{thm:main}]
\label{thm:s5-exact}
For every $n\ge1$ and $U\in\mathrm U(2^n)$, there is a
standard circuit $\mathcal A_U$ satisfying
\[
\mathcal A_U\bigl(\ket{\vct\psi}\ket{0^m}\bigr)
=(U\ket{\vct\psi})\ket{0^m},
\]
for every $n$-qubit state $\ket{\vct\psi}$, with depth $O(n)$ and size $O(4^n)$, using $m=O(4^n/n)$ ancillary qubits.
\end{theorem}

\begin{proof}
For $n=1$, apply the single-qubit gate $U$ directly.
Suppose $n\ge2$ and put $d=2^n$.
First, take the approximate circuit $T_U$ from
Theorem~\ref{thm:s4c-native}. Writing $a$ for its total
ancillary count, that theorem gives
\[
\norm{T_UJ_a-J_aU}_{\op}
<2^{-20}d^{-1/2}\le d^{-1/2},
\]
with depth $D=O(n)$, size $S=O(d^2)$, and
$a=O(d^2/n)$. This is the full-output guarantee:
every coordinate, preparation, and group auxiliary qubit
is already included in $J_a$.

Second, apply Lemma~\ref{lem:s5-exactification} to this
specific unitary circuit $T_U$. It gives exact implementation
and exact cleanup. Its resource bounds become
\begin{align*}
\text{depth:}\quad&
 O(D+n+\log(a+2))
 =O(n+n+\log(d^2/n+2))=O(n),\\
\text{size:}\quad&
 O(S+d^2+a)=O(d^2),\\
\text{ancilla-count:}\quad&
 O(a+d^2/n)=O(d^2/n).
\end{align*}
\end{proof}

\section{Optimal Depth--Ancilla Tradeoffs}
\label{sec:tradeoff}

We derive the tradeoff for every ancillary budget from the exact
linear-depth endpoint of Theorem~\ref{thm:main}. The construction
truncates the quantum Shannon decomposition at a suitable leaf size
and reuses one clean ancillary bank throughout.

\vspace{1em}
\noindent \textbf{Quantum Shannon decomposition (QSD) \cite{ShendeBullockMarkov2006}.} Let \(n\ge2\) and \(U\in\mathrm U(2^n)\). The quantum Shannon decomposition is a recursive factorization of \(U\) into four \((n-1)\)-qubit unitaries and three uniformly controlled single-qubit rotations:
\[
\begin{aligned}
U={}&(I_2\otimes U_4)\,
\mathcal R_1\,
(I_2\otimes U_3)\,
\mathcal R_2(I_2\otimes U_2)\,
\mathcal R_3\,
(I_2\otimes U_1),
\end{aligned}
\]
where \(U_j\in\mathrm U(2^{n-1})\) and $\mathcal{R}_1,\mathcal{R}_2$ and $\mathcal{R}_3$ are $n$-qubit uniformly controlled one-qubit gated. The same factorization is applied recursively to \(U_1,\ldots,U_4\).

\begin{theorem}[Depth and lower bound, \cite{SunEtAl2023,ShendeMarkovBullock2004}]
\label{thm:tradeoff-lower}
For all integers $n\ge1$ and $m\ge0$, there exists a unitary $U\in\mathrm U(2^n)$ which can only be implemented  by a standard circuit of size $\Omega (4^n)$ and depth at least
\begin{equation}
 \Omega\!\left(n+\frac{4^n}{n+m}\right),
\end{equation}
using $m$ ancillary qubits.
\end{theorem}

Now we show the optimal depth-ancilla tradeoff for general unitary synthesis.
\begin{theorem}[Restatement of Theorem~\ref{thm:all-budget}]
\label{thm:tradeoff}
For all integers $n\ge1$ and $m\ge0$, every $U\in\mathrm U(2^n)$
admits an exact standard circuit of size $O(4^n)$ and depth
\begin{equation}
 O\!\left(n+\frac{4^n}{n+m}\right),
 \label{eq:tradeoff-bound}
\end{equation}
using at most $m$ ancillary qubits. The size and depth bounds are simultaneously 
optimal.
\end{theorem}

\begin{proof}
For $n=1$, a single one-qubit gate suffices. Assume $n\ge2$.
By Theorem~\ref{thm:s5-exact}, there is an absolute constant
$c_{\mathrm a}\ge1$ such that every $k$-qubit unitary, for $k\ge2$,
has an exact implementation with
\begin{equation}
 D_k=O(k),\qquad S_k=O(4^k),\qquad
 a_k\le c_{\mathrm a}\frac{4^k}{k}.
 \label{eq:tradeoff-leaf-resources}
\end{equation}
Here $a_k$ counts all ancillary registers. Each implementation acts
exactly as its target on arbitrary logical inputs and restores its
ancillas to zero.

If $m\le8c_{\mathrm a}n$, use the known exact ancilla-free synthesis
of size $O(4^n)$ and depth $O(4^n/n)$ by Lemma \ref{lem:US-no-ancilla}.
Since $n+m\le(1+8c_{\mathrm a})n$, this satisfies the claimed bounds.
For $m>8c_{\mathrm a}n$, choose
\begin{equation}
 k=\max\left\{r\in\{2,\ldots,n\}:
               c_{\mathrm a}\frac{4^r}{r}\le m\right\}.
 \label{eq:tradeoff-cutoff}
\end{equation}
The set is nonempty because $r=2$ satisfies the workspace condition,
since $8c_{\mathrm a}\le m$.
If $k=n$, the endpoint itself gives depth $O(n)$ and size $O(4^n)$
using at most $m$ ancillary qubits. Otherwise, maximality gives
\begin{equation}
 c_{\mathrm a}\frac{4^k}{k}\le m
 <c_{\mathrm a}\frac{4^{k+1}}{k+1}
 <4c_{\mathrm a}\frac{4^k}{k},
 \label{eq:tradeoff-rounding}
\end{equation}
and hence $4^k/k=\Theta(m)$.

For $k<n$, apply the exact quantum Shannon decomposition recursively
down to size $k$. Each $j$-qubit unitary, $j>k$, factors into four
ordinary $(j-1)$-qubit unitaries and three $j$-qubit uniformly
controlled one-qubit gates.
The factorization preserves the prescribed matrix phase.
The child unitaries require no added control, so each terminal leaf
can use Theorem~\ref{thm:main} directly.
There are $4^{n-k}$ leaves and $4^{n-j}$ decomposition nodes at
each size $j>k$.
By Lemma~\ref{lem:tradeoff-ucg}, each intermediate rotation uses
at most $m$ ancillary qubits and has size $O(2^j)$ and depth
$O(j+2^j/(j+m))$.

The total size is
\begin{equation}
 S=O\!\left(4^{n-k}4^k+
       \sum_{j=k+1}^{n}4^{n-j}2^j\right)=O(4^n).
 \label{eq:tradeoff-size}
\end{equation}
For the depth, the geometric tails satisfy
\begin{equation*}
 \sum_{j=k+1}^{\infty}\frac{j}{4^j}\le\frac{k}{4^k},
 \qquad
 \sum_{j=k+1}^{n}\frac{1}{2^j(j+m)}
 \le\frac{2^{-k}}{m}\le\frac1m.
\end{equation*}
Consequently,
\begin{align}
 D&=O\!\left(k4^{n-k}+4^n\sum_{j=k+1}^{n}
       \left[\frac{j}{4^j}+\frac{1}{2^j(j+m)}\right]\right)
       =O\!\left(k4^{n-k}+\frac{4^n}{m}\right)
   =O\!\left(\frac{4^n}{m}\right),
 \label{eq:tradeoff-depth}
\end{align}
where the last step uses $4^k/k=\Theta(m)$.
Since $m>8c_{\mathrm a}n\ge n$, this is $O(4^n/(n+m))$.
Together with the preceding cases, this proves the theorem. Both the size and depth match the lower bounds in Theorem \ref{thm:tradeoff-lower}, which are optimal.
\end{proof}

\newpage
\appendix
\renewcommand{\purpose}[1]{\par\noindent\textbf{Purpose.} #1\par}
\renewcommand{\deliverable}[1]{\par\smallskip\noindent\textbf{Required output.} #1\par}
\appendixspace

\appendix 
\section*{Appendix}
\section{Error bounds for general unitary synthesis}
\label{append:core-error}
In this section, we will show the proofs omitted in Section \ref{subsec:core-error}.
\begin{lemma}[Poisson summation, \cite{STEINETAL2003}]
\label{lem:poisson-summation}
Let \(g\in\mathcal{S}(\mathbb{R}^{\ell})\), where
\(\mathcal{S}(\mathbb{R}^{\ell})\) denotes the Schwartz space. For every \(\vct x\in\mathbb{R}^{\ell}\),
\[
    \sum_{m\in\mathbb{Z}^{\ell}} g(\vct x+\vct m)
    =
    \sum_{\vct k\in\mathbb{Z}^{\ell}}
    \widehat{g}(\vct k)e^{2\pi \iu \vct k^{T}\vct x},\qquad\text{where}\quad\widehat{g}(\vct \xi)
    :=
    \int_{\mathbb{R}^{\ell}}
    g(\vct z)e^{-2\pi \iu \vct z^{T}\vct\xi}\,d\vct z.
\]
\end{lemma}
\subsection{Proof of Lemma \ref{lem:core-lattice}}
\label{append:core-lattice}
\begin{lemma}[Restatement of Lemma \ref{lem:core-lattice}]
\label{lem:core-lattice-append}
{ Let $A\in\mathbb C^{\ell\times\ell}$ be symmetric, with
$G=\operatorname{Re}A$ satisfying $\frac1{64}\Id\preceq G\preceq64\Id$ and let $\vct b\in\mathbb C^\ell\setminus\{\vct0\}$.
Define
\[
    f(\vct z)
       =\alpha(\vct b^T\vct z)e^{-\vct z^TA\vct z/2},
    \qquad \vct z\in\mathbb R^\ell,
\]
where $\alpha\in\mathbb C$ is a normalization constant such that $\int_{\mathbb R^\ell}|f(\vct z)|^2\,d\vct z=1$.
For each $0\le j<\ell$, choose a spacing
$a_j\in\{h,\widehat h\}$ and an offset $\sigma_j\in\mathbb R$.
Define the infinite grid points by
\[
    \vct z_{\vct k}
       =\bigl(a_0(k_0+\sigma_0),\ldots,
               a_{\ell-1}(k_{\ell-1}+\sigma_{\ell-1})\bigr)^T,
    \qquad \vct k\in\mathbb Z^\ell.
\]
Then the following bounds hold:
\begin{align}
    & \left|
        \left(\prod_{j=0}^{\ell-1}a_j\right)
        \sum_{\vct k\in\mathbb Z^\ell}
             |f(\vct z_{\vct k})|^2
        -1
        \right|
        \le18\ell e^{-B^2/128}.\label{eq:dir-norm-estimate-append}\\
        & \left(\prod_{j=0}^{\ell-1}a_j\right)
        \sum_{\substack{\vct k\in\mathbb Z^\ell\\
                        |z_{\vct k,i}|\ge R}}
             |f(\vct z_{\vct k})|^2
        \le6e^{-R^2/128},\qquad 0\le \forall i<\ell,~\forall R\ge0.
        \label{eq:core-sampled-tail-append}
\end{align}

    The same tail bound holds for Eq.~\eqref{eq:core-sampled-tail-append} if, for any $j\ne i$, the index
    $k_j$ is restricted to a finite subset of $\mathbb Z$.}
\end{lemma}
\begin{proof}
Define the Gaussian density function
\begin{equation}
    \rho_G(\vct z)=\frac{\sqrt{\det G}}{\pi^{\ell/2}}
                         e^{-\vct z^TG\vct z}.
\end{equation}
Since $\int z_i z_j \rho_G(\vct z)d\vct z=\frac{1}{2}(G^{-1})_{ij}$, we have
\begin{equation}
   \int |\vct b^T \vct z|^2 \rho_G(\vct z)d\vct z=\int \sum_{i,j} b_i\bar{b}_jz_iz_j \rho_G(\vct z)d\vct z=\frac{1}{2}\vct b^\dagger G^{-1} \vct b. 
\end{equation}
Then the normalization of $f$ gives
\begin{equation}
  |f(\vct z)|^2=\frac{2|\vct b^T\vct z|^2}{\vct b^\dagger G^{-1}\vct b}\rho_G(\vct z).
\end{equation}
Define 
\begin{equation}
    \Phi_G(\vct t):=\int \rho_G(\vct z) e^{-\iu \vct t^T \vct z} d\vct z=e^{- \vct t^T G^{-1}\vct t/4}.
\end{equation}
Differentiating $ \Phi_G(\vct t)$ yields
\begin{equation}
    -\frac{\partial^2 \Phi_G(\vct t)}{\partial_{t_i}\partial_{t_j}}=\int z_iz_j \rho_G(\vct z) e^{-\iu \vct t^T\vct z}d\vct z=\Big[\frac{1}{2}(G^{-1})_{ij}-\frac{1}{4}(G^{-1}\vct t)_{i}(G^{-1}\vct t)_j\Big]e^{-\vct t^T G^{-1}\vct t/4}.
\end{equation}
Then we have
\begin{align}
    \chi(\vct t):= &\int |f(\vct z)|^2e^{-\iu\vct t^T\vct z}\,d\vct z \\
    = & \int\frac{2|\vct b^T\vct z|^2}{\vct b^\dagger G^{-1}\vct b}\rho_G(\vct z)e^{-\iu\vct t^T\vct z}\,d\vct z\\
    = & \frac{2}{\vct b^\dagger G^{-1}\vct b}\sum_{i,j}b_i\bar{b}_j\int z_i z_j\rho_G(\vct z)e^{-\iu\vct t^T\vct z}\,d\vct z\\
    =&\frac{2}{\vct b^\dagger G^{-1}\vct b}\sum_{i,j}b_i\bar{b}_j\Big[\frac{1}{2}(G^{-1})_{ij}-\frac{1}{4}(G^{-1}\vct t)_{i}(G^{-1}\vct t)_j\Big]e^{-\vct t^T G^{-1}\vct t/4}\\
       =&\left(1-
          \frac{|\vct b^TG^{-1}\vct t|^2}
               {2\vct b^\dagger G^{-1}\vct b}\right)
          e^{-\vct t^TG^{-1}\vct t/4}   
    \label{eq:core-gaussian-characteristic}
\end{align}
This formula holds for complex $\vct b$.
If $\Id/K\preceq G\preceq K\Id$ with $K=64$, by Cauchy-Schwarz inequality,
\begin{equation}
    \frac{|\vct b^T G^{-1}\vct t|^2}{2\vct b^\dagger G^{-1}\vct b}\le \frac{(b^\dagger G^{-1} \vct b) (t^T G^{-1} t)}{2\vct b^\dagger G^{-1}\vct b}=\frac{u}{2},\qquad u=\vct t^TG^{-1}\vct t\ge 0.
\end{equation}
Then 
\begin{equation}\label{eq:chi-upper}
    |\chi(\vct t)|\le (1+u/2)e^{-u/4}\le3e^{-u/8}
       \le3e^{-\norm{\vct t}^2/(8K)},
\end{equation}
The last inequality holds because $u\ge \norm{t}^2/K$ induced by $I/K\preceq G^{-1}$.
Apply Poisson summation to $|f|^2$ (Lemma \ref{lem:poisson-summation}):
\begin{align}
    \sum_{\vct k\in\Z^\ell} |f(\vct z_{\vct k})|^2=&\frac{1}{\prod_j a_j}\sum_{\vct m\in \Z^\ell} e^{\iu 2\pi\vct m^T\vct \sigma}\chi(\vct t_{\vct m}), \qquad \vct t_{\vct m,j}=2\pi m_j/a_j\\
    =&\frac{1}{\prod_j a_j}\left(\sum_{\vct m\in \Z^\ell\atop \vct m\neq\vct 0} e^{\iu 2\pi\vct m^T\vct \sigma}\chi(\vct t_{\vct m})+1\right),\label{eq:poisson-sum}
\end{align}
where the last equality holds since $\chi(\vct 0)=1$. 
Therefore, 
\begin{align}
    \left|
        \left(\prod_{j=0}^{\ell-1}a_j\right)
        \sum_{\vct k\in\mathbb Z^\ell}
             |f(\vct z_{\vct k})|^2
        -1
        \right|
         \le & \sum_{\vct m\in\Z^\ell\atop \vct m\neq\vct 0}|\chi(\vct t_{\vct m})| & \text{Eq.~\eqref{eq:poisson-sum}}\\
        \le & 3\sum_{\vct m\in\Z^\ell\atop \vct m\neq\vct 0} e^{-\norm{\vct t_{\vct m}}^2/(8K)} & \text{Eq.~\eqref{eq:chi-upper}}\\
        \le & 3\sum_{\vct m\in\Z^\ell\atop \vct m\neq\vct 0} e^{-B^2\norm{\vct m}^2/(2K)}& \text{($2\pi/a_j\ge 2B$)}\\
        =&3\left[\left(\sum_{m \in\Z} e^{-\frac{B^2}{2K}m^2}\right)^\ell-1\right]\label{eq:bound}
\end{align}
Since $m^2\ge 1+ 3(m-1)$ for any $m\ge 1$ and  $B^2\ge 2^{20}(n+20)$, we have 
\begin{align}
    \sum_{m\in \Z} e^{-\frac{B^2}{2K}m^2_j}=& 1+2\sum_{m_j=1}^\infty e^{-\frac{B^2}{2K}m^2_j} 
    \le  1+2\sum_{m_j=1}^\infty e^{-\frac{B^2}{2K}(1+3(m_j-1))}\\
    =&1+\frac{2e^{-\frac{B^2}{2K}}}{1-e^{-\frac{3B^2}{2K}}}\le 1+\frac{2e^{-\frac{B^2}{2K}}}{1-1/3}=1+3e^{-\frac{B^2}{2K}}\label{eq:inequality1}
\end{align}
For $\ell x\le 1$,
\begin{equation}
    (1+x)^\ell-1\le e^{\ell x}-1=\sum_{r=1}^\infty \frac{(\ell x)^r}{r!}\le (\ell x)\sum_{r=1}^{\infty}\frac{1}{r!}=(e-1)(\ell x)\le 2\ell x.\label{eq:inequality2}
\end{equation}
Recall that $\ell=2^{n+2}$. We can verify that $\ell\cdot 3e^{-\frac{B^2}{2K}}\le 1$. Use Eqs.~\eqref{eq:inequality1} and \eqref{eq:inequality2}, then 
\begin{equation}
\text{Eq.~\eqref{eq:bound}}\le 3\left[(1+3e^{-\frac{B^2}{2K}})^\ell-1\right]\le18\ell e^{-B^2/(2K)}.
    \label{eq:core-theta-product}
\end{equation}
This proves Eq.~\eqref{eq:dir-norm-estimate-append}.

{ Fix $0\le i<\ell$ and $R\ge0$.
Let $\vct e_i\in\R^\ell$ be the $i$th coordinate unit vector,
and define
\begin{equation}
    G'=G-\frac{\vct e_i\vct e_i^T}{2K}.
\end{equation}
Since $\vct e_i\vct e_i^T\preceq\Id$ and
$G\succeq\Id/K$, we have
\begin{align}
    G'
    &\succeq G-\frac{\Id}{2K}
     \succeq \frac{G}{2}
     \succeq \frac{\Id}{2K},\\
    G'
    &\preceq G
     \preceq K\Id
     \preceq 2K\Id.
\end{align}
In particular, $G'$ is positive definite.
Taking inverses reverses the order of positive definite matrices,
so $G'\succeq G/2$ gives
\begin{equation}
    (G')^{-1}\preceq2G^{-1},
    \qquad
    \frac{\vct b^\dagger(G')^{-1}\vct b}
         {\vct b^\dagger G^{-1}\vct b}
    \le2.
\end{equation}
Moreover, $G^{-1}\preceq K\Id$ implies
$(G^{-1})_{ii}\le K$. By the matrix determinant lemma,
\begin{align}
    \det G'
    =&\det G
      \left(1-\frac{\vct e_i^TG^{-1}\vct e_i}{128}\right)
      =\det G
      \left(1-\frac{(G^{-1})_{ii}}{128}\right),\\
    \frac{\det G}{\det G'}
    =&\frac{1}{1-(G^{-1})_{ii}/128}
      \le\frac{1}{1-64/128}
      =2.
\end{align}

Define the weighted density and its integral by
\begin{equation}
    w_i(\vct z)
       :=|f(\vct z)|^2e^{z_i^2/128},
    \qquad
    M_i:=\int_{\R^\ell}w_i(\vct z)\,d\vct z.
\end{equation}
Since
\begin{equation}
    w_i(\vct z)
       =|\alpha|^2|\vct b^T\vct z|^2
          e^{-\vct z^TG'\vct z},
\end{equation}
the Gaussian integral formula gives
\begin{align}
    M_i
    &=|\alpha|^2
      \int_{\R^\ell}|\vct b^T\vct z|^2
          e^{-\vct z^TG'\vct z}\,d\vct z\\
    &=|\alpha|^2
      \frac{\pi^{\ell/2}}{2\sqrt{\det G'}}
      \vct b^\dagger(G')^{-1}\vct b\\
    &=\sqrt{\frac{\det G}{\det G'}}
      \frac{\vct b^\dagger(G')^{-1}\vct b}
           {\vct b^\dagger G^{-1}\vct b}\\
    &\le\sqrt2\cdot2=2\sqrt2.
\end{align}
Here the third equality uses the normalization identity
\begin{equation}
    1=\int_{\R^\ell}|f(\vct z)|^2\,d\vct z
     =|\alpha|^2
       \frac{\pi^{\ell/2}}{2\sqrt{\det G}}
       \vct b^\dagger G^{-1}\vct b.
\end{equation}

Next, define the normalized function
\begin{equation}
    \widetilde f_i(\vct z)
       :=\frac{f(\vct z)e^{z_i^2/256}}{\sqrt{M_i}}
       =\frac{\alpha}{\sqrt{M_i}}
         (\vct b^T\vct z)e^{-\vct z^TA'\vct z/2},
    \qquad
    A'=A-\frac{\vct e_i\vct e_i^T}{128}.
\end{equation}
Then
\begin{equation}
    |\widetilde f_i(\vct z)|^2
       =\frac{w_i(\vct z)}{M_i},
    \qquad
    \int_{\R^\ell}|\widetilde f_i(\vct z)|^2\,d\vct z=1,
    \qquad
    \operatorname{Re}A'=G'.
\end{equation}
Thus $\widetilde f_i$ has the same linear-Gaussian form, with
$\Id/2K\preceq G'\preceq 2K\Id$.

By the similar discussion of Eq.~\eqref{eq:core-theta-product}, we have
\begin{equation}
    \left|
    \left(\prod_j a_j\right)
    \sum_{\vct k\in\Z^\ell}
       |\widetilde f_i(\vct z_{\vct k})|^2-1
    \right|
    \le18\ell e^{-B^2/4K}<\frac12.
\end{equation}
Therefore, the weighted density satisfies
\begin{align}
    \left(\prod_j a_j\right)
    \sum_{\vct k\in\Z^\ell}w_i(\vct z_{\vct k})
    &=
    M_i\left(\prod_j a_j\right)
    \sum_{\vct k\in\Z^\ell}
       |\widetilde f_i(\vct z_{\vct k})|^2\\
    &\le M_i\left(1+\frac12\right)\le2\sqrt2\left(1+\frac12\right)
     <6.
\end{align}

Finally, on the set $|z_{\vct k,i}|\ge R$,
\begin{equation}
    |f(\vct z_{\vct k})|^2
       =e^{-z_{\vct k,i}^2/2K}w_i(\vct z_{\vct k})
       \le e^{-R^2/2K}w_i(\vct z_{\vct k}).
\end{equation}
Hence
\begin{align}
    \left(\prod_j a_j\right)
    \sum_{\substack{\vct k\in\Z^\ell\\
                    |z_{\vct k,i}|\ge R}}
       |f(\vct z_{\vct k})|^2 \le&
    e^{-R^2/2K}
    \left(\prod_j a_j\right)
    \sum_{\substack{\vct k\in\Z^\ell\\
                    |z_{\vct k,i}|\ge R}}
       w_i(\vct z_{\vct k})\\
    &\quad\le
    e^{-R^2/2K}
    \left(\prod_j a_j\right)
    \sum_{\vct k\in\Z^\ell}w_i(\vct z_{\vct k})\\
    &\quad\le6e^{-R^2/2K}.
\end{align}
The second inequality holds because $w_i\ge0$, so extending
the summation range can only increase the sum.
This proves Eq.~\eqref{eq:core-sampled-tail-append}.
Restricting any of the other coordinate labels to a finite set
only removes nonnegative summands, so the same tail bound
remains valid.}
\end{proof}

\subsection{Proof of Lemma \ref{lem:core-one-fourier}}
\label{append:core-one-fourier}
\begin{lemma}[Restatement of Lemma \ref{lem:core-one-fourier}]
\label{lem:core-one-fourier-append}
Let $D_i$ and $\mathcal{F}_i$ denote the discrete and continuous one-coordinate quantum Fourier transform acting on the $i$th coordinate. Suppose $f$ and
$\mathcal F_i f$ satisfy Lemma~\ref{lem:core-lattice}.
Changing coordinate $i$ from $(h,1/2)$ to $(\widehat h,0)$ gives
\begin{equation}
    \norm{D_i\mathsf S_{\vct a,\vct\sigma}f
          -\mathsf S_{\vct a',\vct\sigma'}\mathcal F_i f}
       \le12e^{-B^2/256}.
    \label{eq:core-one-fourier-error-append}
\end{equation}
where $\vct a,\vct a',\vct \sigma, \vct \sigma'$ satisfy
\begin{equation}
    (a_i,\sigma_i)=(h,1/2),
    \qquad
    (a_i',\sigma_i')=(\widehat h,0),
    \qquad
    (a_j',\sigma_j')=(a_j,\sigma_j)
    \quad \forall j\ne i.
\end{equation}
The analogous bound
holds for $D_i^\dagger$ and the inverse continuous transform.
\end{lemma}
{  \begin{proof}
Write $x$ for coordinate $i$ and $\vct z'$ for the remaining
coordinates. Define
\begin{equation}
    g(\zeta,\vct z')
       :=(\mathcal F_i f)(\zeta,\vct z')
       =\frac{1}{\sqrt{2\pi}}
         \int_{\R}f(x,\vct z')e^{-\iu\zeta x}\,dx.
\end{equation}
The input and output grids satisfy
\begin{equation}
    (a_i,\sigma_i)=(h,1/2),
    \qquad
    (a_i',\sigma_i')=(\widehat h,0),
    \qquad
    (a_j',\sigma_j')=(a_j,\sigma_j)
    \quad(j\ne i).
\end{equation}
For the remaining coordinate labels
$\vct k'\in\mathcal I^{\ell-1}$, define
\begin{equation}
    \vct z'_{\vct k'}
       =\bigl(a_j(k_j+\sigma_j)\bigr)_{j\ne i},
    \qquad
    \gamma=\prod_{j\ne i}\sqrt{a_j}.
\end{equation}
We regard finite sampling vectors as their coefficient arrays
in the computational basis. For every $m\in\Z$, define
\begin{align}
    u^{(m)}_{r,\vct k'}
       &=\gamma\sqrt h\,
         f\bigl(h(r+1/2+mQ),\vct z'_{\vct k'}\bigr),
         & r\in\mathcal I,\\
    v^{(m)}_{k,\vct k'}
       &=\gamma\sqrt{\widehat h}\,
         g\bigl(\widehat h(k+mQ),\vct z'_{\vct k'}\bigr),
         & k\in\mathcal I.
\end{align}
In particular,
\begin{equation}
    \vct u^{(0)}=\mathsf S_{\vct a,\vct\sigma}f,
    \qquad
    \vct v^{(0)}
       =\mathsf S_{\vct a',\vct\sigma'}\mathcal F_i f.
\end{equation}
Recall that
\begin{equation}
    Qh=2B,
    \qquad
    Q\widehat h=2\pi B,
    \qquad
    h\widehat h=\frac{2\pi}{Q}.
\end{equation}

First, introduce the comparison vector
\begin{equation}
    \vct u^{\mathrm{fold}}
       :=\sum_{m\in\Z}\vct u^{(m)}.
\end{equation}
For $m\ne0$, let $R_m=(2|m|-1)B$. Since
$|r+1/2|\le Q/2$ for $r\in\mathcal I$, the reverse triangle inequality gives
\begin{align}
    \bigl|h(r+1/2+mQ)\bigr|
       &\ge h\bigl(|m|Q-|r+1/2|\bigr) \ge hQ(|m|-1/2)
        =R_m.
\end{align}
Therefore,
\begin{align}
    \norm{\vct u^{(m)}}^2
    &=
    \gamma^2h
    \sum_{\vct k'\in\mathcal I^{\ell-1}}
    \sum_{r\in\mathcal I}
       \left|
       f\bigl(h(r+1/2+mQ),\vct z'_{\vct k'}\bigr)
       \right|^2\\
    &\le
    \gamma^2h
    \sum_{\vct k'\in\mathcal I^{\ell-1}}
    \sum_{\substack{s\in\Z\\|h(s+1/2)|\ge R_m}}
       \left|
       f\bigl(h(s+1/2),\vct z'_{\vct k'}\bigr)
       \right|^2 & \text{($s:=r+mQ$)}\\
    &\le6e^{-R_m^2/128}.
\end{align}
The last inequality follows from
Eq.~\eqref{eq:core-sampled-tail}, since
$\gamma^2h=\prod_j a_j$; the other coordinate labels may remain
restricted to $\mathcal I$.
Taking square roots gives
\begin{equation}
    \norm{\vct u^{(m)}}
       \le\sqrt6\,
          e^{-(2|m|-1)^2B^2/256},
    \qquad m\ne0.
\end{equation}
Put $c=B^2/256\ge1$. For every integer $m\ge1$,
\begin{equation}
    (2m-1)^2-\bigl(1+8(m-1)\bigr)
       =4(m-1)(m-2)\ge0.
\end{equation}
Consequently,
\begin{align}
    \sum_{m=1}^{\infty}e^{-(2m-1)^2c}
       &\le e^{-c}
          \sum_{m=1}^{\infty}e^{-8c(m-1)} =\frac{e^{-c}}{1-e^{-8c}}.
\end{align}
Using the triangle inequality, we obtain
\begin{align}
    \norm{\vct u^{\mathrm{fold}}-\vct u^{(0)}}
       &=
       \left\|\sum_{m\ne0}\vct u^{(m)}\right\|\le\sum_{m\ne0}\norm{\vct u^{(m)}}\le2\sqrt6
          \sum_{m=1}^{\infty}e^{-(2m-1)^2c}\le\frac{2\sqrt6}{1-e^{-8c}}e^{-c}\le6e^{-c}.
    \label{eq:core-folding-bound}
\end{align}
The last inequality uses $2\sqrt6\le5$ and
$e^{-8c}\le1/6$, which hold because $c\ge1$.
These estimates also show that the series defining
$\vct u^{\mathrm{fold}}$ converges absolutely in the finite
coefficient norm.

Second, apply $D_i$ to this folded vector. For
$k\in\mathcal I$ and $\vct k'\in\mathcal I^{\ell-1}$,
\begin{align}
    (D_i\vct u^{\mathrm{fold}})_{k,\vct k'}
    &=
    \frac{\gamma\sqrt h}{\sqrt Q}
    \sum_{r\in\mathcal I}\sum_{m\in\Z}
       f\bigl(h(r+1/2+mQ),\vct z'_{\vct k'}\bigr)
       e^{-2\pi\iu k(r+1/2)/Q}\\
    &=
    \frac{\gamma\sqrt h}{\sqrt Q}
    \sum_{s\in\Z}
       f\bigl(h(s+1/2),\vct z'_{\vct k'}\bigr)
       e^{-2\pi\iu k(s+1/2)/Q}.
\end{align}
For the second equality, every integer $s$ has a unique
representation $s=r+mQ$ with $r\in\mathcal I$, and
\begin{equation}
    e^{-2\pi\iu k(r+1/2)/Q}
       =e^{-2\pi\iu k(s+1/2)/Q}e^{2\pi\iu km}
       =e^{-2\pi\iu k(s+1/2)/Q}.
\end{equation}

For each fixed $\vct z'$, apply Poisson summation to
$x\mapsto f(x,\vct z')e^{-\iu\widehat h kx}$.
Its continuous Fourier transform at $\zeta$ is
$g(\zeta+\widehat h k,\vct z')$. Hence
\begin{equation}
\begin{aligned}
    &\frac{h}{\sqrt{2\pi}}
      \sum_{s\in\Z}
       f\bigl(h(s+1/2),\vct z'\bigr)
       e^{-2\pi\iu k(s+1/2)/Q}\\
    =&
      \sum_{m\in\Z}e^{2\pi\iu m(1/2)}
       g\left(\widehat h k+\frac{2\pi m}{h},\vct z'\right)\\
    =&
      \sum_{m\in\Z}(-1)^m
       g\bigl(\widehat h(k+mQ),\vct z'\bigr).
\end{aligned}
\label{eq:core-poisson-fourier}
\end{equation}
Here $h\widehat h=2\pi/Q$ gives the exponent on the left,
and $2\pi/h=Q\widehat h$ gives the last equality.
The factor $(-1)^m$ comes from the midpoint offset $1/2$.

Since
\begin{equation}
    \frac{\sqrt h}{\sqrt Q}
       =\sqrt{\widehat h}\,\frac{h}{\sqrt{2\pi}},
\end{equation}
Eq.~\eqref{eq:core-poisson-fourier} yields the exact vector identity
\begin{equation}
    D_i\vct u^{\mathrm{fold}}
       =\sum_{m\in\Z}(-1)^m\vct v^{(m)}.
\end{equation}
The term $m=0$ is precisely the desired output vector
$\vct v^{(0)}$.

For $m\ne0$ and $k\in\mathcal I$, we have
\begin{align}
    |\widehat h(k+mQ)|
       &\ge\widehat h\bigl(|m|Q-|k|\bigr)\ge\widehat hQ(|m|-1/2)=(2|m|-1)\pi B
        =\pi R_m.
\end{align}
Applying Lemma \ref{lem:core-lattice} (Eq.~\eqref{eq:core-sampled-tail}) to $g$, with the
output grid and the same finite sets of other coordinate labels,
gives
\begin{align}
    \norm{\vct v^{(m)}}^2
       &=
       \gamma^2\widehat h
       \sum_{\vct k'\in\mathcal I^{\ell-1}}
       \sum_{k\in\mathcal I}
       \left|
       g\bigl(\widehat h(k+mQ),\vct z'_{\vct k'}\bigr)
       \right|^2
       \le6e^{-\pi^2R_m^2/128}
       \le6e^{-R_m^2/128}.
\end{align}
Thus
\begin{equation}
    \norm{\vct v^{(m)}}
       \le\sqrt6\,
          e^{-(2|m|-1)^2B^2/256},
    \qquad m\ne0.
\end{equation}
Since $|(-1)^m|=1$, the triangle inequality and the same
series estimate used in Eq.~\eqref{eq:core-folding-bound} give
\begin{align}
    \norm{D_i\vct u^{\mathrm{fold}}-\vct v^{(0)}}
       &=
       \left\|
       \sum_{m\ne0}(-1)^m\vct v^{(m)}
       \right\| \le\sum_{m\ne0}\norm{\vct v^{(m)}} \le2\sqrt6
          \sum_{m=1}^{\infty}e^{-(2m-1)^2B^2/256} \le6e^{-B^2/256}.
\end{align}

Finally, $D_i$ is unitary. Therefore,
\begin{align}
    &\norm{
       D_i\mathsf S_{\vct a,\vct\sigma}f
       -\mathsf S_{\vct a',\vct\sigma'}\mathcal F_i f}=
       \norm{D_i\vct u^{(0)}-\vct v^{(0)}}\le
       \norm{D_i(\vct u^{(0)}-\vct u^{\mathrm{fold}})}
       +\norm{D_i\vct u^{\mathrm{fold}}-\vct v^{(0)}}\\
    =&\norm{\vct u^{(0)}-\vct u^{\mathrm{fold}}}+\norm{D_i\vct u^{\mathrm{fold}}-\vct v^{(0)}}\le 6e^{-B^2/256}+6e^{-B^2/256}=12e^{-B^2/256}.
\end{align}
This proves Eq.~\eqref{eq:core-one-fourier-error-append}.

For the inverse transformation, recall that
$g=\mathcal F_i f$, so $\mathcal F_i^\dagger g=f$.
Unitarity of $D_i$ gives
\begin{align}
    &\norm{
       D_i^\dagger\mathsf S_{\vct a',\vct\sigma'}g
       -\mathsf S_{\vct a,\vct\sigma}\mathcal F_i^\dagger g} =
       \norm{D_i^\dagger\vct v^{(0)}-\vct u^{(0)}} =
       \norm{\vct v^{(0)}-D_i\vct u^{(0)}} \le12e^{-B^2/256}.
\end{align}
\end{proof}}

\subsection{Proof of Lemma \ref{lem:red34-normalization}}
\label{append:red34-normalization}
\begin{lemma}[Restatement of Lemma \ref{lem:red34-normalization}]
\label{lem:red34-normalization-append}
For any $n$-qubit state $\ket{\psi}=\sum_{j=0}^{d-1}c_j\ket{j}$ with $\vct c=(c_0,c_1,\ldots,c_{d-1})^T\in\C^d$, define
\begin{equation}
       \ket{\Psi_Q(\vct c)} =\frac{1}{\sqrt2}\sum_{j=0}^{d-1}c_j(\ket{\vct\Phi_{j,Q}} -\iu\ket{\vct\Phi_{d+j,Q}}),
  \end{equation}
  where 
  \begin{equation}
      \ket{\vct\Phi_{t,Q}} =\ket{\vct e_Q}_t \bigotimes_{\substack{0\le\nu<\ell\\ \nu\ne t}} \ket{\vct g_Q}_\nu, \qquad 0\le t<\ell.
  \end{equation}
Let $f_{\vct b}$ be the wave function defined in Eq.~\eqref{eq:gaussian-linear-wavefunction}. Then 
\begin{equation}
 \mathsf S f_{V'V\vct c}
      =\rho_Q\ket{\Psi_Q(\vct c)},\qquad
 \rho_Q^2=2h^\ell\pi^{-\ell/2}Z_1Z_0^{\ell-1},
\label{eq:red34-normalization}
\end{equation}
where $\rho_Q>0$  satisfies
\begin{equation}
 \rho_Q^2\ge1-24\ell e^{-B^2/128}\ge\tfrac14.
\label{eq:red34-normalization-lower}
\end{equation}
\end{lemma}
\begin{proof}
Define the continuous basis functions
\[
\phi_j(\vct z)
=\sqrt2\,z_jG_\ell(\vct z)
=\sqrt2\,\pi^{-\ell/4}z_j
 e^{-\norm{\vct z}^2/2},
\qquad 0\le j<\ell.
\]
For $\vct k\in\mathcal I^\ell$, the sampling definition gives
\begin{equation}
\langle\vct k|\mathsf S\phi_j\rangle
=h^{\ell/2}\phi_j(\vct t_{\vct k})
=\sqrt2\,h^{\ell/2}\pi^{-\ell/4}
  t_{k_j}
  e^{-\frac12\sum_{\nu=0}^{\ell-1}t_{k_\nu}^2},\qquad \vct t_{\vct k}
=(t_{k_0},\ldots,t_{k_{\ell-1}})^T
\end{equation}
On the other hand, the product structure of
$\ket{\vct\Phi_{j,Q}}$ gives
\begin{equation}
\langle\vct k\mid\vct\Phi_{j,Q}\rangle
=
\frac{t_{k_j}e^{-t_{k_j}^2/2}}{\sqrt{Z_1}}
\prod_{\substack{0\le\nu<\ell\\\nu\ne j}}
\frac{e^{-t_{k_\nu}^2/2}}{\sqrt{Z_0}}\\
=
\frac{
t_{k_j}e^{-\frac12\sum_{\nu=0}^{\ell-1}t_{k_\nu}^2}
}{
\sqrt{Z_1}\,Z_0^{(\ell-1)/2}
}.
\end{equation}
Comparing these amplitudes yields
\[
\mathsf S\phi_j
=\rho_Q\ket{\vct\Phi_{j,Q}},
\qquad
\rho_Q
=\sqrt2\,h^{\ell/2}\pi^{-\ell/4}
 \sqrt{Z_1}\,Z_0^{(\ell-1)/2}>0.
\]
The factor $\rho_Q$ is independent of $j$. In particular,
\[
\rho_Q^2
=2h^\ell\pi^{-\ell/2}Z_1Z_0^{\ell-1}.
\]
By the definitions of $V'$ and $V$, we have $V'V\vct c
=\frac1{\sqrt2}
\begin{pmatrix}
\vct c\\
-\iu\vct c\\
\vct0_p
\end{pmatrix}.$
Thus
$ 
f_{V'V\vct c}
=
\frac1{\sqrt2}\sum_{j=0}^{d-1}
c_j\bigl(\phi_j-\iu\phi_{d+j}\bigr)$.
By linearity of $\mathsf S$,
\begin{equation}
\mathsf S f_{V'V\vct c}
=
\frac1{\sqrt2}\sum_{j=0}^{d-1}
c_j\bigl(\mathsf S\phi_j-\iu\mathsf S\phi_{d+j}\bigr)
=
\frac{\rho_Q}{\sqrt2}\sum_{j=0}^{d-1}
c_j\bigl(
\ket{\vct\Phi_{j,Q}}
-\iu\ket{\vct\Phi_{d+j,Q}}
\bigr)
=\rho_Q\ket{\Psi_Q(\vct c)}.
\end{equation}
This proves Eq.~\eqref{eq:red34-normalization}, including
the independence of $\rho_Q$ from $\vct c$.
Morevoer, we establish the lower bound on $\rho_Q^2$.
The Gaussian integrals
\[
\int_{\mathbb R}e^{-x^2}\,dx=\sqrt\pi,
\qquad
\int_{\mathbb R}x^2e^{-x^2}\,dx=\frac{\sqrt\pi}{2}
\]
give
\begin{equation}
\norm{\phi_j}^2
=
2\pi^{-\ell/2}
\int_{\mathbb R^\ell}
z_j^2e^{-\sum_{\nu=0}^{\ell-1}z_\nu^2}\,d\vct z=2\pi^{-\ell/2}
\frac{\sqrt\pi}{2}
(\sqrt\pi)^{\ell-1}=1.
\end{equation}
Its quadratic matrix is $\Id$, so $\phi_j$ satisfies
Lemma~\ref{lem:core-lattice}.
Extend the definition $t_k=h(k+\tfrac12)$ to all
$k\in\mathbb Z$, and define
\begin{equation}
M_j
=h^\ell
  \sum_{\vct k\in\mathbb Z^\ell}
  |\phi_j(\vct t_{\vct k})|^2,\qquad
T_j
=h^\ell
  \sum_{\vct k\in\mathbb Z^\ell\setminus\mathcal I^\ell}
  |\phi_j(\vct t_{\vct k})|^2.
\end{equation}
Here $M_j$ is the sampled mass on the infinite grid,
and $T_j$ is the mass omitted by the finite grid.
Eq.~\eqref{eq:dir-norm-estimate} gives
\[
|M_j-1|\le18\ell e^{-B^2/128},
\qquad
M_j\ge1-18\ell e^{-B^2/128}.
\]

If $\vct k\notin\mathcal I^\ell$, then $k_i\notin\mathcal I$
for at least one coordinate $i$. Such an index satisfies $k_i\ge Q/2
$ or $ k_i\le-Q/2-1.$ 
Since $Qh=2B$, it follows that $|t_{k_i}|
=\left|h(k_i+\tfrac12)\right|
\ge B+\frac h2
\ge B$.
Every omitted grid point therefore belongs to at least one
coordinate tail. All summands are nonnegative, so
\begin{equation}
T_j\le
\sum_{i=0}^{\ell-1}
h^\ell \sum_{\substack{\vct k\in\mathbb Z^\ell
                |t_{k_i}|\ge B}}
|\phi_j(\vct t_{\vct k})|^2 \le\sum_{i=0}^{\ell-1}6e^{-B^2/128}=6\ell e^{-B^2/128}.
\end{equation}
The second inequality is
Eq.~\eqref{eq:core-sampled-tail} with $R=B$.
Since $\mathsf S\phi_j=\rho_Q\ket{\vct\Phi_{j,Q}}$
and $\ket{\vct\Phi_{j,Q}}$ is normalized, we obtain
\begin{equation}
\rho_Q^2
=\norm{\mathsf S\phi_j}^2
=M_j-T_j
\ge
\left(1-18\ell e^{-B^2/128}\right)
-6\ell e^{-B^2/128}
=1-24\ell e^{-B^2/128}.
\end{equation}

Finally, the parameter choices imply $\frac{B^2}{128} \ge8192(n+20)\ge n+7$.
Using $\ell=4\cdot2^n$ and $e>2$, we have $24\ell e^{-B^2/128}
\le96\cdot2^n e^{-(n+7)}
\le96\cdot2^n 2^{-(n+7)}=\frac34.$
Therefore,
\[
\rho_Q^2
\ge1-24\ell e^{-B^2/128}
\ge\frac14,
\]
which proves Eq.~\eqref{eq:red34-normalization-lower}.
\end{proof}

\subsection{Proof of Lemma \ref{lem:red34-gaussian-closure}}
\label{append:red34-gaussian-closure}
\begin{lemma}[Restatement of Lemma \ref{lem:red34-gaussian-closure}]
\label{lem:red34-gaussian-closure-append}
Let $\vct c\in\mathbb C^d$ satisfy
$\norm{\vct c}=1$, and let $f_{\vct b}(\vct z)=\sqrt2\,(\vct b^T\vct z)G_\ell(\vct z)$ be the input wavefunction of $F_U$ for $\vct b=V'V\vct c$.
For each $M\in \mathcal{S}:=\{ \Id,\, J_0,\, T_{R_U}J_0,\,
S_{-R_U^T}T_{R_U}J_0,\, F_U \}$,
define
\[
f_M(\vct z) =f_{\vct b}(M^{-1}\vct z).
\]
Define a partial Fourier transform
\[
\mathcal F_S=\prod_{i\in S}\mathcal F_i,
\qquad
\mathcal F_{\varnothing}=\Id,\qquad S\subseteq\{0,\ldots,p-1\} \quad \text{or}\quad S\subseteq\{p,\ldots,\ell-1\}. 
\]
Then there exist $\alpha\in\mathbb C\setminus\{0\}$,
$\vct\beta\in\mathbb C^\ell\setminus\{0\}$, and
$A\in\mathbb C^{\ell\times\ell}$ such that
\[
(\mathcal F_Sf_M)(\vct w)
=
\alpha(\vct\beta^T\vct w)
e^{-\vct w^TA\vct w/2},
\]
where $\vct w=(\vct\zeta_S,\vct z_{T})$ and
$T=\{0,\ldots,\ell-1\}\setminus S$.
Moreover,
\[
\norm{\mathcal F_Sf_M}=1,
\qquad
A^T=A,
\qquad
\Id/64\preceq\operatorname{Re}A\preceq64\Id.
\]
\end{lemma}
\begin{proof}
First, we show that $f_M$ satisfies the assumption of Lemma \ref{lem:core-lattice}. 
For the orthogonal matrix $R_U\in\R^{p\times p}$,
\begin{equation}
T_{R_U
}=
\begin{pmatrix}
\Id&R_U\\
0&\Id
\end{pmatrix},
\qquad
T_{R_U}^{-1}=T_{-R_U},\qquad S_{-R^T_U
}=
\begin{pmatrix}
\Id& 0\\
-R_U^T&\Id
\end{pmatrix},
\qquad S_{-R_U^T}^{-1}=S_{R_U^T}.
\end{equation}
Hence
\[
\norm{T_{R_U}^{\pm1}}_{\op}
\le1+\norm{R_U}_{\op}=2,\qquad \norm{S_{-R_U^T}^{\pm1}}_{\op}
\le1+\norm{-R_U^T}_{\op}=2.
\]
Both $T_{R_U}^{\pm}$ and $S_{-R_U^T}^{\pm}$ have determinant one.
Moreover, $J_0$ is orthogonal, and all these coordinate
operations have determinant one. Therefore, every matrix $M$ satisfies
\begin{equation}
   \norm M_{\op}\le8,
\qquad
\norm{M^{-1}}_{\op}\le8,
\qquad
\det M=1 
\end{equation}

The coefficient embedding $V'V$ is an isometry, so
$\norm{\vct b}=1$ and $\norm{f_{\vct b}}=1$.
Then
\[
f_M(\vct z)
=
\alpha(\vct\beta^T\vct z)e^{-\vct z^TA\vct z/2},
\]
where
\[
\alpha=\sqrt2\,\pi^{-\ell/4},
\qquad
\vct\beta=M^{-T}\vct b,
\qquad
A=M^{-T}M^{-1}.
\]
The vector $\vct\beta$ is nonzero because $M$ is invertible.
Using the change of variables $\vct z=M\vct u$, we obtain
\begin{equation}
\norm{f_M}^2=\int_{\mathbb R^\ell}
  |f_{\vct b}(M^{-1}\vct z)|^2\,d\vct z=|\det M|
  \int_{\mathbb R^\ell}|f_{\vct b}(\vct u)|^2\,d\vct u=1.
\end{equation}
Furthermore, for every $\vct z\in\mathbb R^\ell$,
\begin{align}
\vct z^TA\vct z
&=\norm{M^{-1}\vct z}^2\le\norm{M^{-1}}_{\op}^2\norm{\vct z}^2
\le64\norm{\vct z}^2,\\
\vct z^TA\vct z
&=\norm{M^{-1}\vct z}^2\ge\frac{\norm{\vct z}^2}{\norm M_{\op}^2}
\ge\frac{\norm{\vct z}^2}{64}
\end{align}
Thus $A$ is real symmetric and
\[
\Id/64\preceq A\preceq64\Id.
\]

Second, we show that the wavefunction obtained by applying the partial Fourier transform to \(f_M\) also satisfies the assumption of Lemma \ref{lem:core-lattice}.  Let $S$ be a nonempty subset of one coordinate block,
and let $T$ be its complement.
The case $S=\varnothing$ requires no transformation.
Place the coordinates in $S$ first, and write
\[
\vct z=
\begin{pmatrix}\vct z_S\\ \vct z_T\end{pmatrix},
\qquad
A=
\begin{pmatrix}
A_{SS}&A_{ST}\\
A_{TS}&A_{TT}
\end{pmatrix},
\qquad
\vct\beta=
\begin{pmatrix}\vct\beta_S\\ \vct\beta_T\end{pmatrix}.
\]
Then we can obtain

\begin{equation}
\vct z^TA\vct z
=
\bigl(\vct z_S+A_{SS}^{-1}A_{ST}\vct z_T\bigr)^T
A_{SS}
\bigl(\vct z_S+A_{SS}^{-1}A_{ST}\vct z_T\bigr)+\vct z_T^TE\vct z_T.
\end{equation}
where $E=A_{TT}-A_{TS}A_{SS}^{-1}A_{ST}$.

Define the Gaussian integral
\[
I_S(\vct\zeta_S,\vct z_T)
=
\frac{1}{(2\pi)^{|S|/2}}
\int_{\mathbb R^{|S|}}
e^{-\vct z^TA\vct z/2-\iu\vct\zeta_S^T\vct z_S}
\,d\vct z_S.
\]
Substitute
$\vct u=\vct z_S+A_{SS}^{-1}A_{ST}\vct z_T$.
It yields
\begin{equation}
    I_S(\vct\zeta_S,\vct z_T)
=e^{-\vct z_T^TE\vct z_T/2
   +\iu\vct\zeta_S^TA_{SS}^{-1}A_{ST}\vct z_T}\times
\frac{1}{(2\pi)^{|S|/2}}
\int_{\mathbb R^{|S|}}
e^{-\vct u^TA_{SS}\vct u/2-\iu\vct\zeta_S^T\vct u}
\,d\vct u.
\end{equation}

The Fourier integral of a real Gaussian gives
\begin{equation}
    \frac{1}{(2\pi)^{|S|/2}}
\int_{\mathbb R^{|S|}}
e^{-\vct u^TA_{SS}\vct u/2-\iu\vct\zeta_S^T\vct u}
\,d\vct u
=
\frac{
e^{-\vct\zeta_S^TA_{SS}^{-1}\vct\zeta_S/2}
}{
\sqrt{\det A_{SS}}
}.
\end{equation}
Therefore,
\begin{equation}
I_S(\vct\zeta_S,\vct z_T)
=
\frac{1}{\sqrt{\det A_{SS}}}
\exp\Bigl(
-\tfrac12\vct\zeta_S^TA_{SS}^{-1}\vct\zeta_S
+\iu\vct\zeta_S^TA_{SS}^{-1}A_{ST}\vct z_T
-\tfrac12\vct z_T^TE\vct z_T
\Bigr).
\end{equation}

To include the linear prefactor, differentiate with respect
to $\vct\zeta_S$. Gaussian decay justifies differentiation
under the integral sign, and
\[
\begin{aligned}
\frac{1}{(2\pi)^{|S|/2}}
\int_{\mathbb R^{|S|}}
\vct z_S
e^{-\vct z^TA\vct z/2-\iu\vct\zeta_S^T\vct z_S}
\,d\vct z_S
&=\iu\nabla_{\vct\zeta_S}I_S=
\bigl(
-\iu A_{SS}^{-1}\vct\zeta_S
-A_{SS}^{-1}A_{ST}\vct z_T
\bigr)I_S.
\end{aligned}
\]
It follows that
\begin{equation}
(\mathcal F_Sf_M)(\vct\zeta_S,\vct z_T)
=\alpha\Bigl[-\iu\vct\beta_S^TA_{SS}^{-1}\vct\zeta_S+
\bigl(
\vct\beta_T^T
-\vct\beta_S^TA_{SS}^{-1}A_{ST}
\bigr)\vct z_T
\Bigr]I_S.
\end{equation}

Writing $\vct w=(\vct\zeta_S,\vct z_T)$, we obtain
\[
(\mathcal F_Sf_M)(\vct w)
=
\alpha'((\vct\beta')^T\vct w)
e^{-\vct w^TA'\vct w/2},
\]
where
\[
\alpha'=\frac{\alpha}{\sqrt{\det A_{SS}}},
\qquad
\vct\beta'_S=-\iu A_{SS}^{-1}\vct\beta_S,
\qquad
\vct\beta'_T
=\vct\beta_T-A_{TS}A_{SS}^{-1}\vct\beta_S,
\]
and
\begin{equation}
A'=
\begin{pmatrix}
A_{SS}^{-1}
&
-\iu A_{SS}^{-1}A_{ST}
\\
-\iu A_{TS}A_{SS}^{-1}
&
A_{TT}-A_{TS}A_{SS}^{-1}A_{ST}
\end{pmatrix}.
\label{eq:red34-partial-fourier}
\end{equation}
Since $A_{SS}$ is invertible, $\vct\beta'=0$ would imply
$\vct\beta_S=0$ and then $\vct\beta_T=0$.
Thus $\vct\beta'\ne0$.
The matrix $A'$ is complex symmetric, and Fourier unitarity gives
\[
\norm{\mathcal F_Sf_M}
=\norm{f_M}=1.
\]

Third, we bound the real part of $A'$.
Since $A$ is real,
\[
\operatorname{Re}A'
=
\begin{pmatrix}
A_{SS}^{-1}&0\\
0&E
\end{pmatrix}.
\]
Restricting the quadratic-form bounds for $A$ to the
coordinates in $S$ gives
\begin{equation}
    \Id/64\preceq A_{SS}\preceq64\Id,\qquad \Id/64\preceq A_{SS}^{-1}\preceq64\Id.
\end{equation}
For the other block, the block inverse formula gives $E^{-1}=(A^{-1})_{TT}$.
The bounds on $A$ imply
\[
\Id/64\preceq A^{-1}\preceq64\Id.
\]
Restricting these bounds to its $TT$ principal block, and
then taking inverses, gives
\[
\Id/64\preceq E^{-1}\preceq64\Id,
\qquad
\Id/64\preceq E\preceq64\Id.
\]
Consequently,
\[
\Id/64\preceq\operatorname{Re}A'\preceq64\Id.
\]
Thus every partial Fourier transform of $f_M$ satisfies
Lemma~\ref{lem:core-lattice}.

\end{proof}

\bibliographystyle{alphaurl}
\bibliography{ref}
\end{document}